\documentclass[12pt]{article}
\usepackage{amsmath,amsfonts,amssymb,amsthm,amsbsy}
\usepackage[tight,footnotesize]{subfigure}
\usepackage{fullpage,graphicx}
\usepackage{psfrag}
\usepackage{lscape}
\usepackage[table]{xcolor}
\usepackage{footnote}
\usepackage{natbib}
\usepackage{color}
\usepackage{multirow}
\usepackage{algorithm}
\usepackage{algpseudocode}
\usepackage{caption}
\usepackage{epsfig}
\usepackage{url}
\usepackage[affil-it]{authblk}

\newtheorem{claim}{Claim}
\newtheorem{lemma}{Lemma}
\newtheorem{proposition}{Proposition}
\newtheorem{theorem}{Theorem}

\newcommand{\eps}{\epsilon}

\newcommand{\argmin} {\operatornamewithlimits{arg\min}}

\newcommand{\Blan}{\begin{landscape}}
\newcommand{\Elan}{\end{landscape}}

\begin{document}

\title{High-dimensional Fused Lasso Regression
using Majorization-Minimization  and Parallel Processing}

\date{}
\renewcommand\footnotemark{}

\author[1]{Donghyeon Yu\thanks{E-mail addresses : \texttt{dhyeon.yu@gmail.com} (D. Yu), \texttt{wonj@korea.ac.kr} (J. Won),\\ \texttt{taehoonlee@snu.ac.kr} (T. Lee), \texttt{johanlim@snu.ac.kr} (J. Lim), \texttt{sryoon@snu.ac.kr} (S. Yoon).}}
\affil[1]{Department of Statistics, Seoul National University}

\author[2]{Joong-Ho Won}
\affil[2]{School of Industrial Management Engineering, Korea University}

\author[3]{Taehoon Lee}
\author[1]{Johan Lim}
\author[3]{Sungroh Yoon}
\affil[3]{Department of Electrical and Computer Engineering, Seoul National University}

\maketitle
\begin{abstract}
\noindent 
In this paper, we propose a majorization-minimization (MM) algorithm for high-dimensional fused lasso regression (FLR) suitable for parallelization using graphics processing units (GPUs). The MM algorithm is stable and flexible as it can solve the FLR problems with various types of design matrices and penalty structures within a few tens of iterations. We also show that the convergence of the proposed algorithm is guaranteed. We conduct numerical studies to compare our algorithm with other existing algorithms, demonstrating that the proposed MM algorithm is competitive in many settings including the two-dimensional FLR with arbitrary design matrices. The merit of GPU parallelization is also exhibited.

\vskip0.5cm 
\noindent{\bf keywords:}  fused lasso regression; majorization-minimization
 algorithm; parallel computation; graphics processing unit.

\end{abstract}

\section{Introduction}
\label{sec:intro}

Regression methods using $\ell_1$-regularization are commonly 
employed to analyze
high-dimensional data, as they tend to yield sparse estimates
of regression coefficients. 
In this paper, we consider the fused lasso regression (FLR),
an important special case
of $\ell_1$-regularized regression.
The FLR minimizes
\begin{align} \label{eqn:obj}
  f(\beta) = \frac{1}{2} \big\| {\bf y } - {\bf X} \beta \big\|_2^2
   + \lambda_1 \sum_{j=1}^p | \beta_j | + \lambda_2 \sum_{(j,k) \in E} | \beta_{j} - \beta_{k}|,
\end{align}
where ${\bf y} \in \mathbb{R}^n$ is the response vector, and
${\bf X} \in \mathbb{R}^{n \times p}$ is the design matrix;
$n$ is the sample size, and $p$ is the number of variables; 
$\lambda_1$ and $\lambda_2$ are non-negative regularization parameters,
and $E$ is the unordered index set of the adjacent pairs of variables specified in the model.
Thus the FLR not only promotes sparsity among the coefficients, but also encourages 
adjacent coefficients to have the same values (``fusion''),
where adjacency is determined by the application.

The above specification of the FLR \eqref{eqn:obj} generalizes the original proposal by \citet{Tibshirani2005},
which considers adjacency on a one-dimensional chain graph, i.e., $E = \big\{(j-1,j)|~ j = 2,\ldots,p \big\}$.
We refer to this particular case as the \emph{standard FLR}. 
Problem \eqref{eqn:obj} also covers a wider class of the index set $E$, as considered by \cite{Chen2012} (see Section \ref{sec:review}).
Thus we call the penalty terms associated with $\lambda_2$ the \emph{generalized fusion penalty}.
For instance, the \emph{two-dimensional FLR} uses a two-dimensional lattice 
$E =
\big\{ \big( (i,j-1), (i,j) \big) |~ 1\le i \le q, j=2,3,\ldots,q \big\}
\cup \big\{ \big( (i-1,j), (i,j) \big) |~ i= 2,3,\ldots,q, 1 \le j \le q \big\}$
for coefficients $\big(\beta_{ij}\big)_{1\le i,j\le q}$;
the \emph{clustered lasso} \citep{She2010} and the \emph{pairwise fused lasso} \citep{Petry2011} use all pairs of the variables 
as the index set. 
Closely related but not exactly specified by \eqref{eqn:obj} includes the \emph{generalized lasso} \citep{TibshiraniRyan2011}.
Regardless of the penalty structure, or the choice of $E$, when ${\bf X}={\bf I}$ the FLR is called the \emph{fused lasso signal approximator} (FLSA).

The major challenge in the FLR as compared to
the classical lasso regression \citep{Tibshirani1996} that corresponds to setting $\lambda_2=0$ in \eqref{eqn:obj} is that the generalized fusion penalty terms are
non-separable as well as non-smooth.
The FLR is basically a quadratic programming (QP)
and can be solved in principle using general purpose QP solvers
 \citep{Gill1997,Grant2008}.
However, the standard QP formulation introduces a large number of
additional variables and constraints, and is not adequate for
general purpose QP solvers if the 
dimension $p$ increases.

To resolve this difficulty, many special algorithms
for solving the FLR problem have been proposed.
Some are based on LARS \citep{Efron2004}-flavored solution path methods
 \citep{Friedman2007, Hoefling2010, TibshiraniRyan2011},
and others on a fixed set of the regularization parameters
$(\lambda_1, \lambda_2)$ \citep{Liu2010,Ye2011,Lin2011,Chen2012}.
Although these algorithms are adequate for
high-dimensional settings, restrictions may apply:
some are only applicable for special design matrices \citep{Friedman2007,Hoefling2010,TibshiraniRyan2011}
or the standard FLR \citep{Liu2010}.
In addition, computational issues still pertain for very high dimensional settings ($p \gg 1000$).

In this paper, we apply the majorization-minimization (MM) algorithm \citep{Lange2000} to solve the FLR problem. The MM algorithm 
 iteratively finds and minimizes a surrogate function  
called the \emph{majorizing function} that is generally
convex and differentiable. 
It is well known that the efficiency of the MM algorithm depends 
largely on the choice of the majorizing function.
Hence we begin the main body of the paper by proposing
a majorizing function suitable for the FLR problem.

Motivated by \cite{HunterLi2005} on their work on the classical lasso,
we first introduce a perturbed version of the objective function
(\ref{eqn:obj})
and propose a majorizing function for this perturbed objective.
The MM algorithm we propose is based on these perturbed objective
function and its majorizing function. 
The MM algorithm has several advantageous points over
other existing algorithms.
First, it has flexibility in the choice of both
design matrix and penalty structure. 
We show that the proposed
MM algorithm converges to the optimal objective (\ref{eqn:obj})
regardless of the rank of the design matrix.
Furthermore, it can be applied to general
penalty structures, while having comparable performance with
state-of-the-art FLR solvers, some of which require a special penalty structure.
Second, as numerically demonstrated in Section 5, our MM algorithm
is stable in the number of iterations to converge
regardless of the choice of design matrix and the penalty structure.
Finally,  
an additional benefit of our MM formulation is that
it opens up the possibility of solving the FLR problem
in a massively parallel fashion using graphics processing units (GPUs).

This paper is organized as follows.
In Section 2, we review existing algorithms for the FLR.
  In Section 3, we propose an MM algorithm for the FLR
  with the generalized fusion penalty and prove the convergence of the  algorithm. 
   We also introduce the preconditioned conjugate gradient (PCG) method to accelerate the proposed MM algorithm.
In Section 4, we explain how to parallelize the MM algorithm using GPUs. 
In Section 5, 
we conduct numerical studies to compare the proposed MM algorithm
with the other existing algorithms.
In Section 6, we conclude the paper with a brief summary.

\section{Review of the existing algorithms}\label{sec:review}

Existing algorithms for solving the FLR problem can be
classified into two categories.
One is based on solution path methods, whose goal
is to find the whole solutions for all values of
the regularization parameters.
The other is based on first-order methods,
which attempt to solve a large scale problem
given a fixed set of regularization parameters 
using first order approximation.

\subsection{Solution path methods}\label{sec:review:path}

Solution path methods aim to provide all the solutions of interest
with a small computational cost after solving
the initial problems to find all change points 
while varying the regularization parameters $\lambda_1$ and $\lambda_2$.
A change point refers to the value of $\lambda_1$ or $\lambda_2$ at which some coefficients fuse, or coefficients that used to be fused split.

\paragraph{Path-wise optimization algorithm} ~\\
\cite{Friedman2007} propose the path-wise optimization algorithm for the ``standard'' FLSA,
\begin{equation} \nonumber
\min_\beta f(\beta) \equiv \frac{1}{2} \| {\bf y} - \beta \|_2^2  
+ \lambda_1 \sum_{j=1}^n | \beta_j | + \lambda_2 \sum_{j=2}^n 
| \beta_{j} - \beta_{j-1} |.
\end{equation}
In this case, it can be shown that $\widehat{\beta}_j(\lambda_1,\lambda_2)$, 
the optimal solution for the standard FLSA for regularization parameters $(\lambda_1,\lambda_2)$,
can be obtained from $\widehat{\beta}_j(0,\lambda_2)$ by soft-thresholding:
\begin{align}\label{eqn:soft}
\widehat{\beta}_j(\lambda_1,\lambda_2) 
= {\rm sign}\big(\widehat{\beta}_j(0,\lambda_2)\big)
\cdot \max\big( |\widehat{\beta}_j(0,\lambda_2)| - \lambda_1, 0 \big)
~~{\rm for}~ j=1,2,\ldots,p.
\end{align}
Hence we can set $\lambda_1 = 0$ without loss of generality and focus on the regularization path by varying only $\lambda_2$.

Because of the non-separable nature of the generalized fusion penalty terms, typical coordinate descent (CD) algorithms may fail to reach the optimal solution, despite the convexity of the objective function \citep{Tseng2001}. This is in contrast to the classical lasso, where the penalty terms (associated with $\lambda_1$ in the case of the FLR) are separable. 
Hence the path-wise optimization algorithm uses a modified version of the CD algorithm
so that two coefficients move together if the coordinate-wise
move fails to reduce the objective function.

An iteration of this modified CD algorithm comprises of three cycles: descent, fusion, and smoothing.
The descent cycle is to update the $i$th coordinate $\widehat{\beta}_i$ of the current solution $\widehat{\beta}$ in a usual coordinate descent fashion: 
minimize $f(\beta)$ with respect to $\beta_i$ holding all the other coefficients fixed.
Since the subdifferential of $f$ with respect to $\beta_i$
\begin{align*}
	 \frac{\partial f(\beta)}{\partial \beta_i} = -(y_i - \beta_i)
	 -\lambda_2 {\rm sgn}(\tilde{\beta}_{i+1} - \beta_i) +
	 \lambda_2 {\rm sgn} (\beta_i - \tilde{\beta}_{i-1}),
\end{align*}
where $\tilde{\beta}$ is the solution from the previous iteration; ${\rm sgn}(x)$ is the sign of $x$ if $x \neq 0$, and any number in $[-1, 1]$ if $x=0$,  
is piecewise linear with breaks at $\tilde{\beta}_{i-1}$ and $\tilde{\beta}_{i+1}$,
the desired coordinate-wise minimization is simple.
As discussed above, the descent cycle may not reduce the objective function $f(\beta)$, hence the fusion cycle minimizes $f(\beta)$ for two adjacent variables $\beta_{i-1}$ and $\beta_i$, under the constraint $\beta_{i-1}=\beta_i$, i.e., the two variables are fused. This cycle in effect reduces the number of the variables by one. 
Still, fusion of neighboring two variables is not enough to reduce the objective function. To overcome this drawback, the smoothing cycle increases $\lambda_2$ by a small amount and update the solution $\widehat\beta(0,\lambda_2)$ by repeatedly applying the descent and the fusion cycles, and keeping track of the variable reduction. This strategy guarantees convergence to the exact solution if it starts from $\lambda_2=0$ where $\widehat{\beta}={\bf y}$, because in the standard FLSA, a fusion occurs at most two neighboring variables if the increment of $\lambda_2$ is sufficiently small, and the fused variables do not split for the larger values of $\lambda_2$.

The resulting algorithm produces a forward stagewise-type path over a fine grid of $\lambda_2$, which can be obtained for all values of $\lambda_1$ due to 
\eqref{eqn:soft}. This algorithm has also been applied to the two-dimensional FLSA in \cite{Friedman2007}. However, its convergence is not guaranteed since the conditions described at the end of the last paragraph do not hold for this case. For the same reason, its extension to the FLR with general design matrices is limited.



\paragraph{Path algorithm for the FLSA} ~\\
\cite{Hoefling2010} proposes a path algorithm for the FLSA with the generalized fusion penalty,
\begin{align} \label{eqn:genflsa}
	  \min_\beta f(\beta) \equiv \frac{1}{2} \big\| {\bf y } - \beta \big\|_2^2 + \lambda_1 \sum_{i=1}^n |\beta_i|
	  +\lambda_2 \sum_{(j,k) \in E} |\beta_j - \beta_k|,
\end{align}
which gives an exact regularization path, in contrast to the approximate one for the standard FLSA by \cite{Friedman2007}.
The main idea of the path algorithm arises from the observation that the sets of fused coefficients do not change with $\lambda_2$ except for finitely many points.
(Again \eqref{eqn:soft} holds and we can assume $\lambda_1=0$ without loss of generality.) For each interval of $\lambda_2$ constructed by these finite change points, the objective function $f(\beta)$ in \eqref{eqn:genflsa} can be written as
\begin{align}\label{eqn:genflsa2}
	f^*(\beta) = \displaystyle \frac{1}{2} \sum_{i=1}^{n_F}
	   \Big(\sum_{j\in F_i(\lambda_2)}
	  \big(y_j - \beta_{F_i}(\lambda_2)\big)^2 \Big) 
	  \displaystyle +\lambda_2 \sum_{1\le i<j \le n_F} n_{ij}|\beta_{F_i}(\lambda_2) - \beta_{F_j}(\lambda_2)|,
\end{align}
where $F_1(\lambda_2), F_2(\lambda_2), \ldots, F_{n_F}(\lambda_2)$ are the sets of fused coefficients, $n_F=n_F(\lambda_2)$ is the number of the sets, and  $n_{ij} = |\{(k,l)\in E ~|~ k\in F_i(\lambda_2), l\in F_j(\lambda_2) \}|$ is the number of edges connecting the sets $F_i(\lambda_2)$ and $F_j(\lambda_2)$.
The minimizer of \eqref{eqn:genflsa2}, hence of \eqref{eqn:genflsa}, varies linearly with $\lambda_2$ within this interval:
\begin{align}\label{eqn:flsaslope}
\frac{\partial \widehat{\beta}_{F_i}(\lambda_2)}{\partial \lambda_2} = - \frac{\sum_{j\neq i}n_{ij}\mbox{sgn}\left(\widehat{\beta}_{F_i}(\lambda_2)-\widehat{\beta}_{F_j}(\lambda_2)\right)}{|F_i(\lambda_2)|},
\end{align}
i.e., the solution path is piecewise linear. Since at $\lambda_2 = 0$, $\widehat{\beta}_{F_i}(0) = y_i$ and $F_i = \{i\}$ for $i=1,2,\ldots,n$, by keeping track of the change points at which the solution path changes its slope \eqref{eqn:flsaslope} starting from $\lambda_2=0$, we can determine the entire solution path.

Regarding the change points, it can be shown that the optimality condition for \eqref{eqn:genflsa}
\[
\frac{\partial f}{\partial \beta_k}(\widehat{\beta}) = \widehat\beta_k - y_k + \lambda_2 \sum_{l: (k,l)\in E} \hat{t}_{kl}(\lambda_2) = 0,
\]
where $\hat{t}_{kl}(\lambda_2) = \mbox{sgn}\big( \widehat{\beta}_k(\lambda_2) - \widehat{\beta}_l(\lambda_2) \big)$,
is satisfied by  affine functions $\widehat{\beta}_k(\lambda_2)$ and $\hat{\tau}_{kl}(\lambda_2)=\lambda_2 \hat{t}_{kl}(\lambda_2)$ for an interval $[\lambda_2^0, \lambda_2^0+\Delta]$, with $\lambda_2^0$ a change point. In this case, we have
\begin{align}\label{eqn:affine}
\frac{\partial\widehat{\beta}_{F_i}}{\partial\lambda_2}(\lambda_2^0) + \sum_{l \in F_j(\lambda_2^0), j\neq i,(k,l)\in E} \hat{t}_{kl}(\lambda_2^0) 
+ \sum_{l \in F_i(\lambda_2^0), (k,l)\in E} \frac{\partial\hat{\tau}_{kl}}{\partial\lambda_2}(\lambda_2^0) = 0,
\end{align}
for $k \in F_i(\lambda_2^0)$, $i=1, 2, \ldots, n_F(\lambda_2^0)$.
At the beginning of the interval, $\partial\widehat{\beta}_{F_i}/\partial \lambda_2$ is given by \eqref{eqn:flsaslope}, and $\hat{t}_{kl}=\pm 1$ for $(k,l)\in E$, $l \in F_j(\lambda_2^0)$ with $i \neq j$. The length of the interval $\Delta$, hence the next change point, can be determined by examining where the sets of fused coefficients merge or split. Similar to the standard FLSA, at most two sets can merge, and a set splits into at most two sets for a sufficiently small increase in $\lambda_2$. 
The first merge after $\lambda_2^0$ may occur when the paths for any two sets of fused coefficients meet:
\begin{align*}
h(\lambda_2^0) = \min_{i=1,\ldots,n_F(\lambda_2^0)} \min_{j:h_{ij}(\lambda_2^0)>\lambda_2^0} h_{ij}(\lambda_2^0), 
\quad\mbox{where}\quad
h_{ij}(\lambda_2^0) = \frac{\widehat{\beta}_{F_i}(\lambda_2^0)
 - \widehat{\beta}_{F_j}(\lambda_2^0)}{\frac{\partial \widehat{\beta}_{F_j}(\lambda_2^0)}{\partial \lambda_2^0}
 - \frac{\partial \widehat{\beta}_{F_i}(\lambda_2^0)}{\partial \lambda_2^0}} + \lambda_2^0,
\end{align*}
unless a split occurs before this point. 
The first split may occur at the point where the condition \eqref{eqn:affine} is violated, i.e., 
\begin{align*}
v(\lambda_2^0) = \min_{i=1,\ldots,n_F(\lambda_2^0)}  \delta_i + \lambda_2^0,
\end{align*}
where $\delta_i$ can be found by solving the linear program (LP)
\begin{align}\label{eqn:flsalp}
\begin{array}{ll}
\mbox{minimize} & 1/\delta \\
\mbox{subject to} & \displaystyle
\sum_{l \in F_i(\lambda_2^0), (k,l)\in E} f_{kl} = p_k, \quad k \in F_i(\lambda_2^0), \\
~ & -1 - (1/\delta)(\lambda_2^0+\hat{\tau}_{kl}(\lambda_2^0)) \le f_{kl} \le 1 + (1/\delta)(\lambda_2^0-\hat{\tau}_{kl}(\lambda_2^0)), \quad k,l \in F_i(\lambda_2^0), (k,l) \in E,
\end{array}
\end{align}
in $1/\delta$ and $f_{kl}$, where $p_k = -\frac{\partial\widehat{\beta}_{F_i}}{\partial\lambda_2}(\lambda_2^0) - \sum_{l \in F_j(\lambda_2^0), j\neq i,(k,l)\in E} \hat{t}_{kl}(\lambda_2^0)$, for each $i=1,2,\ldots,n_F(\lambda_2^0)$.
This LP can be solved by iteratively applying a maximum flow algorithm, e.g., that of \citet{Ford1956}.
The maximum flow algorithm solves \eqref{eqn:flsalp} for fixed $\delta$, and the final solution $\hat{f}_{kl}$ to \eqref{eqn:flsalp} gives the value of $\frac{\partial\hat{\tau}_{kl}}{\partial\lambda_2}(\lambda_2^0)$ in \eqref{eqn:affine}.
Thus, $\Delta=\min\{h(\lambda_2^0),v(\lambda_2^0)\}-\lambda_2^0$. At split, \eqref{eqn:flsalp} also determines how the set $F_i$ is partitioned.

The computational complexity of this path algorithm depends almost on that of solving the maximum flow problem corresponding to \eqref{eqn:flsalp}. 
It becomes inefficient as the dimension $p$ increases because a large scale maximum flow problem is difficult to solve.
However, when restricted to the standard FLSA, this path algorithm is very efficient, since the solution path of the standard FLSA has only fusions
and not splits, hence does not require solving the maximum flow problem.

\paragraph{Path algorithm for the generalized lasso} ~\\
\cite{TibshiraniRyan2011}
propose  a path algorithm for the generalized lasso problem,
which replaces the penalty terms in (\ref{eqn:obj}) with
a generalized $\ell_1$-norm penalty $\lambda \| {\bf D} \beta \|_1$:
\begin{align} \label{eqn:genlasso}
\min_{\beta} f(\beta) \equiv \frac{1}{2} \big\| {\bf y} -{\bf X} \beta \big\|_2^2 + \lambda \big\| {\bf D}\beta \big\|_1,
\end{align}
where $\lambda$ is a non-negative regularization parameter and ${\bf D}$ is a
 $m \times p$ dimensional matrix corresponding to the dependent structure of  coefficients.
The generalized lasso has the FLR with the generalized fusion penalty as a special case. 
The dual problem of \eqref{eqn:genlasso} is given by
\begin{align} \label{eqn:gen_gen}
\begin{array}{ll} 
\mbox{minimize} & \displaystyle \frac{1}{2} \big\|\tilde{\bf y} -\tilde{\bf D}^T{\bf  u} \big\|_2^2\\  
{\rm subject ~ to~} & \displaystyle \| {\bf u} \|_\infty \le \lambda,~ {\bf D}^T {\bf u} \in {\rm row}({\bf X})
\end{array}
\end{align}
on ${\bf u} \in {\mathbb R}^m$,
where $\|{\bf u}\|_\infty = \max_j |u_j|$, $\tilde{\bf y} = {\bf X} ({\bf X}^T {\bf X})^{\dagger} {\bf X}^T {\bf y}$, $\tilde{\bf D} = {\bf D} ({\bf X}^T {\bf X})^{\dagger} {\bf X}^T$, 
and ${\rm row}({\bf X})$ is a row space of ${\bf X}$; $A^{\dagger}$ denotes the Moore-Penrose pseudo-inverse of the matrix $A$.
This dual problem is difficult to solve due to the
the constraint ${\bf D}^T {\bf u} \in {\rm row}({\bf X})$, 
but this row space constraint can be
removed when ${\rm rank}({\bf X}) = p$.
Thus the path algorithm focuses on the case 
when the design matrix $\bf X$ has a full rank.

The algorithm starts from $\lambda=\infty$ and decreases $\lambda$. Suppose $\lambda_k$ is the current value of the regularization parameter $\lambda$. 
The current dual solution $\widehat{\bf u}(\lambda_k)$ is given by
\begin{align*}
\displaystyle
\widehat{\bf u}_{\mathcal{B}}(\lambda_k) &= \lambda_k {\bf s},\\  \displaystyle
\widehat{\bf u}_{-\mathcal{B}}(\lambda_k) &= \big({\bf D}_{-\mathcal{B}}({\bf D}_{-\mathcal{B}})^T
\big)^{\dagger} {\bf D}_{-\mathcal{B}} \big( {\bf y}- \lambda_k ({\bf D}_{\mathcal{B}})^T {\bf s} \big),
\end{align*}
where $\mathcal{B} = \mathcal{B}(\lambda) = \{~ j ~\big| ~|u_j| = \lambda, j=1,2,\ldots,m\}$ is the active set of the constraint $\| {\bf u} \|_\infty \le \lambda$;
$-\mathcal{B}=\mathcal{B}^c$; 
for an index set $A \subset \{1,2,\ldots,m\}$ and for a vector ${\bf v}=(v_i)_{1 \le i \le m} \in \mathbb{R}^m$, 
${\bf v}_A$ denotes the $|A|$-dimensional vector such that
${\bf v}_A=(v_i)_{i\in A}$,
and for a matrix ${\bf D} = ({\bf d}_i)_{1 \le i \le m} \in \mathbb{R}^{p \times m}$ with ${\bf d}_i \in \mathbb{R}^p$, 
${\bf D}_A$ denotes the $p \times |A|$
dimensional matrix such that ${\bf D}_A=({\bf d}_i)_{i\in A}$;
and ${\bf s} = (s_1,\ldots,s_{|\mathcal{B}|})^T$ with $s_k = \mbox{sign}(u_{n_k}(\lambda_k))$, $n_k \in \mathcal{B}$.
The primal solution is obtained from the dual solution by the following primal-dual relationship
\begin{align}\label{eqn:genlassoprimaldual}
\widehat{\beta}(\lambda) = \tilde{\bf y} - \tilde{\bf D}^T \widehat{\bf u}(\lambda),
\end{align}
hence the solution path is piecewise linear. The intervals in which all $\widehat{\beta}_j(\lambda)$ have constant slopes can be found by keeping track of the active set $\mathcal{B}$.
An inactive coordinate $j \in -\mathcal{B}$ at $\lambda=\lambda_k$, possibly joins $\mathcal{B}$ at the hitting time
\begin{align*}
t_j^{\rm (hit)} =
\frac{ \Big[ \big( {\bf D}_{-\mathcal{B}} ({\bf D}_{-\mathcal{B}})^T\big)^{\dagger}
{\bf D}_{-\mathcal{B}} {\bf y} \Big]_j }{
\Big[ \big( {\bf D}_{-\mathcal{B}} ({\bf D}_{-\mathcal{B}})^T\big)^{\dagger}
{\bf D}_{-\mathcal{B}} ({\bf D}_{\mathcal{B}})^T {\bf s} \Big]_j \pm 1},
\end{align*}
for which only one is less than $\lambda_k$, 
and an active coordinate $j' \in \mathcal{B}$, possibly leaves $\mathcal{B}$ at the leaving time
\begin{align*}
t_{j'}^{\rm (leave)} = 
\begin{cases}
\displaystyle \frac{ s_{j'} \Big[ {\bf D}_{\mathcal{B}} \big[ {\bf I} - {\bf D}_{-\mathcal{B}}^T
 \big( {\bf D}_{-\mathcal{B}} {\bf D}_{-\mathcal{B}}^T\big)^{\dagger}
{\bf D}_{-\mathcal{B}} \big] {\bf y} \Big]_{j'} }{
s_{j'} \Big[ {\bf D}_{\mathcal{B}} \big[ {\bf I} - {\bf D}_{-\mathcal{B}}^T
 \big( {\bf D}_{-\mathcal{B}} {\bf D}_{-\mathcal{B}}^T\big)^{\dagger}
{\bf D}_{-\mathcal{B}} \big] {\bf D}_{\mathcal{B}}^T {\bf s} \Big]_{j'} }\equiv \frac{c_{j'}}{d_{j'}},
 & \mbox{if} ~ c_{j'},d_{j'}<0,\\ 
 0, & \mbox{otherwise}.
\end{cases}
\end{align*}
Thus $\mathcal{B}$ changes at $\lambda=\lambda_{k+1}$ where
\[
\lambda_{k+1} = \max\{
\max_{j \in -\mathcal{B}}~ t_j^{\rm{(hit)}},
\max_{j \in \mathcal{B}}~ t_j^{\rm{(leave)}}
\}.
\]

The path algorithm for the generalized lasso problem can solve various $\ell_1$-regularization problems and obtains the exact solution path. 
After sequentially solving the dual, the solution path is obtained by the
primal-dual relationship \eqref{eqn:genlassoprimaldual}.
It is also efficient for solving the standard FLSA.
However, 
if the rank of the design matrix ${\bf X}$ is not full rank,
the additional constraint ${\bf D}^T {\bf u} \in {\rm row}({\bf X})$ does not disappear, 
making the dual \eqref{eqn:gen_gen} difficult to solve.
To resolve this problem, \cite{TibshiraniRyan2011} suggest adding an additional penalty $\epsilon \|\beta\|_2^2$ to the primal \eqref{eqn:genlasso},
which in effect substitutes ${\bf y}$ and ${\bf X} $ with ${\bf y}^* = ({\bf y}^T, {\bf 0}^T)^T$ and
 ${\bf X}^* = \left( {\bf X}^T ,\sqrt{\epsilon}~ {\bf I} \right)^T$, respectively.
This modification makes the new design matrix ${\bf X}^*$ full rank, but 
a small value of $\epsilon$ may lead an ill-conditioned matrix $({\bf X}^*)^T {\bf X}^*$.
A further difficulty is that as the number of rows $m$ of the matrix $\bf D$, i.e., the number of penalty terms on $\beta$, increases, the path algorithm becomes less efficient since $m$ is the number of dual variables.

\subsection{First-order methods}

To avoid the restrictions in the solution path methods,
several optimization algorithms,
which target to solve the FLR problem with a fixed set of
regularization parameters, have been developed.
These algorithms 
can solve the FLR problems with the general design matrix ${\bf X}$
regardless of its rank.
For scalability with the dimension $p$,
they employ gradient descent-flavored first-order
optimization methods.

\paragraph{Efficient fused lasso algorithm} ~\\
\cite{Liu2010} propose
 the efficient fused lasso algorithm (EFLA) that
 solves the standard FLR
\begin{align} \label{eqn:org_obj}
 \min_\beta f(\beta)
   &\equiv  \displaystyle  \frac{1}{2} \big\| {\bf y } - {\bf X} \beta \big\|_2^2
   + \lambda_1 \sum_{j=1}^p | \beta_j | + \lambda_2 \sum_{j=2}^p | \beta_{j} - \beta_{j-1}|.
\end{align}
At the $(r+1)$th iteration, the EFLA solves \eqref{eqn:org_obj} by minimizing an approximation of $f(\beta)$ in which the squared error term is replaced by its first-order Taylor expansion at the approximate solution $\widehat{\beta}^{(r)}$ obtained from the $r$th iteration, followed by an additional quadratic regularization term $\frac{L_r}{2}\| \beta - \widehat{\beta}^{(r)}\|_2^2$, $L_r>0$.
Minimization of this approximate objective can be written as
\begin{align}\label{eqn:EFLA}
\min_\beta \frac{1}{2} \| {\bf v}^{(r)} - \beta\|_2^2   +\frac{\lambda_1}{L_r} \sum_{j=1}^p | \beta_j |+ \frac{\lambda_2}{L_r} \sum_{j=2}^p | \beta_j  - \beta_{j-1}|,
\end{align}
where ${\bf v}^{(r)} = \widehat{\beta}^{(r)} -({\bf X}^{T} {\bf X} \widehat{\beta}^{(r)} + {\bf X}^T {\bf y})/{L_r} $. The minimizer of \eqref{eqn:EFLA} is denoted by  $\widehat{\beta}^{(r+1)}$.
Note that this problem is the standard FLSA. 

Although \eqref{eqn:EFLA} can be solved by the path algorithms reviewed in Section \ref{sec:review}, for fixed $\lambda_1$ and $\lambda_2$, more efficient methods can be employed. \cite{Liu2010} advocate the use of a gradient descent method on the dual of \eqref{eqn:EFLA}, given by
\begin{align} \label{eqn:sfa_obj}
\min_{\|{\bf u}\|_\infty \le \lambda_2} \frac{1}{2}  {\bf u}^T{\bf D}{\bf D}^T {\bf u} - {\bf u}^T {\bf D} {\bf v}^{(r)},
\end{align}
where ${\bf D} \in {\mathbb R}^{(p-1) \times p}$ is the finite different matrix on the one-dimensional grid of size $p$, a special case of the $m \times p$ penalty matrix of the generalized lasso \eqref{eqn:genlasso}.
Note here that we set $\lambda_1=0$ without loss of generality, due to the relation \eqref{eqn:soft}.
Solution to \eqref{eqn:EFLA} is obtained from the primal-dual relationship $\widehat{\beta}^{(r+1)} = {\bf v}^{(r)} - {\rm D}^T {\bf u}^{\star}$, where ${\bf u}^{\star}$ is the solution to \eqref{eqn:sfa_obj}. 
Since \eqref{eqn:sfa_obj} is a box-constrained quadratic program, it is solved efficiently by the following projected gradient method:
\begin{align*}
	\widehat{\bf u}^{(k+1)} =  P_{\lambda_2}(\widehat{\bf u}^{(k)} - \alpha ({\bf D} {\bf D}^T\widehat{\bf u}^{(k)} - {\bf D}{\bf v}^{(r)}) ),
\end{align*} 
where $P_{\lambda_2}(\cdot)$ is the projection onto the $l_{\infty}$-ball $\mathcal{B}=\{{\bf u}:\|{\bf u}\|_\infty \le \lambda_2 \}$ with radius $\lambda_2$, and $\alpha$ is the reciprocal of the largest eigenvalue of the matrix $\bf{D}\bf{D}^T$. 
Further acceleration is achieved by using a restart technique that keeps track of the coordinates of $\bf u$ for which the box constraints are active, and by replacing $\widehat{\beta}^{(r)}$ in \eqref{eqn:org_obj} with $\bar{\beta}^{(r)} = \widehat{\beta}^{(r)} + \eta_r (\widehat{\beta}^{(r)} - \widehat{\beta}^{(r-1)})$. The acceleration constant $\eta_r$ and the additional regularization constant $L_r$ in \eqref{eqn:EFLA} are chosen using Nestrov's method\citep{Nesterov2003,Nesterov2007}. 

The efficiency of the EFLA strongly depends on the special structure that the finite difference matrix $\bf D$ takes. Hence its generalizability to the generalized fusion penalty other than the standard one is limited. One may want to use one of the path algorithms for FLSA shown in Section \ref{sec:review} instead of solving the dual \eqref{eqn:sfa_obj}. However, this modification does not assure the efficiency of the algorithm, because a path algorithm always starts from $\lambda_2=0$ \citep{Hoefling2010} or $\lambda_2=\infty$ \citep{TibshiraniRyan2011} while for EFLA a fixed value of $\lambda_2$ suffices.

\paragraph{Smoothing proximal gradient method} ~\\
\cite{Chen2012} propose the smoothing proximal gradient (SPG) method
 that solves the regression problems with structured penalties 
 including the generalized fusion penalty.
Recall the FLR problem \eqref{eqn:obj} can be written
\begin{align}\label{eqn:spg}
\displaystyle  \min_\beta f(\beta) &\equiv \displaystyle \frac{1}{2} \big\| {\bf y } - {\bf X} \beta \big\|_2^2
   + \lambda_2 \sum_{(j,k) \in E} | \beta_{j} - \beta_{k}| + \lambda_1 \sum_{j=1}^p | \beta_j | \nonumber \\ 
   & \equiv   \displaystyle \frac{1}{2} \big\| {\bf y } - {\bf X} \beta \big\|_2^2 + \lambda_2 \|{\bf D}\beta\|_1 + \lambda_1 \| \beta \|_1, 
\end{align}
where ${\bf D}$ is an $m \times p$ dimensional matrix such that
$\|{\bf D}\beta\|_1 = \sum_{(j,k)\in E} |\beta_j - \beta_k|$.
The main idea of the SPG method is to approximate the non-smooth and non-separable penalty term $\| {\bf D} \beta \|_1$ by a smooth function and to solve the resulting smooth surrogate problem iteratively in a computationally efficient fashion.

For the smooth approximation, recall that the penalty term $\|{\bf D} \beta \|_1$  can be reformulated as a maximization problem with an auxiliary vector $\alpha \in {\mathbb R}^{m}$,
\begin{align*}
   	\|{\bf D} \beta \|_1 \equiv \max_{\|\alpha\|_\infty \le 1} \alpha^T {\bf D} \beta \equiv \Omega(\beta).
\end{align*}
Now consider a smooth function 
\begin{align}\label{eqn:norm}
\tilde{\Omega}_{\mu}(\beta) \equiv \max_{\|\alpha\|_\infty \le 1} \big(
\alpha^T {\bf D} \beta - \frac{\mu}{2} \| \alpha\|_2^2 \big),
\end{align}   
that approximates $\Omega(\beta)$ with a positive smoothing parameter $\mu$, chosen as $\epsilon/m$ to obtain the desired accuracy $\epsilon>0$ of the approximation. The function $\tilde{\Omega}_{\mu}(\beta)$ is convex and continuously differentiable with respect to $\beta$ \citep{Nesterov2005}.

For the smooth surrogate problem, the SPG method solves the following.
\begin{align}\label{eqn:spgsurrogate}
\displaystyle \min_\beta \tilde{f}(\beta) 
   & \equiv   \displaystyle \underbrace{\frac{1}{2}\| {\bf y} - {\bf X}\beta \|_2^2 + \lambda_2 \tilde{\Omega}(\beta,\mu)}_{h(\beta)} + \lambda_1 \| \beta \|_1.
\end{align}
This is the sum of a smooth convex function $h(\beta)$ and a non-smooth but \emph{separable} function $\lambda_1 \| \beta \|_1$, hence can be solved efficiently using a proximal gradient method, e.g., the fast iterative shrinkage-thresholding algorithm (FISTA, \citet{Beck2009}). The FISTA iteratively approximates $h(\beta)$ by its first-order Taylor expansion with an additional quadratic regularization term of the form $\frac{L}{2}\| \beta - \widehat{\beta}^{(r)} \|_2^2$, with essentially the same manner as the EFLA that yields \eqref{eqn:EFLA}. As a result, the FISTA solves at each iteration a subproblem that minimizes the sum of a quadratic function of $\beta$ whose Hessian is the identity, and the classical lasso penalty $\lambda_1 \| \beta \|_1$. The solution of this subproblem is readily given by soft-thresholding.

The computational complexity of the SPG method depends on that of determining 
the constant $L$ for the additional quadratic regularization term, which the FISTA chooses as the Lipschitz constant of the gradient of $h(\beta)$, i.e., 
$L = \|{\bf X}\|_2^2 +\frac{ \lambda_2 }{\mu}\|{\bf D}\|_2^2$,
where $\| {\bf A}\|_2$ is the spectral norm of a matrix ${\bf A}$. The computational complexity for computing $L$ is $O\big( \min(m^2p,mp^2) \big)$, which is costly if either $m$ or $p$ increases. 
To reduce this computational cost,
\cite{Chen2012} suggest using a line search on $L$, only to ensure that at the minimum of objective of the FISTA subproblem is not less than the surrogate function $\tilde{f}(\beta)$ evaluated at the minimizer of the same problem. However, since such an evaluation of $\tilde{f}(\beta)$
takes $O\big(\max(mp,np) \big)$ time, this scheme does not save much cost if the surrogate function is evaluated frequently, especially when either $m$ or $p$ is large.

\paragraph{Split Bregman algorithm} ~\\
\cite{Ye2011} propose the split Bregman (SB) algorithm for the standard FLR and it can directly extend to the general FLR \eqref{eqn:obj} and \eqref{eqn:spg}. Note that \eqref{eqn:spg} can be reformulated as a constrained form:
\begin{align}\label{eqn:sb}
\begin{array}{ll}
\min_{\beta, {\bf a}, {\bf b}} & \displaystyle \frac{1}{2} \big\| {\bf y } - {\bf X} \beta \big\|_2^2 + \lambda_1 \| {\bf a} \|_1 + \lambda_2 \|{\bf b} \|_1 \\
\mbox{subject~to} & {\bf a} = \beta \\
~ & {\bf b} = {\bf D}\beta,
\end{array}
\end{align}
with auxiliary variables ${\bf a} \in {\mathbb R}^p$ and $ {\bf b} \in {\mathbb R}^m$.
The SB algorithm is derived from the augmented Lagrangian \citep{Hestenes1969,Rockafellar1973} of \eqref{eqn:sb} given by
\begin{align}\label{eqn:aug}
\displaystyle L(\beta,{\bf a},{\bf b},{\bf u},{\bf v}) &= \displaystyle \frac{1}{2} \big\| {\bf y } - {\bf X} \beta \big\|_2^2 + \lambda_1 \|{\bf a}\|_1 + \lambda_2 \|{\bf b}\|_1 
+ {\bf u}^T(\beta - {\bf a})
+ {\bf v}^T({\bf D}\beta - {\bf b}) \notag\\
 & \qquad \displaystyle  + \frac{\mu_1}{2} \| \beta - {\bf a} \|_2^2 + \frac{\mu_2}{2} \| {\bf D}\beta -{\bf b} \|_2^2,
\end{align}
where ${\bf u} \in {\mathbb R}^p$, ${\bf v} \in {\mathbb R}^m$ are dual variables, and $\mu_1$, $\mu_2$ are positive constants. This is the usual Lagrangian of \eqref{eqn:sb} augmented by adding the quadratic penalty terms $\frac{\mu_1}{2}\| \beta - {\bf a} \|_2^2$ and $\frac{\mu_2}{2}\| {\bf D}\beta -{\bf b} \|_2^2$, for violating the equality constraints in \eqref{eqn:sb}. The primal problem associated with the augmented Lagrangian \eqref{eqn:aug} is to minimize $\sup_{{\bf u},{\bf v}}L(\beta,{\bf a},{\bf b},{\bf u},{\bf v})$ over $(\beta,{\bf a},{\bf b})$, and the dual is to maximize $\inf_{\beta,{\bf a},{\bf b}}L(\beta,{\bf a},{\bf b},{\bf u},{\bf v})$ over $({\bf u},{\bf v})$.
%

The SB algorithm solves the primal and the dual in an alternating fashion, in which the primal is solved separately for each of $\beta$, $\bf a$, and $\bf b$:
\begin{align}
\widehat{\beta}^{(r)} &= \displaystyle \argmin_{\beta}~\displaystyle \frac{1}{2} \big\| {\bf y } - {\bf X} \beta \big\|_2^2 + ({\bf u}^{(r-1)})^{T}(\beta - {\bf a}^{(r-1)})
 \displaystyle + ({\bf v}^{(r-1)})^T({\bf D}\beta - {\bf b}^{(r-1)}) \nonumber\\ 
& \qquad \qquad \quad \displaystyle+ \frac{\mu_1}{2} \big\|\beta-{\bf a}^{(r-1)}\big\|_2^2
+ \frac{\mu_2}{2} \big\|{\bf D}\beta-{\bf b}^{(r-1)}\big\|_2^2, \label{eqn:sbprimal1}\\
{\bf a}^{(r)} &= \displaystyle \argmin_{\bf a}~\displaystyle \lambda_1 \|{\bf a}\|_1 + (\widehat{\beta}^{(r)} -{\bf a})^T{\bf u}^{(r-1)}
 \displaystyle + \frac{\mu_1}{2} \big\|\widehat{\beta}^{(r)}-{\bf a}\big\|_2^2,\label{eqn:sbprimal2}\\
 {\bf b}^{(r)} &= \displaystyle \argmin_{\bf b}~\displaystyle \lambda_1 \|{\bf b}\|_1 + ({\bf D}\widehat{\beta}^{(r)} -{\bf b})^T{\bf v}^{(r-1)} \displaystyle + \frac{\mu_2}{2} \big\|{\bf D}\widehat{\beta}^{(r)}-{\bf b}\big\|_2^2, \label{eqn:sbprimal3}
\end{align}
while the dual is solved by gradient ascent:
\begin{align}
{\bf u}^{(r)} &= {\bf u}^{(r-1)} + \delta_1 (\widehat{\beta}^{(r)} - {\bf a}^{(r)}), \label{eqn:sbdual1}\\
{\bf v}^{(r)} &= {\bf v}^{(r-1)} + \delta_2 ({\bf D}\widehat{\beta}^{(r)} - {\bf b}^{(r)}). \label{eqn:sbdual2}
\end{align}
where $\delta_1$ and $\delta_2$ are the step sizes.

The dual updates \eqref{eqn:sbdual1} and \eqref{eqn:sbdual2} are simple; the primal updates for the auxiliary variables \eqref{eqn:sbprimal2} and \eqref{eqn:sbprimal3} are readily given by soft-thresholding, in a similar manner as the SPG case \eqref{eqn:spgsurrogate}. The primal update for the coefficient $\beta$ \eqref{eqn:sbprimal1} is equivalent to solving the linear system
\begin{align} \label{eqn:sb_lin}
({\bf X}^T {\bf X} + \mu_1 {\bf I} + \mu_2 {\bf D}^T {\bf D}) \beta = 
{\bf X}^T {\bf y} + (\mu_1 {\bf a}^{(r-1)} -  {\bf u}^{(r-1)})
+ {\bf D}^T (\mu_2 {\bf b}^{(r-1)} - {\bf v}^{(r-1)})
\end{align}   
whose computational complexity is $O(p^3)$ for general $\bf D$.
However, for standard FLR in which $\bf D$ is the finite difference matrix on a one-dimensional grid, 
\eqref{eqn:sb_lin} can be solved quickly by using the preconditioned conjugate
gradient (PCG) algorithm \citep{Demmel1997}. 
In this case, ${\bf M}=\mu_1 {\bf I} + \mu_2 {\bf D}^T {\bf D}$
is used as a preconditioner because $\bf M$ is tridiagonal and easy to invert.

Note that the separation of updating equations \eqref{eqn:sbprimal1}, \eqref{eqn:sbprimal2} and \eqref{eqn:sbprimal3} for the primal is possible
because the non-differentiable and non-separable $\ell_1$-norm penalties on the coefficients are transferred to the auxiliary variables that are completely decoupled and separable. Also note that the updating equations for the auxiliary variables \eqref{eqn:sbprimal2} and \eqref{eqn:sbprimal3} reduce to soft-thresholding due to the augmented quadratic terms and the separable non-differentiable terms.

The SB algorithm needs to choose the augmentation constants $\mu_1$ and $\mu_2$ and the step sizes $\delta_1$ and $\delta_2$.
In the implementation, \citet{Ye2011} use
a common value $\mu$ for all of the four quantities, and consider a pre-trial procedure
to obtain highest convergence rate.
Although the choice of $\mu$ does not affect the convergence of the algorithm, 
it is known that the rate of convergence is very sensitive to this choice \citep{Lin2011,Ghadimi2012}.
This makes the SB algorithm to stall or fail to reach the optimal solution; see Appendix B for further details.

\paragraph{Alternating linearization} ~\\  
\citet{Lin2011} propose the alternating linearization (ALIN) algorithm for the generalized lasso \eqref{eqn:genlasso}. Rewrite the problem \eqref{eqn:genlasso} as a sum of two convex functions
\begin{align} \label{eqn:sum}
\min_{\beta} f(\beta) \equiv \underbrace{\frac{1}{2} \big\| {\bf y} -{\bf X} \beta \big\|_2^2}_{g(\beta)} + \underbrace{\lambda \big\| {\bf D}\beta \big\|_1}_{h(\beta)}.
\end{align}
The main idea of the ALIN algorithm is to alternately solve two subproblems, in which $g(\beta)$ is linearized ($h$-subproblem) and added by a quadratic regularization term, and so is $h(\beta)$ ($g$-subproblem), respectively. 
This algorithm maintains three sequences of solutions: $\{\widehat{\beta}^{(r)}\}$ is the sequence of solutions of the problem (\ref{eqn:sum}); $\{\tilde{\beta}_h^{(r)} \}$ and $\{\tilde{\beta}_g^{(r)} \}$ are the sequences of solutions of the $h$-subproblem and the $g$-subproblem, respectively.

For $r$th iteration of subproblems,
the $h$-subproblem is given by
\begin{align}\label{eqn:hsub1}
\tilde{\beta}_h^{(r)} = \argmin_{\beta} \underbrace{\frac{1}{2} \| {\bf y} -{\bf X} \tilde{\beta}_g^{(r-1)} \|_2^2 + {\bf s}_g^T (\beta - \tilde{\beta}_g^{(r-1)})}_{\tilde{g}(\beta)}
	+ \frac{1}{2}  \| \beta - \widehat{\beta}^{(r-1)}  \|_{\bf A}^2	
	+ \lambda  \|{\bf D}\beta \|_1
\end{align}
where $\tilde{g}(\beta)$ is a linearization of $g(\beta)$ at $\tilde{\beta}_g^{(r-1)}$,
 ${\bf s}_g = \nabla g(\tilde{\beta}_g^{(r-1)}) = {\bf X}^T({\bf X} \tilde{\beta}_g^{(r-1)} - {\bf y})$, and
$\| \beta - \widehat{\beta} \|_{\bf A}^2 = (\beta - \widehat{\beta})^{T} {\bf A} (\beta - \widehat{\beta})$ with a diagonal matrix ${\bf A} = \mbox{diag}\big({\bf X}^T {\bf X}\big)$. Problem \eqref{eqn:hsub1} is equivalent to
\begin{align}\label{eqn:hsub2}
\min_{\beta,{\bf z}}~ {\bf s}_g^T \beta + \frac{1}{2}\| \beta - \widehat{\beta}^{(r-1)} \|_{\bf A}^2  
+ \lambda \|{\bf z}\|_1 \quad {\rm subject~to} \quad {\bf D}\beta = {\bf z}.
\end{align}
Similar to \eqref{eqn:gen_gen}, the dual of \eqref{eqn:hsub2} is given by
\begin{align} \label{eqn:alin}
\displaystyle \min_{{\bf u}} \frac{1}{2} {\bf u}^T {\bf D} {\bf A}^{-1} {\bf D}^T {\bf u} 
- {\bf u}^T{\bf D}(\widehat{\beta}^{(r-1)} - {\bf A}^{-1} {\bf s}_g) \quad
\displaystyle {\rm subject~ to} \quad
\| {\bf u} \|_\infty \le \lambda.
\end{align}
The solution of $h$-subproblem is obtained by the primal-dual relationship
 $\tilde{\beta}_h^{(r)} = \widehat{\beta}^{(r-1)} - {\bf A}^{-1}({\bf s}_g + {\bf D}^T {\bf u}^*)$, where ${\bf u}^*$ is the optimal solution of \eqref{eqn:alin}. The dual \eqref{eqn:alin} is efficiently solved by an active-set box-constrained PCG algorithm with $\mbox{diag}({\bf D} {\bf A}^{-1} {\bf D}^T)$ as a preconditioner.

\noindent The $g$-subproblem is also given by
\begin{align}\label{eqn:gsub1}
\tilde{\beta}_g^{(r)} &= \argmin_{\beta} \frac{1}{2} \big\| {\bf y} -{\bf X}\beta \big\|_2^2 
+ \underbrace{\lambda \|{\bf D}\tilde{\beta}_h^{(r)} \|_1
+ {\bf s}_h^T (\beta - \tilde{\beta}_h^{(r)})}_{\tilde{h}(\beta)}
+ \frac{1}{2} \| \beta - \widehat{\beta}^{(r-1)} \|_{\bf A}^2	
\nonumber \\
&=	
\argmin_{\beta} {\bf s}_h^T \beta + \frac{1}{2} \| {\bf y} - {\bf X}\beta \|_2^2
+ \frac{1}{2} \| \beta - \widehat{\beta}^{(r-1)} \|_{\bf A}^2,
\end{align}
where ${\bf s}_h$ is the subgradient of $h(\beta)$ at $\tilde{\beta}_h^{(r)}$ that can be calculated as ${\bf s}_h = -{\bf s}_g - {\bf A}(\tilde{\beta}_h^{(r)} - \widehat{\beta}^{(r-1)})$. 
Problem \eqref{eqn:gsub1} is equivalent to solving the following linear system.
\begin{equation} \label{eqn:alin_fsub}
({\bf X}^T {\bf X} + {\bf A})(\beta - \widehat{\beta}^{(r-1)}) = {\bf X}^T({\bf y}-{\bf X}\widehat{\beta}^{(r-1)}) - {\bf s}_h,
\end{equation}
which can be solved by using the PCG algorithm with ${\bf A} = \mbox{diag}({\bf X}^T {\bf X})$ as a preconditioner.


After solving each subproblem, the ALIN algorithm checks its stopping and updating criteria.
The stopping criteria for the $h$-subproblem and the $g$-subproblem
are defined as
\begin{align*}
\begin{array}{l}
h\mbox{-sub} ~:~
\tilde{g}(\tilde{\beta}_h^{(r)}) + {h}(\tilde{\beta}_h^{(r)}) \ge g(\widehat{\beta}^{(r-1)}) +
h(\widehat{\beta}^{(r-1)}) - \eps,\\
g\mbox{-sub} ~:~
g(\tilde{\beta}_g^{(r)}) + \tilde{h}(\tilde{\beta}_g^{(r)}) \ge g(\widehat{\beta}^{(r-1)}) +
h(\widehat{\beta}^{(r-1)}) - \eps,
\end{array}
\end{align*}
where $\eps$ is a tolerance of the algorithm.
If one of the stopping criteria is met, then the algorithm terminates; otherwise it proceeds to  check the updating criteria:
\begin{align*}
\begin{array}{l}
h\mbox{-sub} ~:~ 
g(\tilde{\beta}_h^{(r)}) + {h}(\tilde{\beta}_h^{(r)}) \le (1-\gamma)[g(\widehat{\beta}^{(r-1)}) +
h(\widehat{\beta}^{(r-1)})] + \gamma [\tilde{g}(\tilde{\beta}_h^{(r)}) +
h(\tilde{\beta}_h^{(r)})],\\
g\mbox{-sub} ~:~
g(\tilde{\beta}_g^{(r)}) + {h}(\tilde{\beta}_g^{(r)}) \le (1-\gamma)[g(\widehat{\beta}^{(r-1)}) +
h(\widehat{\beta}^{(r-1)})] + \gamma [g(\tilde{\beta}_g^{(r)}) +
\tilde{h}(\tilde{\beta}_g^{(r)})],
\end{array}
\end{align*}
where $\gamma \in (0,1)$.
If one of the updating criterion is satisfied,
then $\widehat{\beta}^{(r)}$ is updated as the solution of the corresponding subproblem, i.e., $\tilde{\beta}_g^{(r)}$ or $\tilde{\beta}_h^{(r)}$.
Otherwise $\widehat{\beta}^{(r)}$ remains unchanged: $\widehat{\beta}^{(r)}=\widehat{\beta}^{(r-1)}$.

As it solves the dual problem \eqref{eqn:alin}, the ALIN algorithm has a similar drawback as the path algorithm for the generalized lasso: if the number of rows $m$ of $\bf D$ increases, solving the dual is less efficient than solving the primal. Moreover, the matrix ${\bf D}{\bf A}^{-1} {\bf D}^T$ in \eqref{eqn:alin}  
is not  positive definite when $m$ is greater than $p$. This violates the assumption of the PCG.

\subsection{Summary of the reviewed algorithms} 

Table 1 summarizes the existing algorithms explained in Section 2.1 and 2.2 and
the MM algorithm proposed in Section 3 according to
types of the design matrix ${\bf X}$ and penalty structure.
The items marked with a filled circle ($\bullet$) 
denote the availability of the algorithm.

\begin{table}[!htb]
\caption{Summary of algorithms for solving the FLSA and the FLR problems.}
\medskip
\begin{minipage}{\textwidth}
\centering
\begin{tabular}{|l|c|c|c|c|} \hline
\multirow{2}*{Method\let\thefootnote\relax\footnotetext{Abbreviations of the names of the algorithms are given within the parentheses for future references.}} & \multicolumn{2}{c|}{FLSA (${\bf X} = {\bf I}$)} & \multicolumn{2}{c|}{FLR (general ${\bf X}$)} \\ \cline{2-5}
&STD\let\thefootnote\relax\footnotetext{STD denotes the standard fusion penalty $\sum_{j=2}^p |\beta_j -\beta_{j-1}|$.}  & GEN\let\thefootnote\relax\footnotetext{GEN denotes the generalized fusion penalty $\sum_{(j,k) \in E} |\beta_j -\beta_{k}|$ for a given
index set $E$.} &$\phantom{a}$STD$\phantom{a}$ & GEN$$\\ \hline
Path-wise optimization (\textsf{pathwise}) & $\bullet$\let\thefootnote\relax\footnotetext{$\bullet$ denotes that the method
is applicable.}& & &  \\ 
Path algorithm for the FLSA (\textsf{pathFLSA}) &$\bullet$& $\bullet$ & &  \\ 
Path algorithm for the generalized lasso (\textsf{genlasso}) & $\bullet$&  $\vartriangle$\let\thefootnote\relax\footnotetext{$\vartriangle$ denotes that the dual solution path is computationally infeasible in large scale image denoising.}& $\blacktriangle$\let\thefootnote\relax\footnotetext{$\blacktriangle$ denotes the method
is only applicable when ${\bf X}$ has full rank.} & $\blacktriangle$  \\ 
Efficient Fused Lasso Algorithm (\textsf{EFLA}) & $\bullet$& &  $\bullet$&  \\ 
Smooth Proximal Gradient  (\textsf{SPG}) & $\bullet$&$\circ$\let\thefootnote\relax\footnotetext{$\circ$ denotes the method is not adequate since computation of spectral matrix norm when $p$ is large.} &$\bullet$ &$\circ$\\ 
Split Bregman  (\textsf{SB}) & $\bullet$&$\bullet$ &$\bullet$ &$\bullet$ \\ 
Alternating Linearization  (\textsf{ALIN})  & $\bullet$&$\bullet$ &$\bullet$ &$\bullet$ \\  \hline
Majorization-Minimization  (\textsf{MM}) &$\bullet$&$\bullet$ &$\bullet$ &$\bullet$
 \\ 
MM with
GPU parallelization (\textsf{MMGPU}) &$\bullet$& &$\bullet$ & \\ \hline
\end{tabular}
\end{minipage}
\end{table}

\section{MM algorithm for fused lasso problem}

\subsection{MM algorithm}

In this section we propose an MM algorithm to solve the FLR
 (\ref{eqn:obj}). The MM algorithm iterates two steps, the majorization step and  
the minimization step. Given the current estimate $\widehat{\beta}^{0}$
of the optimal solution to (\ref{eqn:obj}), the majorization step finds a majorizing function $g\big(\beta \big|\widehat{\beta}^0 \big)$ such that
 $f\big(\widehat{\beta}^0\big) = g \big(\widehat{\beta}^0  \big|\widehat{\beta}^0 \big)$ and $f \big(\beta \big) \le g \big( \beta \big|\widehat{\beta}^0\big)$ for all $\beta \neq \widehat{\beta}^0$. 
The minimization step updates the estimate with 
\begin{equation} \nonumber 
\widehat{\beta} = \argmin_{\beta} ~g \big( \beta \big| \widehat{\beta}^0\big).
\end{equation}
It is known that the MM algorithm has a descent property in the 
sense that $f\big(\widehat{\beta}\big) \le f \big(\widehat{\beta}^0 \big)$.
Furthermore, if $f\big(\beta\big)$
and  $g \big( \beta \big|\widehat{\beta}^0\big)$ satisfy some regularity conditions, e.g., those in \cite{Vaida2005}, the MM algorithm converges to the optimal solution of $f\big(\beta\big)$.

The MM algorithm has been applied to the classical lasso problem by \citet{HunterLi2005}. They propose to minimize a perturbed 
version of the objective function of the classical lasso
\begin{equation} \nonumber
 f_{\lambda,\eps}(\beta) = \frac{1}{2} \|{\bf y}- {\bf X}\beta\|_2^2
 		+ \lambda \sum_{j=1}^p \Big(|\beta_j| - \eps \log \big( 1 + \frac{|\beta_j|}{\eps} \big) \Big)
\end{equation}
instead of
\begin{equation} \nonumber
f_\lambda(\beta) =  \frac{1}{2} \|{\bf y}- {\bf X}\beta\|_2^2 + \lambda \|\beta\|_1
\end{equation}
to avoid  division by zero in the algorithm.
The majorizing function for $f_{\lambda,\eps}(\beta)$ at $\widehat{\beta}^0$ is given by
 \begin{eqnarray}
 	 g_{\lambda,\eps}(\beta |\widehat{\beta}^0) &=&
 	\frac{1}{2} \|{\bf y}- {\bf X}\beta\|_2^2  
  ~+ \lambda \sum_{j=1}^p \left\{|\widehat{\beta}^0_j| - \eps \log \left( 1 + \frac{|\widehat{\beta}^0_j|}
 		{\eps} \right)  + \frac{\beta_j^2 - (\widehat{\beta}^0_j)^2}{2 \Big(|\widehat{\beta}^0_j| + \eps\Big)} \right\}. \nonumber
 \end{eqnarray}
 \citet{HunterLi2005} show that the sequence $\{ \widehat{\beta}^{(r)}\}_{r \ge 0}$ with $\widehat{\beta}^{(r+1)} = \argmin_\beta g_{\lambda,\eps}\big(\beta \big| \widehat{\beta}^{(r)}\big)$ converges to
  the minimizer of $f_{\lambda,\eps}(\beta)$, and that $f_{\lambda,\eps}(\beta)$ 
  converges to $f_{\lambda}(\beta)$ uniformly as $\eps$ approaches $0$.

Motivated by the above development, we introduce 
a perturbed version $f_\eps(\beta)$ of the objective (\ref{eqn:obj})  and
the majorizing function  $g_\eps(\beta|\widehat{\beta}^0)$ for the FLR problem (\ref{eqn:obj}):
\begin{align} \label{eqn:pert}
f_{\eps} (\beta) &= \displaystyle \frac{1}{2} \| {\bf y}- {\bf X} \beta \|_2^2  
+ \lambda_1 \sum_{j=1}^p \left\{|\beta_j | - \eps \log \Big(1 + \frac{|\beta_j|}{\eps} \Big) \right\} \nonumber \\ 
& \displaystyle + \lambda_2 \sum_{(j,k) \in E} \left\{|\beta_j - \beta_{k}| - \eps \log \Big(1 + \frac{|\beta_j - \beta_{k}|}{\eps} \Big) \right\}, 
\end{align}
and
\begin{align} \label{eqn:major}
g_{\eps} (\beta | \widehat{\beta}^0) &=  
\displaystyle \frac{1}{2} \| {\bf y}- {\bf X} \beta \|_2^2  
 \displaystyle + \lambda_1 \sum_{j=1}^p \Bigg\{|\widehat{\beta}^0_j|  - \eps \log \Big(1 + \frac{|\widehat{\beta}^0_j|}{\eps} \Big)  
 + \frac{\beta_j^2 - (\widehat{\beta}^0_j)^2}{2(|\widehat{\beta}^0_j|+\eps)} \Bigg\} \nonumber  \\
& \quad \displaystyle  + \lambda_2 \sum_{(j,k) \in E}  \Bigg\{
\big|\widehat{\beta}^0_j-\widehat{\beta}^0_k \big| - \eps \log \Big(1 + \frac{|\widehat{\beta}^0_j -\widehat{\beta}^0_k|}{\eps} \Big)   
 + \frac{(\beta_j-\beta_{k})^2 - (\widehat{\beta}^0_j-\widehat{\beta}^0_k)^2}{2(|\widehat{\beta}^0_j-\widehat{\beta}^0_k|+\eps)} 
\Bigg\}. 
\end{align}
Note that $g_\eps(\beta | \widehat{\beta}^0)$ is a smooth convex function in $\beta$. Below we see that $g_\eps(\beta | \widehat{\beta}^0)$ indeed majorizes $f_\eps(\beta)$ at $\widehat{\beta}^0$, and has a unique global minimum  regardless of the rank of ${\bf X}$, together with several properties of the perturbed objective $f_{\eps}(\beta)$:

\begin{proposition}\label{prop:mm}
For $\eps >0$ and $\lambda_1>0$, 
(i) $f_\eps (\beta)$ is continuous, convex, and satisfies $\lim_{ \|\beta\|_2 \rightarrow  \infty} f_{\epsilon} (\beta) = \infty$, 
(ii) $f_{\eps}(\beta)$ converges to $f(\beta)$ uniformly on any compact set $\bf C$ as $\eps$ approaches to zero.
(iii) $g_\eps (\beta|\beta')$ majorizes $f_\eps (\beta)$ at $\beta'$, and $g_\eps(\beta|\beta')$ has a unique global minimum for all $\beta'$.
\end{proposition}
\begin{proof}
See Appendix A.1.
\end{proof}

Now let $\widehat{\beta}^{(r)}$ be the current estimate in the $r$th iteration.
The majorizing function $g_{\eps}(\beta |\widehat{\beta}^{(r)})$ is minimized over $\beta$ when
$\partial g_\eps (\beta|\widehat{\beta}^{(r)}) / \partial \beta = 0$, 
which can be written as a linear system of equations
\begin{equation} \label{eqn:lin-sys}
 \big( {\bf X}^{T} {\bf X} + \lambda_1 {\bf A}^{(r)}  + \lambda_2 {\bf B}^{(r)} \big) \beta = {\bf X}^{T}  {\bf y},
\end{equation}
where ${\bf A}^{(r)} = {\rm diag}(a_1^{(r)},a_2^{(r)},\ldots,a_p^{(r)})$ with
 $a_{j}^{(r)} =  1/(|\widehat{\beta}_j^{(r)}|+\eps)$ for $1 \le j \le p$, and ${\bf B}^{(r)}= \big(b_{jk}^{(r)} \big)_{1\le j,k\le p}$ is a symmetric and positive
 semidefinite matrix with
 \begin{eqnarray}
b_{jj}^{(r)} &=&  \sum_{k: (j,k)\in E} \frac{1}{|\widehat{\beta}_j^{(r)} - \widehat{\beta}_k^{(r)}| + \eps}, \quad j=1,\ldots ,p, \nonumber \\
b_{jk}^{(r)} &=&b_{kj}^{(r)} =-\frac{1}{|\widehat{\beta}_j^{(r)} - \widehat{\beta}_k^{(r)}| + \eps}, \quad \forall (j,k) \in E. \nonumber
\end{eqnarray}

The procedure described so far is summarized as Algorithm \ref{mm_general}.
Note that once the matrix $ {\bf X}^{T} {\bf X} + \lambda_1 {\bf A}^{(r)}  + \lambda_2 {\bf B}^{(r)}$ is constructed,
the linear system (\ref{eqn:lin-sys}) can be efficiently solved.
Since this matrix is symmetric and positive definite, hence
Cholesky decomposition and back substitution can be applied,
e.g., using \texttt{dpotrf} and \texttt{dpotri} functions
in LAPACK \citep{Anderson1995}.

Furthermore, the construction of the matrix
$ \lambda_1 {\bf A}^{(r)}  + \lambda_2 {\bf B}^{(r)}$
is efficient in the MM algorithm.
Recall from \eqref{eqn:spg} that the generalized fusion penalty in \eqref{eqn:obj} can be written as
$\| {\bf D} \beta \|_1$ for a $m \times p$ dimensional
matrix ${\bf D}$, where $m$ is the number of penalty terms.
In Algorithm \ref{mm_general}, 
the matrix $ \lambda_1 {\bf A}^{(r)}  + \lambda_2 {\bf B}^{(r)}$ is constructed in $O(m)$ time, i.e., independent of the problem dimension $p$.
This is a great advantage of the MM algorithm over other algorithms reviewed in Section \ref{sec:review}:
\textsf{SPG} requires to check a specific condition
to guarantee its convergence, which takes at least $O\big(np\big)$ time
for every iteration;
\textsf{SB} needs to construct the matrix  
 $ \mu {\bf I}  + \mu {\bf D}^T {\bf D}$,
where $\mu$ is the  parameter discussed in Section 2.2, and
it takes $O(mp^2)$ time;
\textsf{ALIN} and \textsf{genlasso} need
 to compute the matrix ${\bf D}{\bf D}^T$ for solving
 the dual, hence requires $O(m^2p)$ time.
Although in the MM algorithm the matrix 
$ \lambda_1 {\bf A}^{(r)}  + \lambda_2 {\bf B}^{(r)}$
needs to be constructed for each iteration,
our experience is that the MM algorithm tends to converge
within a few tens of iterations; see Section \ref{sec:numerical}. Thus
if $p$ is more than a few tens, which is of our genuine
interest, the MM algorithm is expected to run faster than the other algorithms.

In case the matrix ${\bf M}^{(r)}= \lambda_1 {\bf A}^{(r)} +\lambda_2 {\bf B}^{(r)}$
is tridiagonal as in the standard FLR, or
block tridiagonal as in the two-dimensional FLR,
we can employ the preconditioned conjugate gradient (PCG)
method by using the matrix ${\bf M}^{(r)}$ as the preconditioner
to efficiently solve the linear system (\ref{eqn:lin-sys}).
The PCG method is an iterative method for
solving a linear system ${\bf Q}{\bf x} ={\bf c}$,
where ${\bf Q}$ is a symmetric positive definite matrix,
with a preconditioner matrix ${\bf M}$ \citep{Demmel1997}.
For example, we can use ${\bf M}={\rm diag}({\bf X}^T {\bf X})$ for ${\bf Q} = {\bf X}^T {\bf X}$.
The PCG method achieves the solution by iteratively updating
four sequences $\{{\bf x}_j\}_{j\ge0}$, $\{{\bf r}_j\}_{j\ge0}$, $\{{\bf z}_j\}_{j\ge0}$, and $\{{\bf p}_j\}_{j\ge1}$ with
\begin{equation} \nonumber
\begin{array}{ccl}
{\bf x}_j &=& {\bf x}_{j-1} + \nu_j {\bf p}_j, \\
{\bf r}_j &=& {\bf r}_{j-1} + \nu_j {\bf Q} {\bf p}_j,\\
{\bf z}_j &=& {\bf M}^{-1} {\bf r}_j,\\
{\bf p}_{j+1} &=& {\bf z}_j + \gamma_{j} {\bf p}_j,\\
\end{array}
\end{equation}
where $\nu_j = ({\bf z}_{j-1}^T {\bf r}_{j-1})/({\bf p}_j^T {\bf Q} {\bf p}_j)$, $\gamma_{j} = ({\bf z}_j^T {\bf r}_j)/({\bf z}_{j-1}^T {\bf r}_{j-1})$, ${\bf x}_0 = 0$, ${\bf r}_0 = {\bf c}$, ${\bf z}_0 = {\bf M}^{-1} {\bf c}$, and ${\bf p}_1 = {\bf z}_0$.
For the standard FLR, the matrix ${\bf M}^{(r)}$ is tridiagonal and
its Cholesky decomposition can be evaluated within
$ {O\big( p \big)}$ time, much faster than
that of general positive definite matrices.
In this case, we can solve the linear system ${\bf M}^{(r)} {\bf z}_j = {\bf r}_j$
using LAPACK functions \texttt{dpbtrf} and \texttt{dpbtrs}
that perform Cholesky decomposition and  linear system solving for symmetric positive definite band matrix, respectively.
Algorithm \ref{mm_pcg} describes the proposed MM algorithm using the PCG method for the standard FLR.

\begin{algorithm}
\caption{MM algorithm for the FLR}\label{mm_general}
\begin{algorithmic}[1]
\Require ${\bf y}$, ${\bf X}$, $\lambda_1$, $\lambda_2$,
convergence tolerance $\delta$, perturbation constant $\eps$.
      \For{$j = 1,\cdots, p$}
       \medskip
      \State $\displaystyle \widehat{\beta}_j^{(0)} \gets \frac{ X_j^{T} {\bf y} }{ X_j^{T}X_j}$\Comment{initialization}
       \medskip
   \EndFor
    \medskip
   \Repeat{~~$r = 0,1,2,\ldots$}
    \medskip
      \State $ {\bf M}^{(r)} \gets  {\bf X}^{T} {\bf X} + \lambda_1 {\bf A}^{(r)}  + \lambda_2 {\bf B}^{(r)}$
       \medskip
      \State $\widehat{\beta}^{(r+1)} \gets \big({\bf M}^{(r)}\big)^{-1} {\bf X}^{T} {\bf y}$\Comment{using \texttt{dpotrf} and \texttt{dpotri} in LAPACK }
     \medskip
    \Until{ $ \displaystyle \frac{| f(\widehat{\beta}^{(r+1)}) - f(\widehat{\beta}^{(r)}) |}{ | f(\widehat{\beta}^{(r)}) | }\le \delta$}
     \medskip
\Ensure $\widehat{\beta}_\eps  = \widehat{\beta}^{(r+1)}$
\end{algorithmic}
\end{algorithm}

\begin{algorithm}
\caption{\newline MM algorithm using the PCG method for the standard FLR and the
two-dimensional FLR} \label{mm_pcg}
\begin{algorithmic}[1]
\Require ${\bf y}$, ${\bf X}$, $\lambda_1$, $\lambda_2$,
convergence tolerance $\delta$, perturbation constant $\eps$.

            \For{$j = 1,\cdots, p$}
            \medskip
      \State $\displaystyle \widehat{\beta}_j^{(0)} \gets \frac{ X_j^{T} {\bf y} }{ X_j^{T}X_j}$\Comment{initialization}
      \medskip
   \EndFor
   \medskip
   \Repeat{  $r = 0,1,2,\ldots$}
   \medskip
         \State ${\bf M }^{(r)} \gets \lambda_1 {\bf A}^{(r)} + \lambda_2 {\bf B}^{(r)}$\Comment{preconditioner}
      \medskip
       \State ${\bf x}_0 \gets \widehat{\beta}^{(r)}$ \Comment{begin the PCG method}
      \medskip
      \State ${\bf r}_0 \gets {\bf X}^T {\bf y} - ({\bf X}^T {\bf X} + {\bf M}^{(r)}) {\bf x}_0$ 
      \medskip
      \State ${\bf z}_0 \gets \big({\bf M}^{(r)}\big)^{-1} {\bf r}_0$
      \medskip
      \State ${\bf p}_1 \gets {\bf z}_0$
      \medskip
      \Repeat{ ~~$j = 1,2,\ldots$}
      \medskip
               \State $ \displaystyle \nu_j \gets \frac{{\bf r}^T_{j-1} {\bf z}_{j-1} }{ {\bf p}^T_j ({\bf X}^T {\bf X} + {\bf M}^{(r)}) {\bf p}_j}$
\medskip
\State $ {\bf x}_j \gets {\bf x}_{j-1} + \nu_j {\bf p}_j$
         \medskip
         \State $ {\bf r}_j \gets {\bf r}_{j-1} + \nu_j ({\bf X}^T {\bf X} + {\bf M}^{(r)}) {\bf p}_j$
         \medskip
            \State Solve ${\bf M}^{(r)} {\bf z}_j = {\bf r}_j$\Comment{using \texttt{dpbtrf} and \texttt{dpbtrs} in LAPACK }
         \medskip
         \State $\displaystyle \gamma_j \gets \frac{{\bf r}^T_j {\bf z}_j}{ {\bf r}^T_{j-1} {\bf z}_{j-1}}$
         \medskip
         \State ${\bf p}_{j+1} \gets {\bf z}_j + \gamma_j {\bf p}_{j}$\Comment{update conjugate gradient}
         \medskip
      \Until{ $\displaystyle \frac{\| {\bf r}_j\|_2 }{ \|{\bf X}^T {\bf y}\|_2} \le \delta$} \Comment{end the PCG method}
      \medskip
       \State $\widehat{\beta}^{(r+1)} \gets {\bf x}_j$\Comment{update solution}
       \medskip
   \Until{ $ \displaystyle \frac{| f(\widehat{\beta}^{(r+1)}) - f(\widehat{\beta}^{(r)}) |}{ | f(\widehat{\beta}^{(r)}) | }\le \delta$}
  \medskip
\Ensure $\widehat{\beta}_\eps  = \widehat{\beta}^{(r+1)}$
\end{algorithmic}
\end{algorithm}

\subsection{Convergence}

In this section we provide some results on the convergence of the proposed MM algorithm for the FLR. 
We first show that a minimizer $\widehat{\beta}_\eps$ of the perturbed version \eqref{eqn:pert} of the objective \eqref{eqn:obj} exhibits a minimal difference
from a minimizer $\widehat{\beta}$ of the true objective \eqref{eqn:obj} for sufficiently small $\eps$. 
We then  see that the sequence of the solutions $\{\widehat{\beta}_{\epsilon}^{(r)}\}_{r \ge 0}$ of the MM algorithm converges to $\widehat{\beta}_{\epsilon}$.

The key to the proof of the following lemma is to see that the level sets of both the perturbed objective $f_{\eps}(\beta)$ and the true objective $f(\beta)$ are compact, which is a consequence of Proposition \ref{prop:mm}.
\begin{lemma}\label{lemma:perturbed}
Consider an arbitrary deceasing sequence $\big\{\epsilon_n, n=1,2,\ldots,\infty \big\}$ that converges to 0. Then, any limit point of $\widehat{\beta}_{\epsilon_n}$ is a minimizer of $f\big(\beta \big)$, provided that  $\{ \beta ~|~ f(\beta) = f(\widehat{\beta})\}$ is non-empty.
\end{lemma}
\begin{proof} 
See Appendix A.2.
\end{proof}

Lemma \ref{lemma:lyapunov} states that $f_{\eps}(\beta)$ serves as a Lyapunov function for the nonlinear recursive updating rule, implicitly defined by the proposed MM algorithm, that generates the sequence $\{\widehat{\beta}_{\eps}^{(r)}\}$. Monotone convergence immediately follows from the descent property of the MM algorithm.%

\begin{lemma}\label{lemma:lyapunov}
Let ${\bf S}$ be a set of stationary points of $f_{\eps} (\beta)$ and $\{ \widehat{\beta}_{\eps}^{(r)} \}_{r\ge 0}$ be a sequence of solutions of the proposed MM algorithm given by $\widehat{\beta}_{\eps}^{(r+1)} = M(\widehat{\beta}_{\eps}^{(r)}) = \argmin_{\beta} g_\eps(\beta|\widehat{\beta}_{\eps}^{(r)})$. 
Then the limit points of $\{\widehat{\beta}_\eps^{(r)} \}_{r \ge 0}$ are stationary points of $f_\eps(\beta)$, i.e, $\partial f_{\eps}(\beta)/\partial \beta = 0$ at these limit points.
Moreover, $f_{\eps}(\widehat{\beta}_{\eps}^{(r)})$ converges monotonically to $f_\eps(\beta_{\eps}^*)$ for some $\beta_{\eps}^* \in {\bf S}$.
\end{lemma}

\begin{proof}
See Appendix A.3.
\end{proof}

Combining Lemmas 1 and 2, it is straightforward to see that the MM algorithm generates solutions that converge to an optimal solution of the FLR problem \eqref{eqn:obj} as $\eps$ decreases to zero.
\begin{theorem}
The sequence of the solutions $\{\widehat{\beta}_{\eps}^{(r)}\}_{r\ge 0}$ generated by the proposed MM algorithm converges to a minimizer of $f(\beta)$. 
Moreover, the sequence of functionals $\{f_{\eps}(\widehat{\beta}_{\eps}^{(r)})\}_{r\ge 0}$ converges to the minimum value of $f(\beta)$.
\end{theorem} 
\begin{proof} 
Note for the first part
\[
\big\| \widehat{\beta}_{\eps}^{(r)} - \beta^{**} \big\| \le \big\| \widehat{\beta}_{\eps}^{(r)} - \beta_{\eps}^*  \big\| + \big\| \beta_{\eps}^* - \beta^{**} \big\|, 
\]
where $\beta_{\eps}^*$ is a stationary point of, hence minimizes, $f_{\eps}(\beta)$, and $\beta^{**}$ is a limit point of $\{\beta_{\eps}^*\}$ as $\eps \downarrow 0$.
The first term in the right-hand side becomes arbitrarily small for sufficiently large $r$ by lemma \ref{lemma:lyapunov}, 
whereas the second term does for sufficiently small $\eps$ by lemma \ref{lemma:perturbed}. The limit point $\beta^{**}$ is a minimizer of $f(\beta)$ also by lemma \ref{lemma:perturbed}.
Now we see 
\[
|| f_{\eps}(\widehat{\beta}_{\eps}^{(r)})-f(\beta^{**}) || \le
|| f_{\eps}(\widehat{\beta}_{\eps}^{(r)})-f_{\eps}(\beta_{\eps}^*) || +
|| f_{\eps}(\widehat{\beta}_{\eps})-f(\beta_{\eps}^{*}) || +
|| f(\beta_{\eps}^{*})-f(\beta^{**}) ||.
\]
The first in the right-hand side vanishes by the continuity of $f_{\eps}(\beta)$; the second term by the uniform convergence of $f_{\eps}(\beta)$ to $f(\beta)$, as shown in the proof of lemma \ref{lemma:perturbed}; and the third term by the continuity of $f(\beta)$.
\end{proof}

\section{Parallelization of the MM algorithm with GPU}

The graphics processing unit (GPU) was originally developed as a co-processor to reduce the workload of the central processing unit (CPU) for computationally expensive graphics operations. A GPU has a collection of thousands of light-weight computing cores,  which makes it suitable for massively parallel applications \citep{Owens2008}. While early GPUs could only perform simple bitmap processing for generating pixels on screen, as the applications require high-level three-dimensional graphics performance it has become possible that users can program so-called shader controlling directly each step of the graphics pipeline that generates pixels from the three-dimensional geometry models.
As a GPU generates a massive number of pixels on screen simultaneously, shader programming can be understood as conducting the identical computation with many geometric primitives (such as vertices) as the input to determine the output, or the intensities of the pixels.

Thus for a general-purpose computing problem, if the problem has \emph{data-level parallelism}, that is, on many portions of data the computation can be executed at the same time, then by writing an appropriate shader program we can map the data to the geometric primitives and can parallelize the general-purpose computing.
Due to this single-instruction multiple-data (SIMD) parallel architecture together with its economic attraction, the use of the GPU has been expanded to general-purpose scientific computing beyond graphics computation. As the demand in general-purpose computing increases, both the architecture and the library for the GPU have been evolved to support such computation, e.g., NVIDIA's Compute Unified Device Architecture (CUDA) \citep{Kirk2010,Farber2011}.
As mentioned above, applications in which the GPU is effective involve separation of data and parameters, such as nonnegative matrix factorization, positron emission tomography, and multidimensional scaling  \citep{Zhou2010}.

For the standard FLR, the PCG step in (\ref{eqn:lin-sys})
that is frequently encountered in the inner iteration
of Algorithm \ref{mm_pcg} can be greatly accelerated by the
GPU parallelization. 
The motivation of the parallelization comes from
the parallel algorithms such as
cyclic reduction (CR), parallel cyclic reduction (PCR), and
recursive doubling (RD) \citep{Hockney1965,Hockney1981,Stone1973} that can solve a tridiagonal system to which \eqref{eqn:lin-sys} reduces. 
 Although these parallel algorithms require more operations 
 than a basic Gaussian elimination for tridiagonal system \citep{Zhang2010},
they can be executed more efficiently than the basic Gaussian elimination if many cores are utilized simultaneously.
Moreover, a hybrid algorithm, which combines the CR and the PCR,
has been developed for GPUs recently \citep{Zhang2010} and is implemented in 
CUDA 
sparse matrix library (\texttt{cuSPARSE}) as the function \texttt{cusparseDgtsv}.
We use this function to solve the linear system ${\bf M}^{(r)} {\bf z_j} = {\bf r_j}$ in the GPU implementation of Algorithm \ref{mm_pcg}.

In addition to the tridiagonal system, 
Algorithm \ref{mm_pcg} also contains a set of
linear operations suitable for GPU parallelization.
For instance, lines 7, 11--13, and 15--17 rely on basic linear operations such as matrix-matrix multiplication, matrix-vector multiplication, vector-vector addition, inner-product, and $\ell_2$-norm computation. 
These operations involve identical arithmetic for 
each coordinate, hence the SIMD architecture of GPU is appropriate
for parallelizing the algorithm.
The CUDA 
basic linear algebra subroutines
library (\texttt{cuBLAS}) provides an efficient implementation of such operations, hence we use \texttt{cuBLAS} for parallelizing lines 7, 11--13, and 15--17
with GPU.

Unfortunately, our parallelization of the PCG step
is limited to the standard FLR and at this moment is
not extendable to the generalized fusion penalty.
In this case, we recommend to solve the linear system
\eqref{eqn:lin-sys} with LAPACK functions \texttt{dpotrf} and \texttt{dpotri} using CPUs as in Algorithm \ref{mm_general}. Currently this is more efficient than using GPUs.

\section{Numerical studies}\label{sec:numerical}

In this section, we compare the performance of the MM algorithm with
the other existing algorithms using several data sets
that generally fit into the ``large $p$, small $n$'' setting.
The general QP solvers are excluded in the comparison,
since they are slower than other specialized algorithms
for the FLR.
These numerical studies include five scenarios for the FLR problem.
In the first three scenarios,
we consider the standard FLR with various sparsities of the true
coefficients. 
Since the path algorithms
have restrictions on solving the FLR with
general design matrices,
we only compare the proposed MM algorithms 
(with and without GPU parallelization, denoted by \textsf{MMGPU} and \textsf{MM}, respectively) with \textsf{EFLA}, \textsf{SPG}, and \textsf{SB} algorithms.
In next two scenarios,
we consider the two-dimensional FLR  with a general design matrix
and its application to the image denoising problem.
In the two-dimensional FLR with a general design matrix,
only \textsf{MM}, \textsf{SPG}, and \textsf{SB} are available.
In the image denoising problem, which is equivalent to the two-dimensional FLSA problem,
we additionally consider the \textsf{pathFLSA} algorithm. 
The \textsf{MM} and \textsf{MMGPU} algorithms are implemented in C and CUDA C.

The convergence of the algorithm is measured using a relative error of the objective function at
each iteration. The relative error at the $r$th iteration is defined as
${\rm RE}(\widehat{\beta}^{(r)}) = |f(\widehat{\beta}^{(r)}) - f(\widehat{\beta}^{(r-1)})|/f(\widehat{\beta}^{(r-1)})$, 
where  $\widehat{\beta}^{(r)}$ is  the estimate in $r$th iteration (see Algorithms \ref{mm_general} and \ref{mm_pcg}). 
All the algorithms terminate
when the relative error becomes smaller than $10^{-5}$ or the number of iterations exceeds $10000$. The \textsf{SB} and \textsf{MM} algorithms have additional parameters 
$\mu$ and $\epsilon$, respectively. The details of choosing $\mu$ of the \textsf{SB} are given in Appendix B. In the \textsf{MM} and \textsf{MMGPU}, we set the perturbation $\eps = 10^{-8}$ to avoid machine precision error.\footnote[1]{Although setting $\eps$ a constant does not exactly satisfy the condition for Theorem 1, this choice of constant is sufficiently small within the machine precision and to prevent divide by zero.}

The algorithms are compared in the computation the time and the number of iterations for aforementioned five scenarios. We also investigate the sensitivity of the SB algorithm and the MM algorithm to the choice of their additional parameters $\mu$ and $\epsilon$. All algorithms are implemented in MATLAB except for \textsf{pathFLSA}, The computation times are measured in CPU time by using a desktop PC (Intel Core2 extreme X9650 CPU (3.00 GHz) and 8 GB RAM) with NVIDIA GeForce GTX 465 GPU.

\subsection{Standard FLR}

We consider three scenarios for the standard FLR and 
check the efficiency and the stability of the five algorithms, i.e.,
\textsf{MMGPU}, \textsf{MM}, \textsf{EFLA}, \textsf{SPG}, and \textsf{SB}, with four sets of
regularization parameters $(\lambda_1,\lambda_2)
\in \big\{ (0.1,0.1), (0.1,1),(1,0.1),(1,1) \big\}$.\footnote[2]{
These choices of regularization parameters are similar
in \cite{,Liu2010, Lin2011, Ye2011}.} 
We generate $n$ samples $X_1, X_2, \ldots, X_n$ from  a $p$-dimensional multivariate
normal distribution $N({\bf 0}, {\bf I}_p)$.
The response variable ${\bf y}$ is generated from the model
\begin{equation} \label{eqn:gener}
 {\bf y} = {\rm \bf X} {\bf \beta} + \epsilon,
\end{equation}
where $\epsilon \sim N\big(0,{\bf I}_n\big)$ and ${\rm \bf X} = (X_1,
X_2, \ldots, X_n)^T$.
To examine the performance with the dimension $p$,
we try $p = 200,~1000,~10000,~{\rm and}~ 20000$
for sample size $n=1000$.
 We generate 10 data sets according to the following scenarios ${\bf (C1)}$--${\bf (C3)}$.  
\begin{itemize}
\item[{\bf (C1)}] Sparse case.

Following \cite{Ye2011}, we set 41 coefficients of $\beta$ nonzero:
\begin{equation} \nonumber
\beta_j= \left\{
\begin{array}{lll}
2 & & {\rm for}~~j=1,2,\ldots,20, 121,\ldots,125\\
3 & & {\rm for}~~j = 41\\
1 & & {\rm for}~~j=71,\ldots,85\\
0 & & {\rm otherwise}
\end{array} \right.
\end{equation}

\item[{\bf (C2)}] Moderately sparse case.

Following \cite{Lin2011}, we set $30$\% of the coefficients nonzero:
\begin{equation} \nonumber
\beta_j= \left\{
\begin{array}{lll}
1 & & {\rm for}~~ j=p/10+1,~ p/10+2,~\ldots,~2p/10\\
2 & & {\rm for}~~ j=2p/10 + 1,~2p/10+2,~\ldots,~4p/10\\
0 & & {\rm otherwise}
\end{array} \right.
\end{equation}

\item[{\bf (C3)}] Dense case.

We set all the coefficients nonzero: 
\begin{equation} \nonumber
\beta_j= \left\{
\begin{array}{lll}
1 & & {\rm for}~~ j=1,~ 2,~\ldots,~5p/10\\
-1 & & {\rm for}~~ j=5p/10 + 1,~5p/10+2,~\ldots,~p\\
\end{array} \right.
\end{equation}

\end{itemize}

\begin{figure}[htb!] 
\begin{center}
  \begin{minipage}[b]{.32\linewidth}
  \centering   \centerline{\epsfig{file=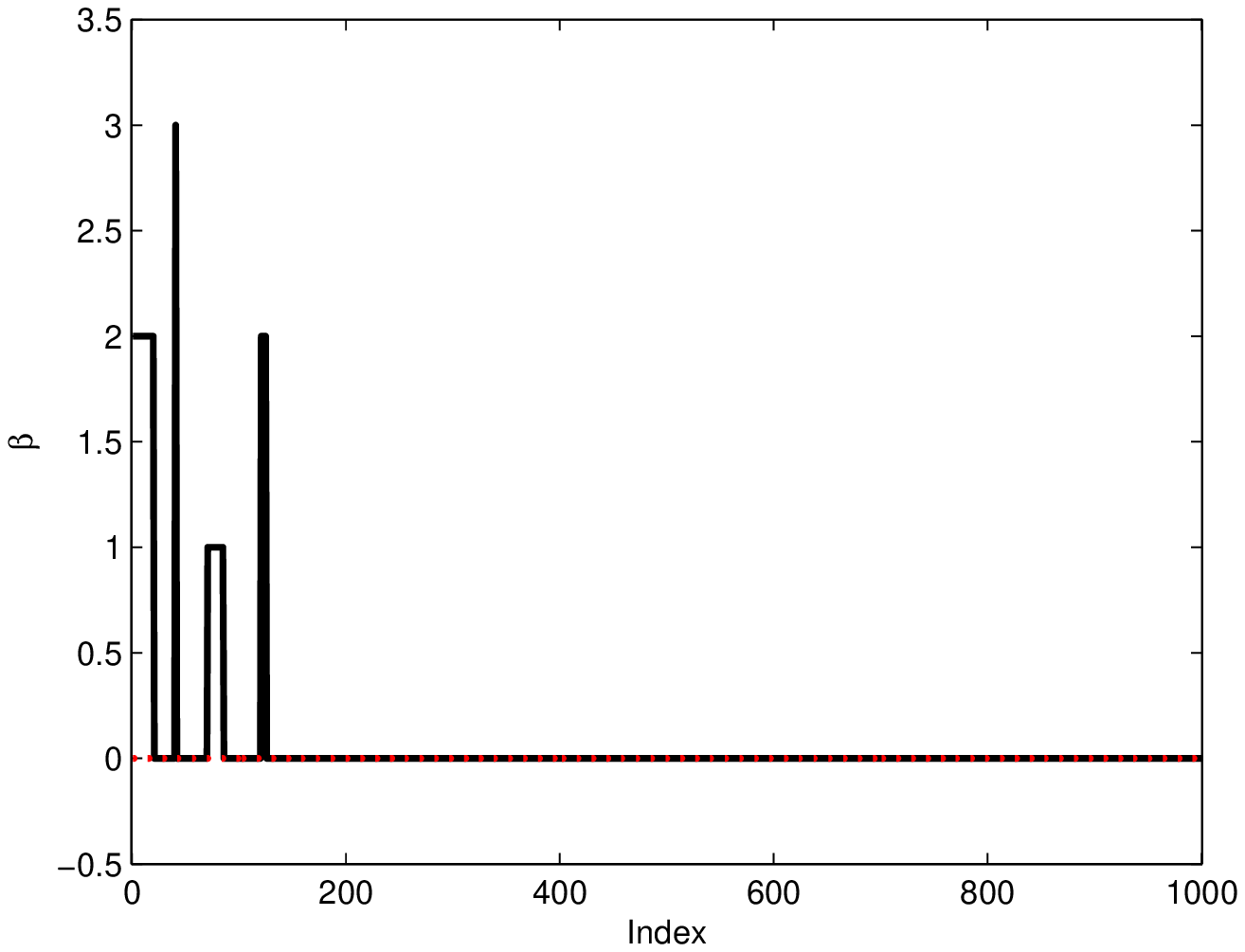,width=\textwidth,height=0.22\textheight}}
  \centerline{(a) Case 1 {\bf (C1)}}
 \end{minipage}
\medskip
  \begin{minipage}[b]{.32\linewidth}
  \centering   \centerline{\epsfig{file=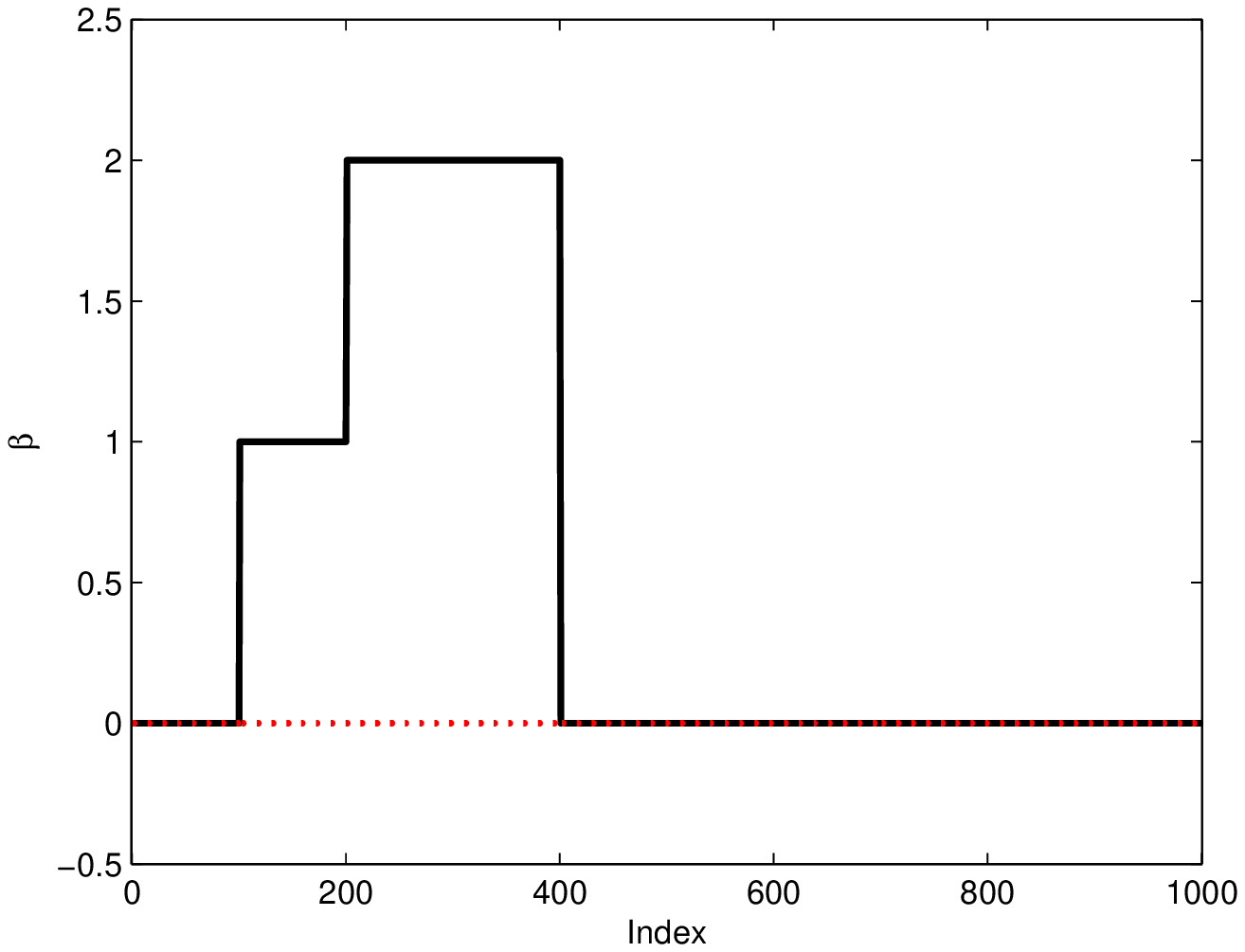,width=\textwidth,height=0.22\textheight}}
  \centerline{(b) Case 2 {\bf(C2)}}
\end{minipage}
\medskip
  \begin{minipage}[b]{.32\linewidth}
  \centering   \centerline{\epsfig{file=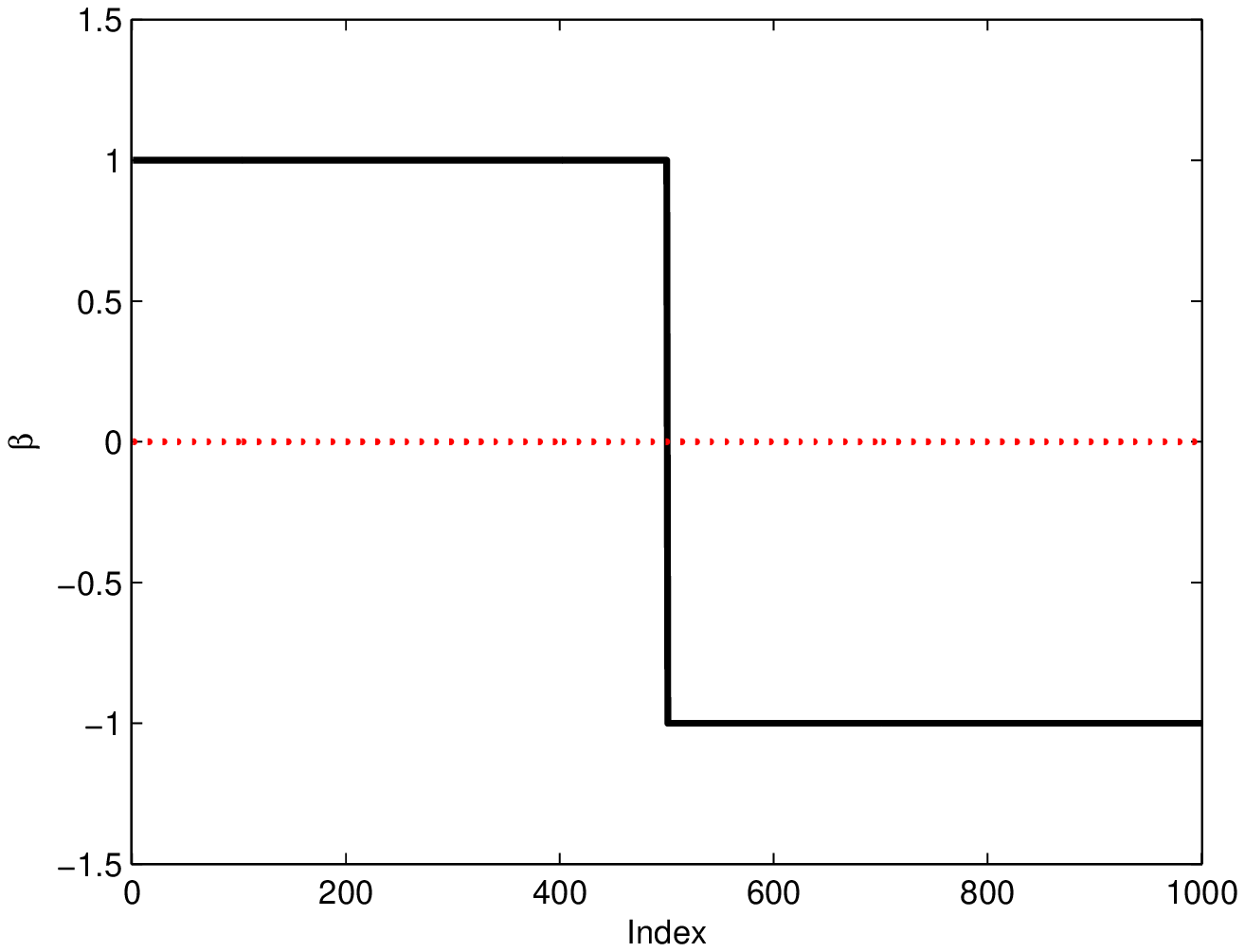,width=\textwidth,height=0.22\textheight}}
  \centerline{(c) Case 3 {\bf(C3)}}
\end{minipage}

\caption{Plots of the true coefficients in ${\bf (C1)}$--${\bf (C3)}$.}
\label{tr_c1} 
\end{center}
\end{figure}
\noindent These true coefficient structures are illustrated in 
 Figure \ref{tr_c1}.

We report the average computation time and the average number of iterations  
in Appendix C, Tables C.1--C.3, and summarize the average computation times in Figure \ref{res_c1}.
Excluding the \textsf{MMGPU} (discussed below),
the \textsf{EFLA} is generally the fastest in most of the scenarios considered.
For relatively dense case ({\bf (C2)} and {\bf (C3)})
with a small penalty ($\lambda_1=\lambda_2=0.1$)
for large dimensions ($p=10000, 20000$), MM is the fastest.
We see that the \textsf{MM} is $1.14 \sim 4.85$  times slower than the \textsf{EFLA},
and is comparable to the \textsf{SPG} and similar to the \textsf{SB} for high dimensions ($p=10000$ or $20000$),
each of which is the contender for the second place with the \textsf{MM} in
the respective dimensions.
Taking into account that the \textsf{EFLA} is only applicable to the standard FLR,
we believe that the performance of the \textsf{MM} is acceptable.

The benefit of parallelization is
visible in large dimensions.
the \textsf{MMGPU} has within about 10\% of the computation time of the \textsf{MM}
and the fastest among all the algorithms
considered when $p = 10000$ and $20000$.
\textsf{MMGPU} is $2 - 63$ times faster than EFLA,
 which in turn is the fastest among \textsf{MM}, \textsf{SPG}, and \textsf{SB}.
When the dimension is small ($p=200$), 
however, \textsf{MMGPU} is the slowest among the five algorithms.
This is due to the memory transfer overhead
between CPU and GPU. When $p$ is small, memory
transfer occupies most of computation time.
As the dimension increases, the relative portion of
the memory transfer in computation time decreases,
and \textsf{MMGPU} becomes efficient.

Focusing on the number of iterations, we see that the \textsf{MM} converges
within a few tens of iterations whereas other algorithms
require up to a few hundreds of iterations, especially
for large $p$.
The number of iterations of \textsf{MM} is also insensitive to the sparsity
structure and the choice of $\lambda_1$ and $\lambda_2$.
(\textsf{EFLA} and \textsf{SPG} need more iterations as
the coefficients become dense, together with \textsf{SB},
they also require more iterations for small $\lambda_1$ or $\lambda_2$.)
This stability in the number of iterations contribute
the performance gain in \textsf{MMGPU}.
\textsf{MMGPU} speeds up each iteration, hence if the number of iteration
does not rely much on the input, we can consistently
expect a gain by parallelization.

\begin{figure}[htb!] 
\begin{center}
  \begin{minipage}[b]{.44\linewidth}
  \centering   \centerline{\epsfig{file=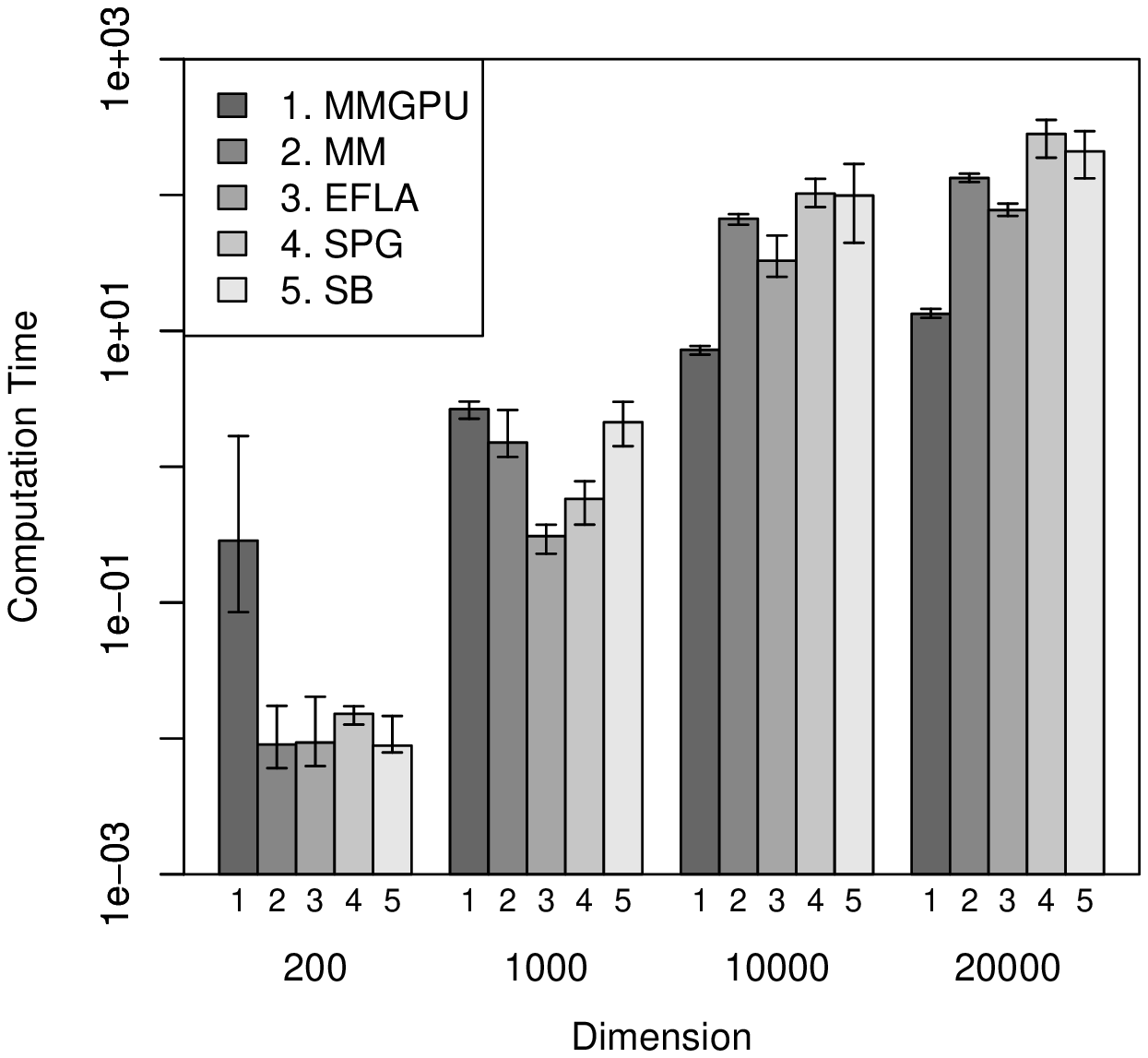,width=0.8\textwidth,
  height=0.21\textheight}}
  \centerline{(a)  {\bf (C1)} with $(\lambda_1,\lambda_2) = (0.1,0.1)$}
 \end{minipage}
\medskip
  \begin{minipage}[b]{.44\linewidth}
  \centering   \centerline{\epsfig{file=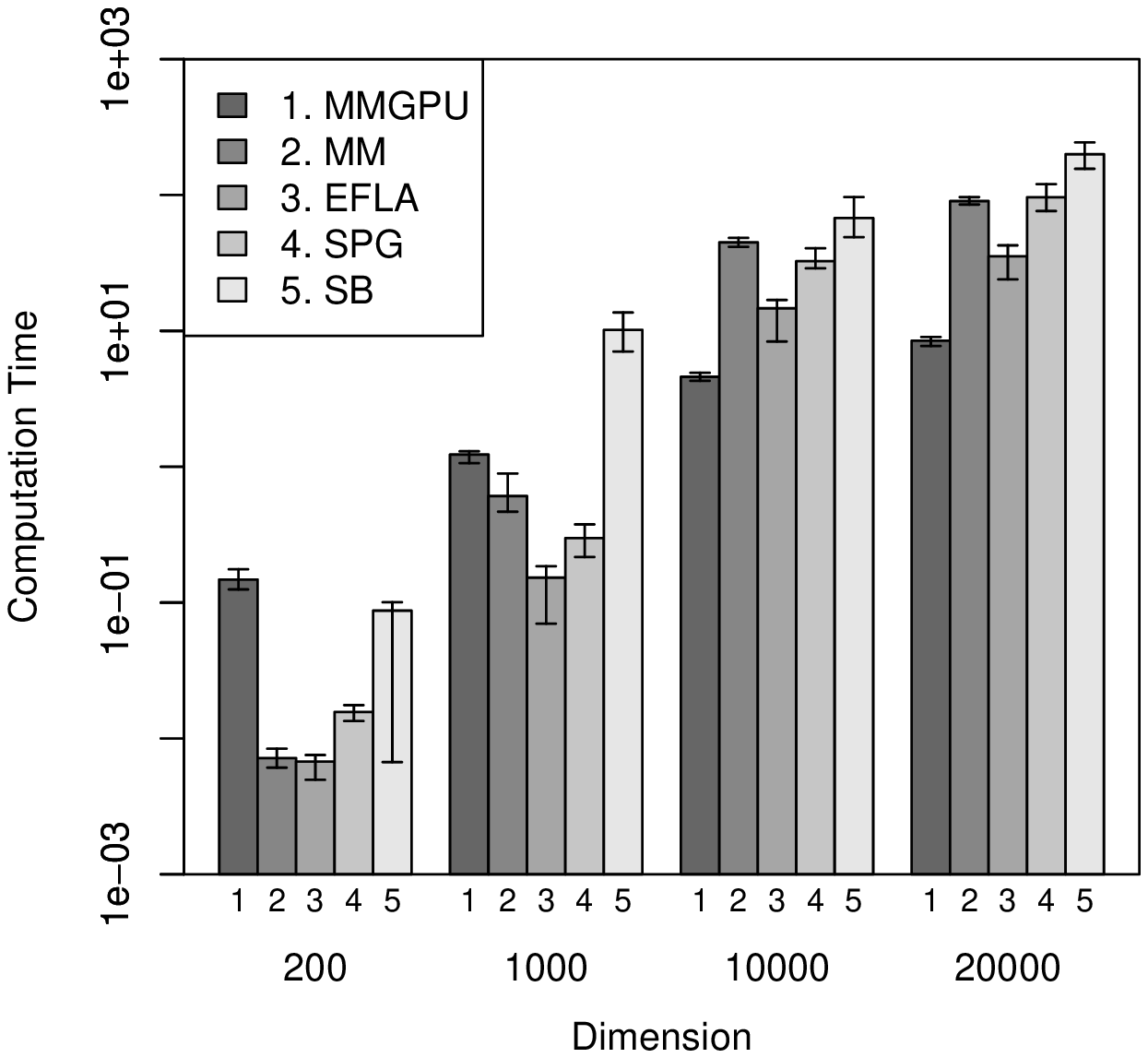,width=0.8\textwidth,
  height=0.21\textheight}}
  \centerline{(b) {\bf (C1)} with $(\lambda_1,\lambda_2) = (1,1)$}
\end{minipage}
\medskip
\begin{minipage}[b]{.44\linewidth}
  \centering   \centerline{\epsfig{file=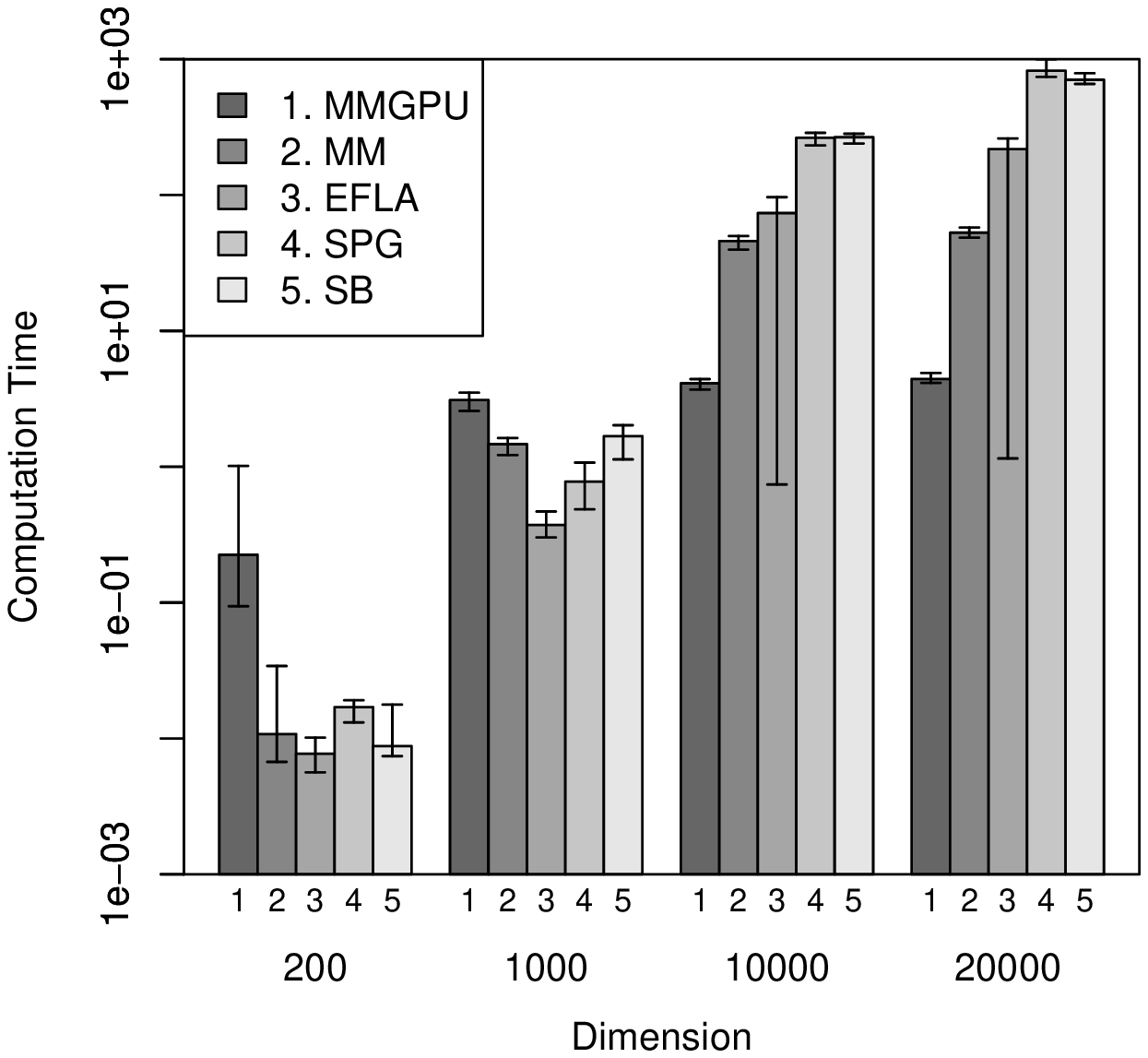,width=0.8\textwidth,
  height=0.21\textheight}}
  \centerline{(c) {\bf (C2)} with $(\lambda_1,\lambda_2) = (0.1,0.1)$}
\end{minipage}
\medskip
  \begin{minipage}[b]{.44\linewidth}
  \centering   \centerline{\epsfig{file=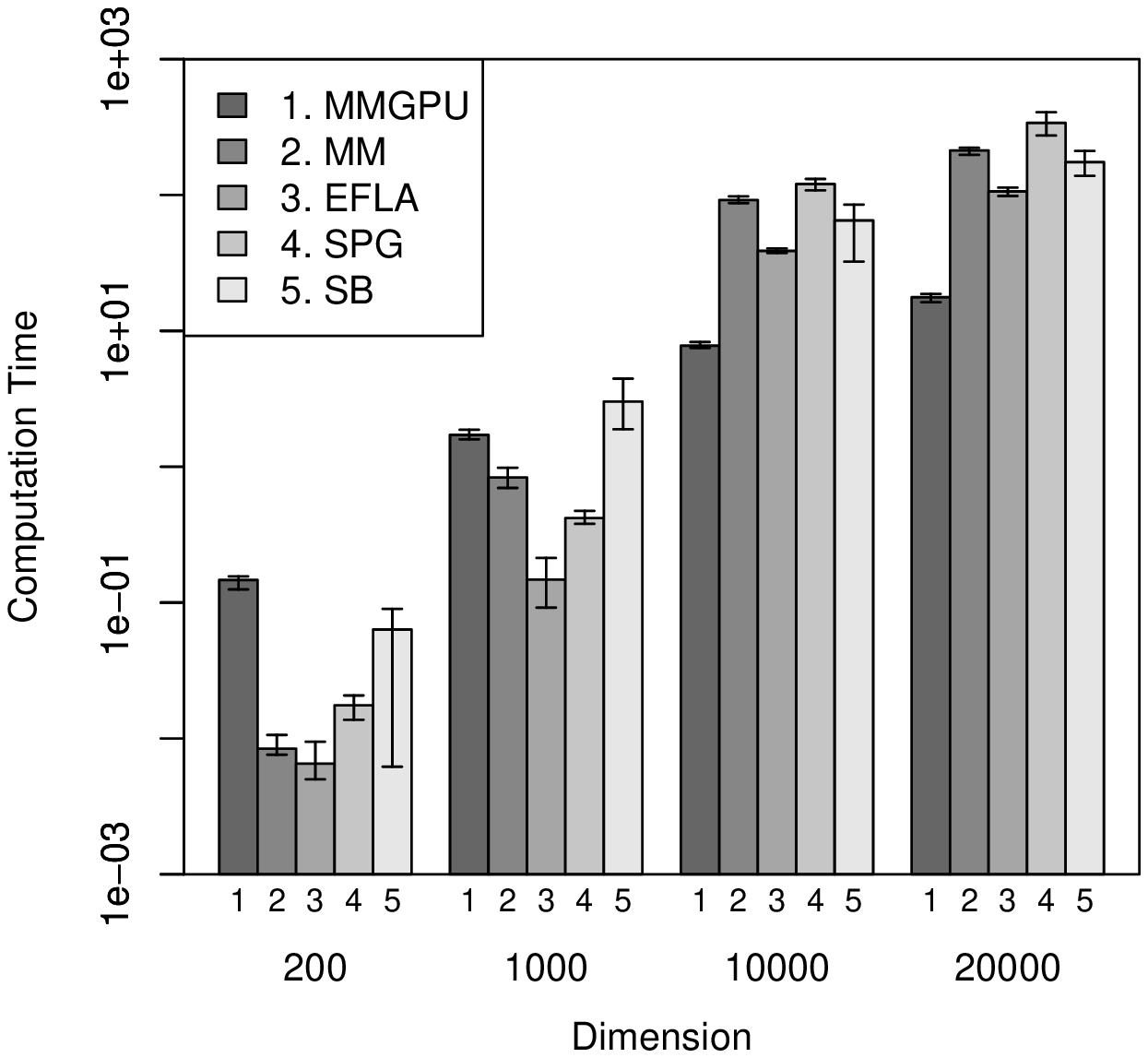,width=0.8\textwidth,
  height=0.21\textheight}}
  \centerline{(d) {\bf (C2)} with $(\lambda_1,\lambda_2) = (1,1)$}
\end{minipage}
\medskip
\begin{minipage}[b]{.44\linewidth}
  \centering   \centerline{\epsfig{file=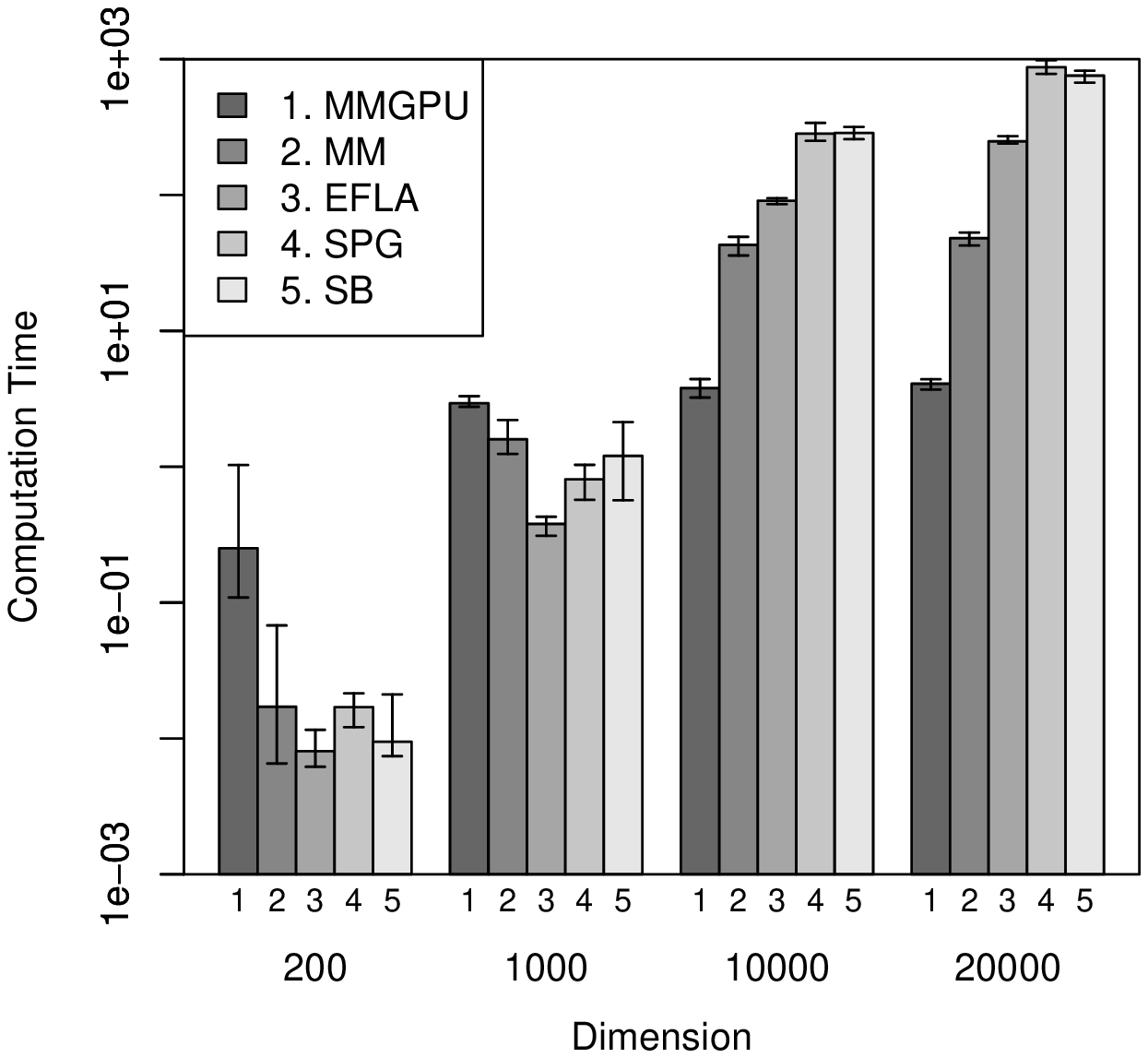,width=0.8\textwidth,
  height=0.21\textheight}}
  \centerline{(e) {\bf (C3)} with $(\lambda_1,\lambda_2) = (0.1,0.1)$}
\end{minipage}
\medskip
  \begin{minipage}[b]{.44\linewidth}
  \centering   \centerline{\epsfig{file=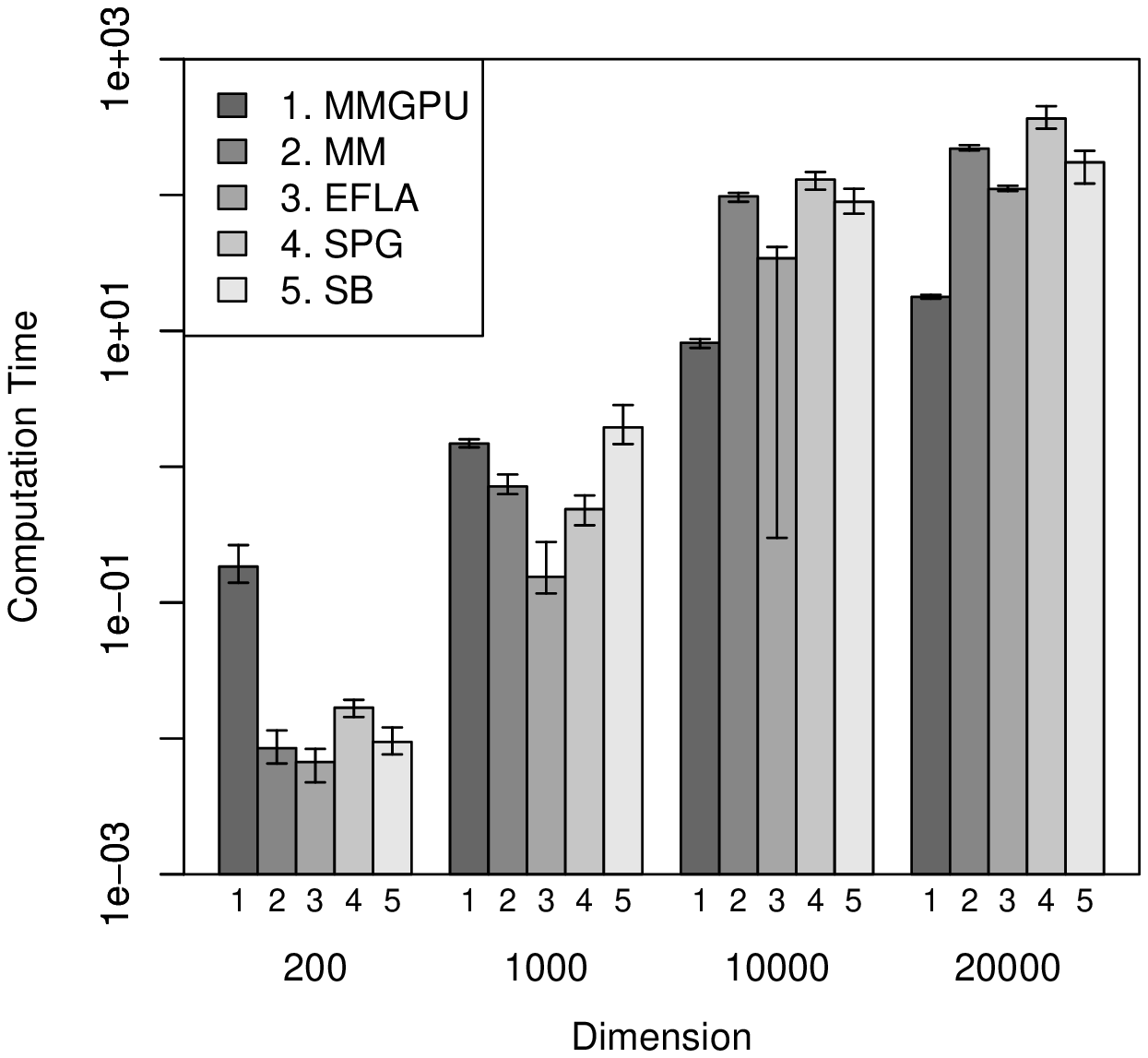,width=0.8\textwidth,
  height=0.21\textheight}}
  \centerline{(f) {\bf (C3)} with $(\lambda_1,\lambda_2) = (1,1)$}
\end{minipage}
\caption{Summary of the results of {\bf (C1)} -- {\bf (C3)} for $n=1000$.
The vertical lines indicate the range of computation times for 10 data sets.
Let $p$ denote the dimension of problem.
In {\bf (C2)} with $(0.1,0.1)$ at $p=10000,~ 20,000$ and {\bf (C3)} with $(1,1)$ at $p=10000$,
EFLA fails to converge the optimal solution at 
the lower bounds of vertical lines of EFLA (See Appendix C.) }
\label{res_c1}
\end{center}
\end{figure}

\subsection{Two-dimensional FLR}

We consider two scenarios for
the two-dimensional FLR, where the coefficients are
indexed on $q \times q$ lattice (hence $p = q^2$)
and the  penalty structure is imposed on the horizontally and vertically adjacent coefficients. 
The first is the two-dimensional FLR with general design matrix ${\bf X}$.
The second is image denoising, which is equivalent to the two-dimensional FLSA (${\bf X} = {\bf I}$).
Note that
EFLA and MMGPU are not applicable for these scenarios.
Instead, \textsf{pathFLSA}
is considered for the second scenario,
since it is known as one of the most efficient algorithms for the two-dimensional FLSA.

For the example of  the two-dimensional FLR problem with
general ${\bf X}$, we generate 10 data sets according to
the following scenario.

\begin{itemize}
\item[{\bf (C4)}] Two-dimensional FLR with general design matrix,
moderately sparse case.

We set the $q \times q$ dimensional matrix $B = (b_{ij})$ as follows.
\begin{equation} \nonumber
b_{ij}= \left\{
\begin{array}{lll}
2 & & {\rm for}~~ \frac{q}{4}k +1 \le i,j \le \frac{q}{4}(k+1), ~~k=0,1,2,3\\
-2 & & {\rm for}~~ \frac{q}{4}k +1 \le i \le \frac{q}{4}(k+1),~
\frac{q}{4}(3-k) +1 \le j \le \frac{q}{4}(4-k)~~ k=0,1,2,3\\
0 & & {\rm otherwise}
\end{array}\right.
\end{equation}
The true coefficient vector $\beta$ is 
the vectorization of the matrix $B$ denoted by $\beta = {\rm vec}(B)$.
The structure of the two-dimensional coefficient matrix $B$ is shown
in Figure \ref{tr_c4} (a).
\end{itemize}

As the previous scenarios, we generate ${\rm \bf X}$ and ${\rm \bf y}$
from equation (\ref{eqn:gener}) with true coefficients defined in (C4).
The candidate dimensions are $q= 16,~32,~64,~{\rm and}~128$ for sample size $n=1000$. 
We also consider four sets of regularization parameters
$(\lambda_1,\lambda_2)
\in \big\{ (0.1,0.1), (0.1,1),(1,0.1),(1,1) \big\}$ as in Section 5.1.

The computation times of (C4) are shown in Figure \ref{tr_c4} (b) and (c).
MM gains its competitiveness  competitive as
the dimension $p~(= q^2)$ increases.
MM is $1.6 \sim 19.5$ times faster than SB when $p >n$,
has similar performance to SPG when $p=10000$, and
MM is $1.3 \sim 12.3$ times faster than SPG when $p=20000$.
A detailed report on the computation time can be found in Appendix C,
Table C.4.
It can be seen that the
variation  in the number of iterations for both dimensions and values of $(\lambda_1,\lambda_2)$ is the smallest for
MM. Also it can be seen that MM tends to converge
within a few tens of iterations.
This stability is consistent with the observations from
$\bf (C1) - (C3)$.
\begin{figure}[htb!]
\begin{center}
\begin{minipage}[b]{.32\linewidth}
  \centering   \centerline{\epsfig{file=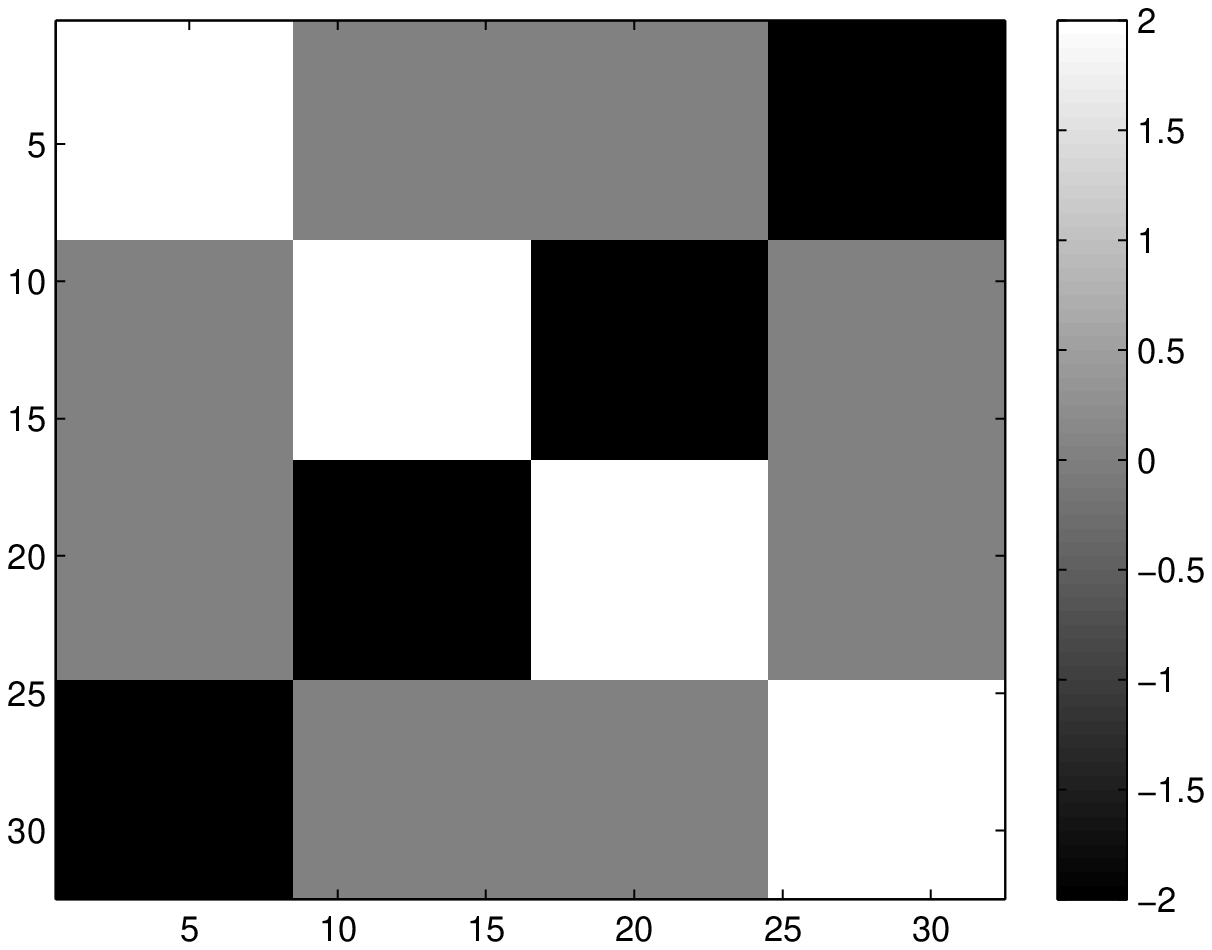,width=\textwidth,height=0.2\textheight}}
  \centerline{(a) True Coefficients}
\end{minipage}
\medskip
    \begin{minipage}[b]{.32\linewidth}
  \centering   \centerline{\epsfig{file=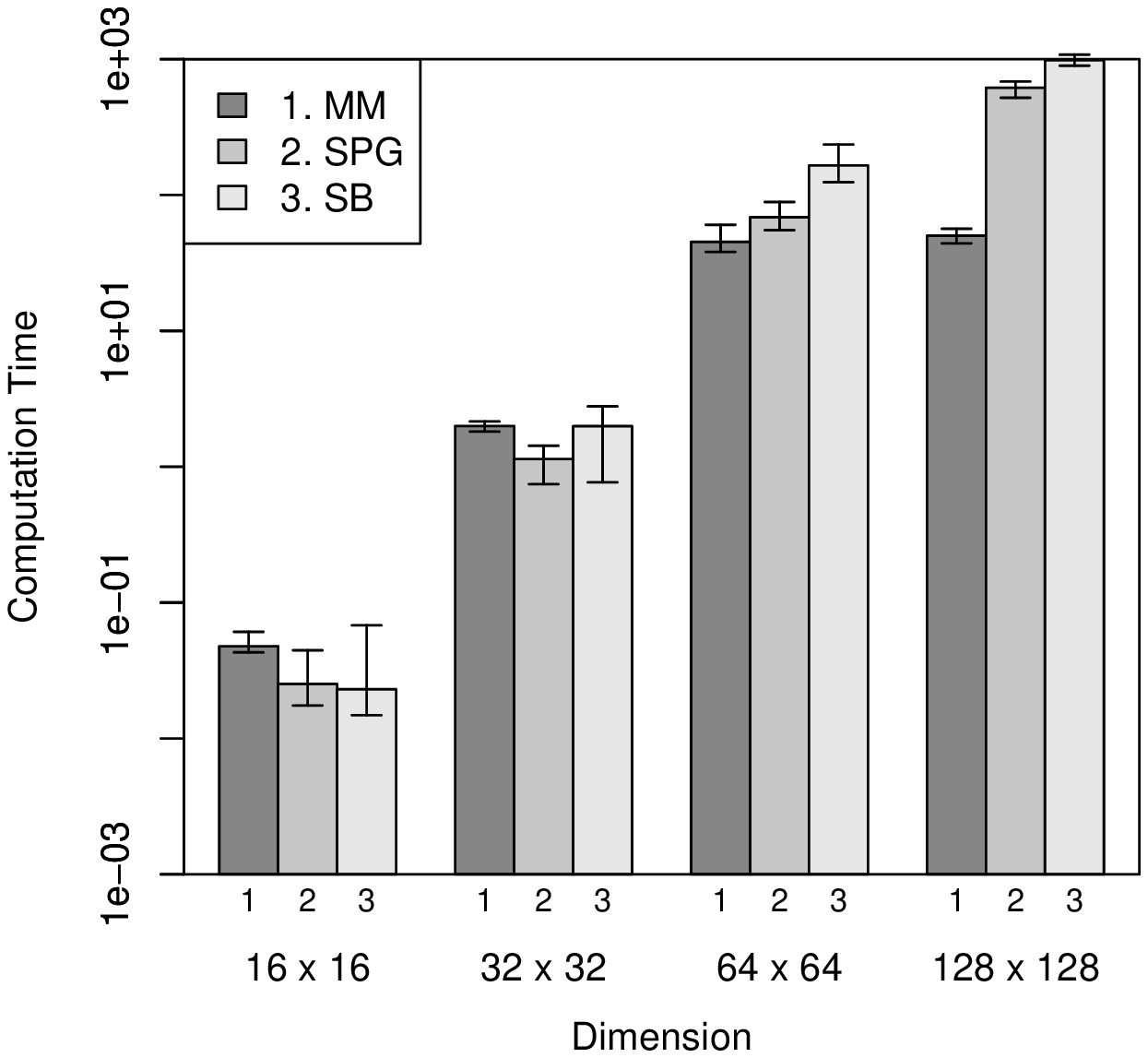,width=\textwidth,
  height=0.2\textheight}}
  \centerline{(b) $(\lambda_1,\lambda_2) = (0.1,0.1)$ }
 \end{minipage}
\medskip
  \begin{minipage}[b]{.32\linewidth}
  \centering   \centerline{\epsfig{file=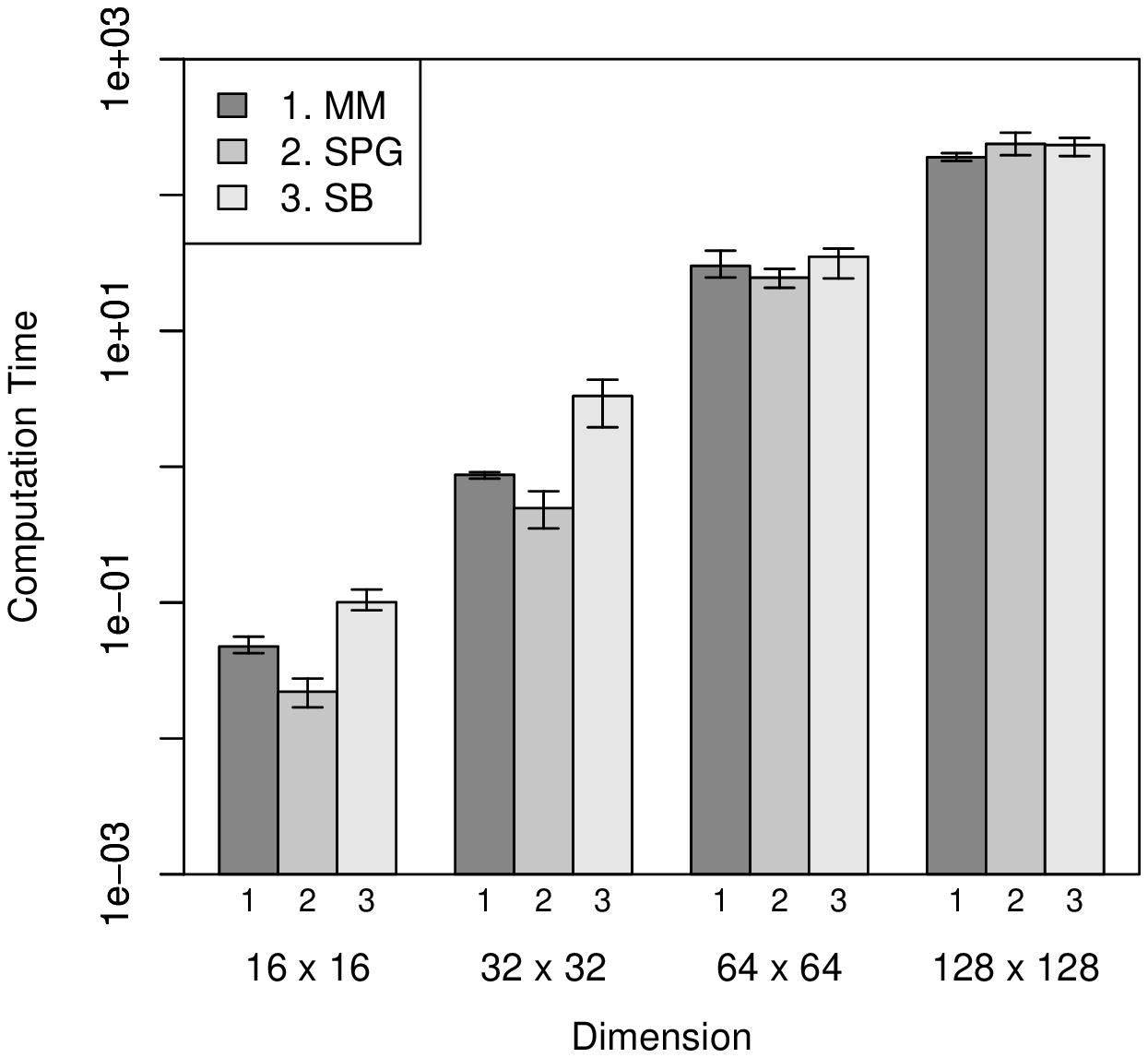,width=\textwidth,
  height=0.2\textheight}}
  \centerline{(d) $(\lambda_1,\lambda_2) = (1,1)$}
\end{minipage}
\caption{Summary of Case 4 {\bf (C4)} for sample size $n=1000$.}
 \label{tr_c4}
\end{center}
\end{figure}

\begin{itemize}
\item[{\bf (C5)}] Image denoising (${\bf X} = {\bf I}$)
 
We use a $256 \times 256$ gray scale
image (i.e., $q=256$, $p=q^2 = 65536$) of R.A. Fisher, as used in \cite{Friedman2007}.
After standardizing pixel intensities, we add Gaussian
noise with a standard variation $0.3$.
We set $\lambda_1 = 0$ because the application is image denoising.
\end{itemize}

We consider
 two regularization parameters $\lambda_2 =0.1, 1$.
  We report the average computation times and the average numbers of iterations in 
 Table \ref{summ_c5}. We present the true image, the noisy image, and the denoising results 
 at $\lambda_2 = 0.1$ in Figure \ref{image_res}.
Since the path for all the solutions for two-dimensional FLSA
is computationally infeasible, we terminate PathFLSA  at the
desired value of  $\lambda_2$. 
For  small $\lambda_2$ ($\lambda_2 = 0.1$),
MM is slower than SPG and PathFLSA, but much faster than SB.
As $\lambda_2$ increases, MM becomes faster than PathFLSA since
PathFLSA always has to start from $\lambda_2 = 0$.
In addition, SPG fails to obtain the solution at $\lambda_2 = 1$.
(See the Figure \ref{image_res_2}.)
MM also exhibits a smaller variation in the number of iterations and the objective function value
at the obtained solution than SPG and SB.
Again, this stability is a consistent behavior of MM throughout $\bf (C1)-(C5)$.

\begin{table}[htb!]
\caption{Summary of the computation time and
the numbers of iterations of the \textsf{MM}, \textsf{SPG}, \textsf{SB}, and \textsf{pathFLSA} algorithms for case ${\bf (C5)}$.
N/A denotes the method fails to obtain the optimal solution.}
\label{summ_c5}
\medskip
\begin{minipage}{\textwidth}
\centering
{\small
\begin{tabular}{|c|l|c|c|c|}\hline
$(\lambda_1, \lambda_2)$ &Method & Computation time (sec.)& \# of iterations$^*$\let\thefootnote\relax\footnotetext{$*$ Since PathFLSA is a path algorithm, the number of
	iterations of PathFLSA is omitted.} & $f(\widehat{\beta})$\\\hline
\multirow{4}*{$(0,0.1)$}	&	\textsf{MM}	&	21.2082 & 26	& 2810.16 	\\ 
	&	\textsf{SPG}	&	1.9964	&	140 & 2817.31 	 	\\  
	&	\textsf{SB}	&	217.2167	&	698 & 2822.40	 	\\  
	&	\textsf{pathFLSA}	&	3.2151	& -	& 2809.96	\\  \hline
	\multirow{4}*{$(0,1)$}	&	\textsf{MM}	&	48.7812 	& 56 & 6920.43	\\ 
	&	\textsf{SPG}$^\vartriangle$\let\thefootnote\relax\footnotetext{$\vartriangle$ SPG
	stops after 11 iterations with $f(\widehat{\beta})= 35595.82$. }	&	N/A & N/A	& N/A	\\   
	&	\textsf{SB}	&	635.8651 & 4431 & 7308.07		\\  
		&	\textsf{pathFLSA}	&	359.2660 &	 - & 6919.06  	\\  \hline
	\end{tabular}
	}
\end{minipage}
\end{table}

\begin{figure}[htb!] 
\begin{center}
  \begin{minipage}[b]{.48\linewidth}
  \centering   \centerline{\epsfig{file=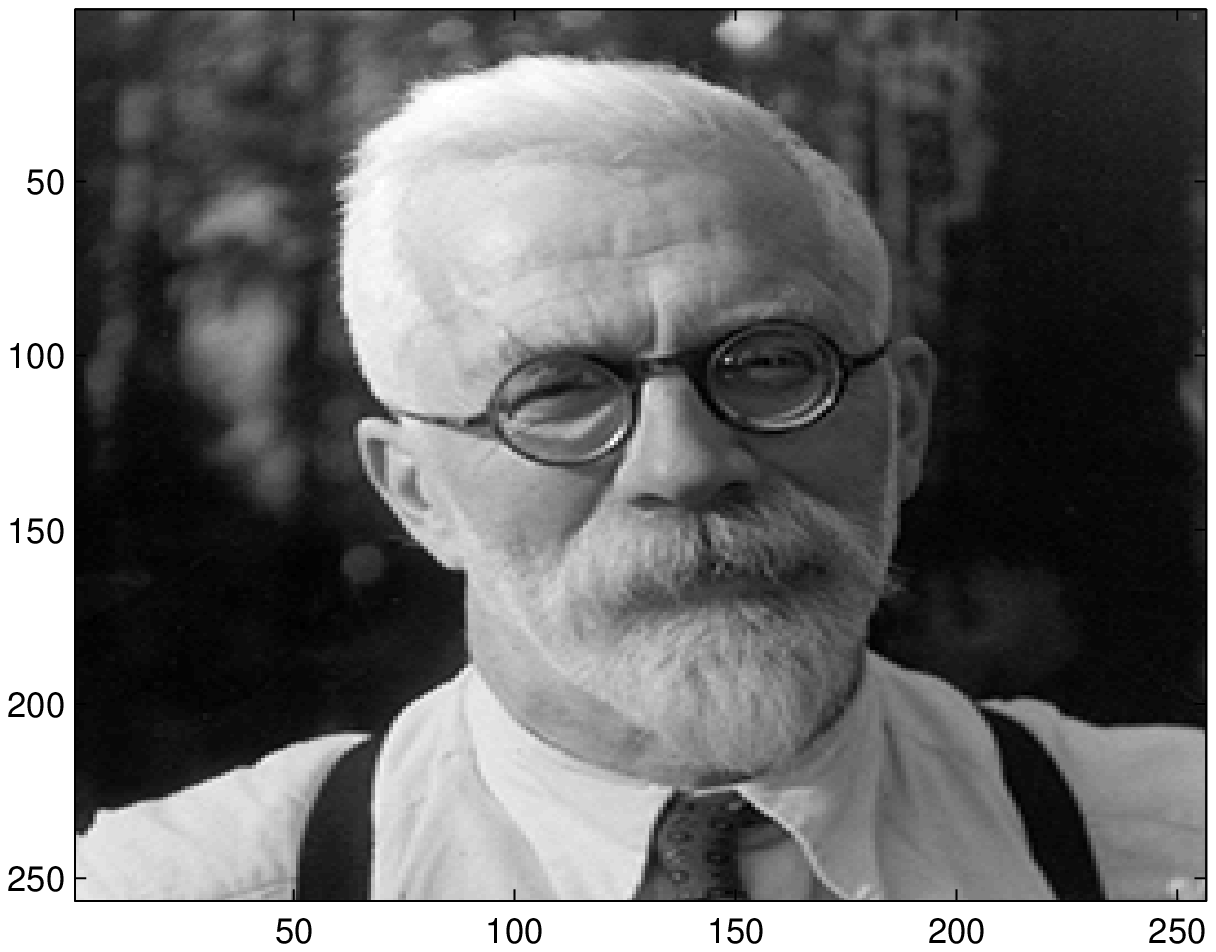,width=0.8\textwidth,height=0.24\textheight}}
  \centerline{(a) True image}
\end{minipage}
\medskip
  \begin{minipage}[b]{.48\linewidth}
  \centering   \centerline{\epsfig{file=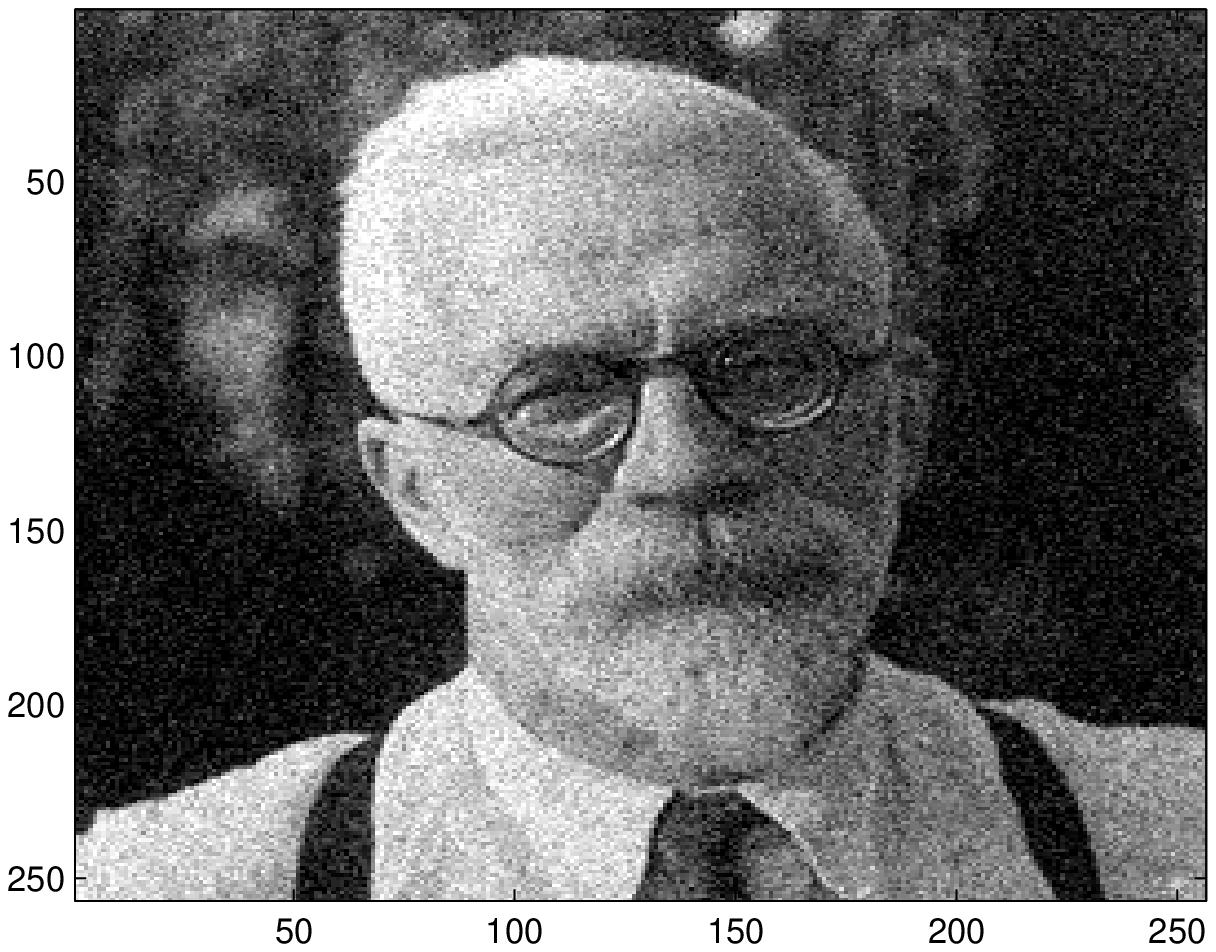,width=0.8\textwidth,height=0.24\textheight}}
  \centerline{(b) Noisy image}
\end{minipage}
\medskip

  \begin{minipage}[b]{.48\linewidth}
  \centering   \centerline{\epsfig{file=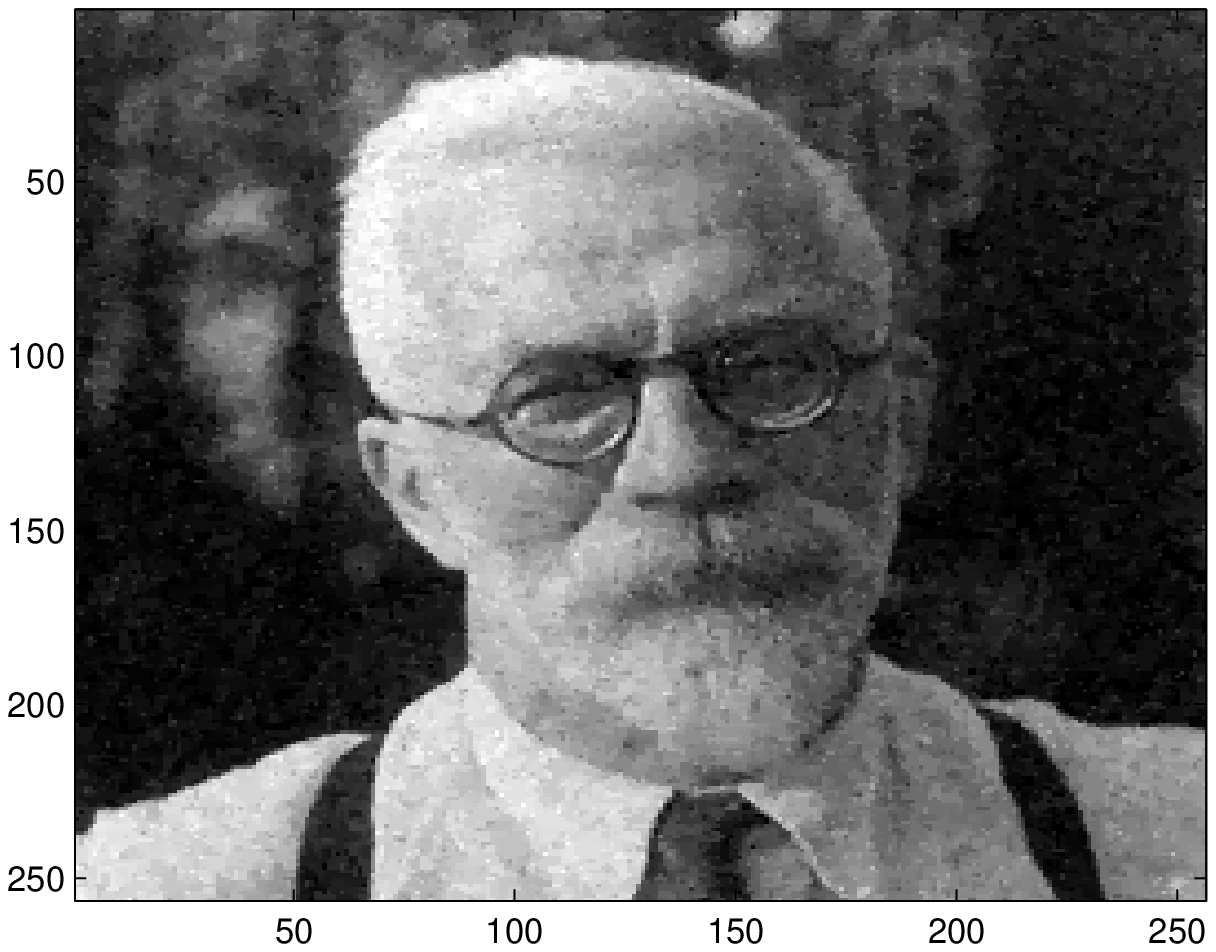,width=0.8\textwidth,height=0.24\textheight}}
  \centerline{(c) MM}
\end{minipage}
\medskip
\begin{minipage}[b]{.48\linewidth}
  \centering   \centerline{\epsfig{file=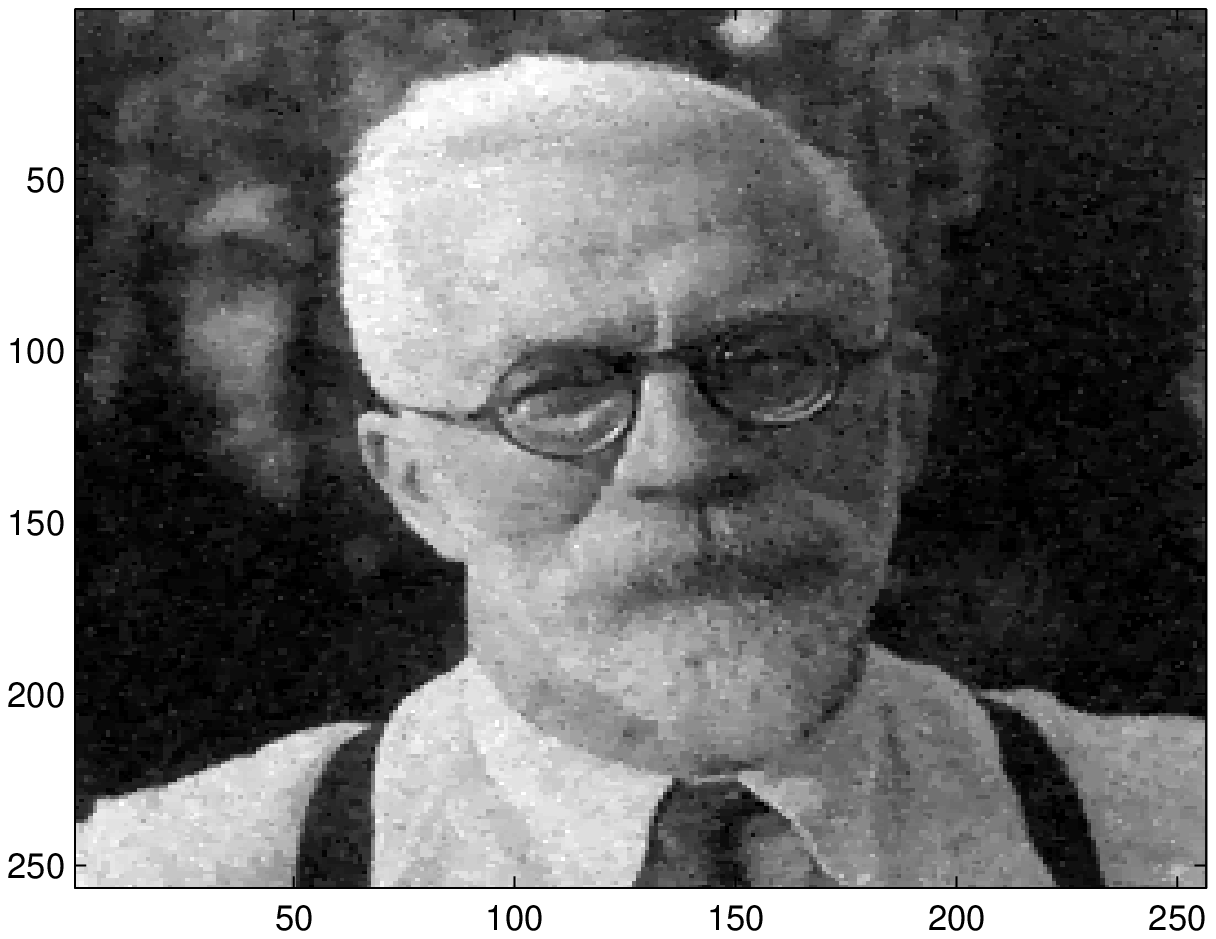,width=0.8\textwidth,height=0.24\textheight}}
  \centerline{(d) SPG}
\end{minipage}
\medskip
\begin{minipage}[b]{.48\linewidth}
  \centering   \centerline{\epsfig{file=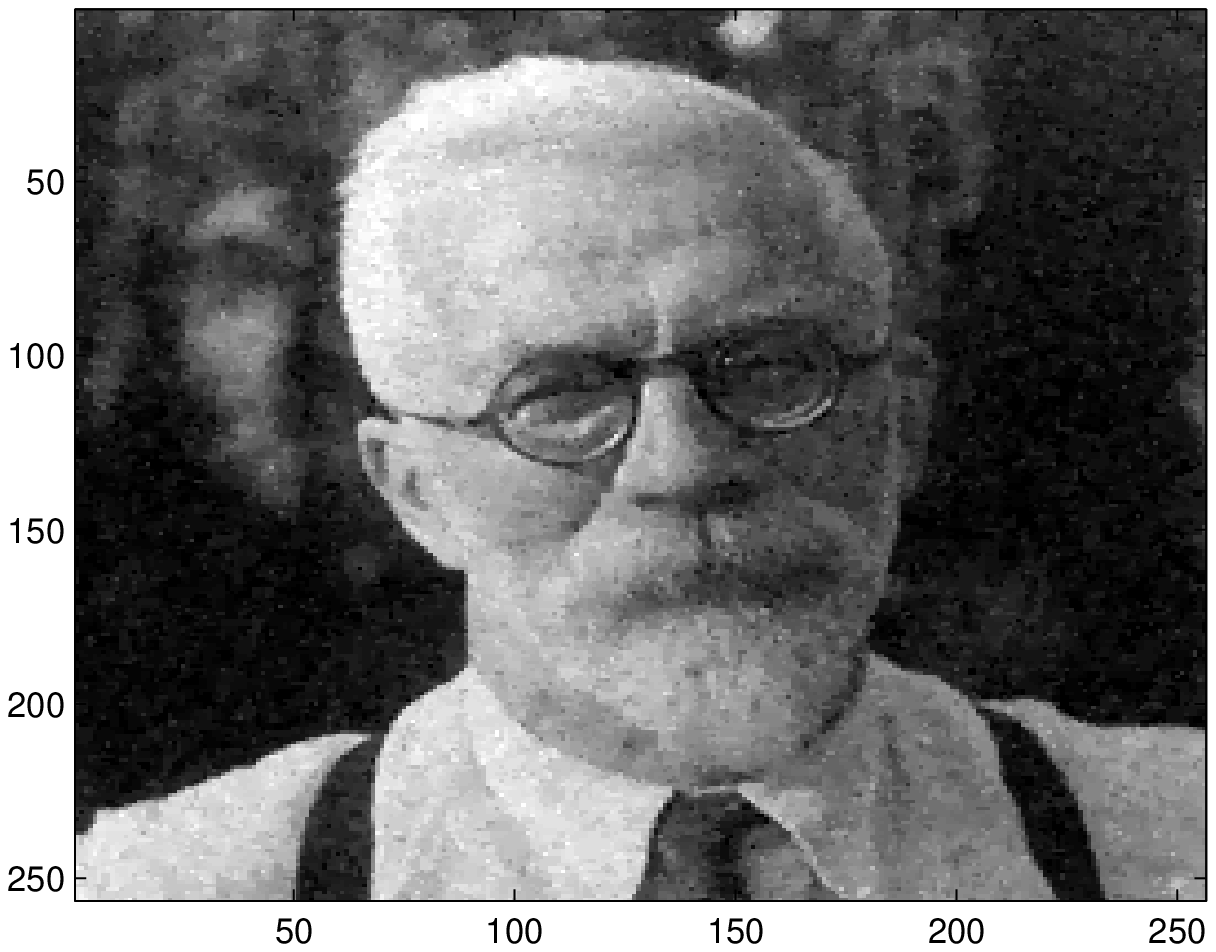,width=0.8\textwidth,height=0.24\textheight}}
  \centerline{(e) SB}
\end{minipage}
\medskip
    \begin{minipage}[b]{.48\linewidth}
  \centering   \centerline{\epsfig{file=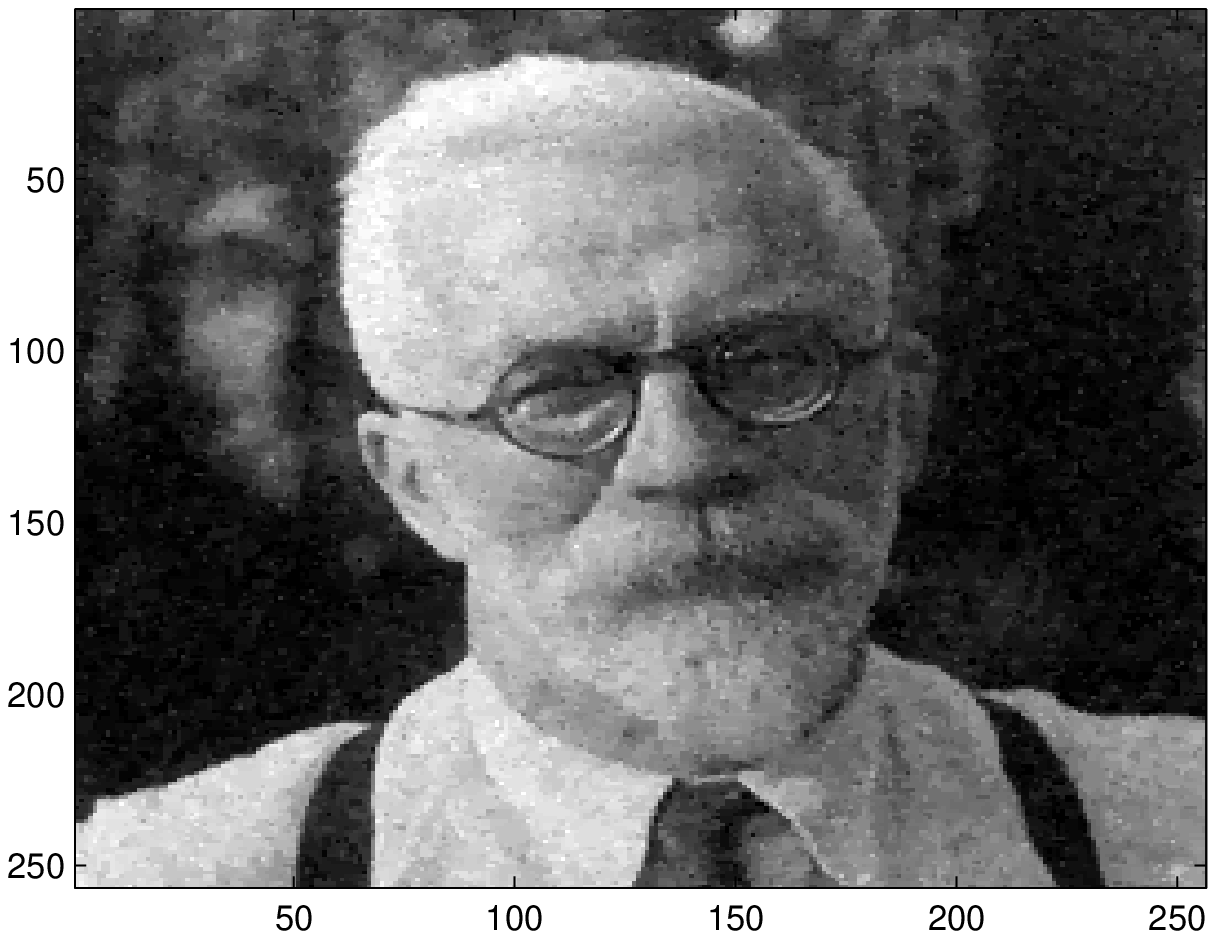,width=0.8\textwidth,height=0.24\textheight}}
  \centerline{(f) PathFLSA}
\end{minipage}

\caption{Image denoising with various algorithms for the FLR ($\lambda_2 = 0.1$)}
\label{image_res}
\end{center}
\end{figure}

\begin{figure}[htb!] 
\begin{center}
  \begin{minipage}[b]{.48\linewidth}
  \centering   \centerline{\epsfig{file=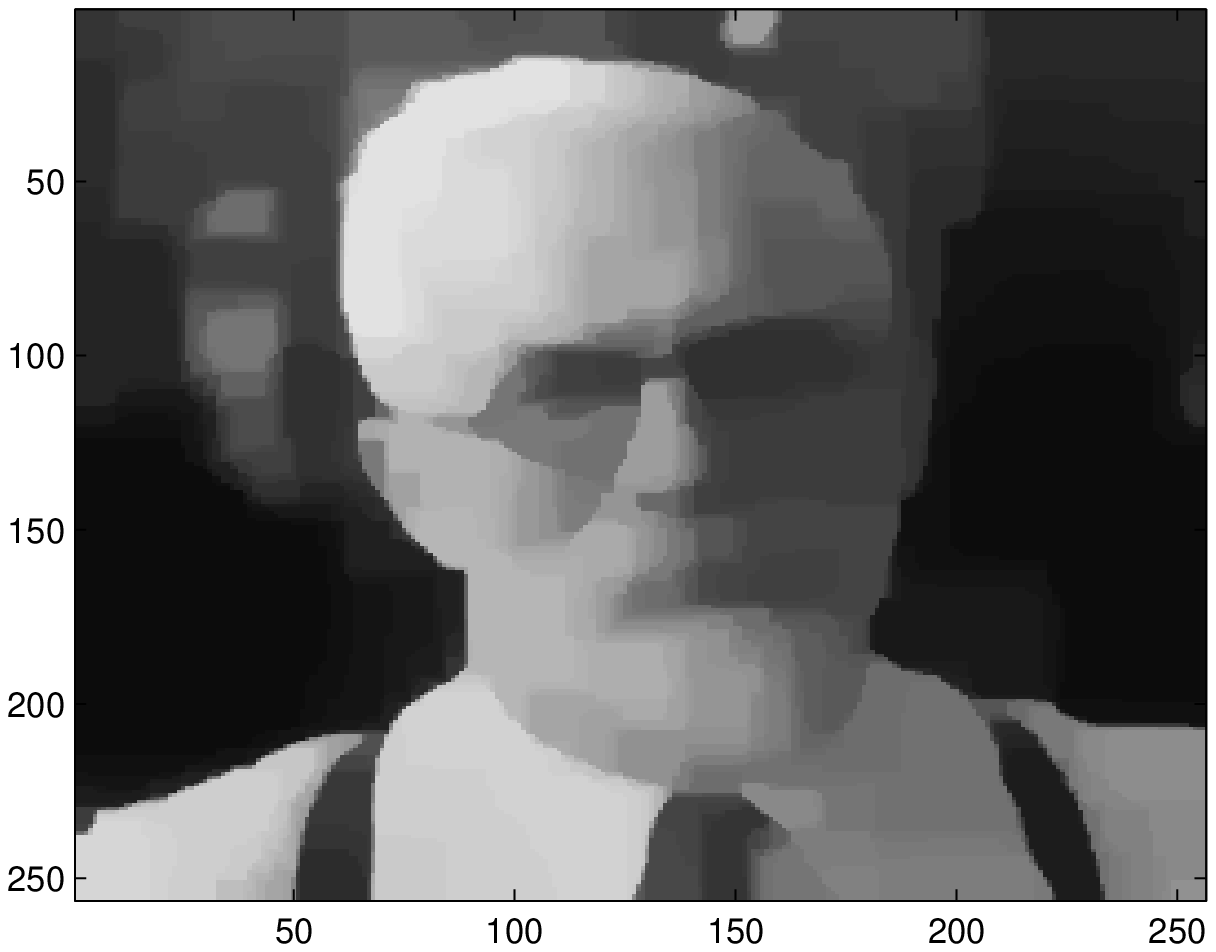,width=0.8\textwidth,height=0.24\textheight}}
  \centerline{(a) MM}
\end{minipage}
\medskip
\begin{minipage}[b]{.48\linewidth}
  \centering   \centerline{\epsfig{file=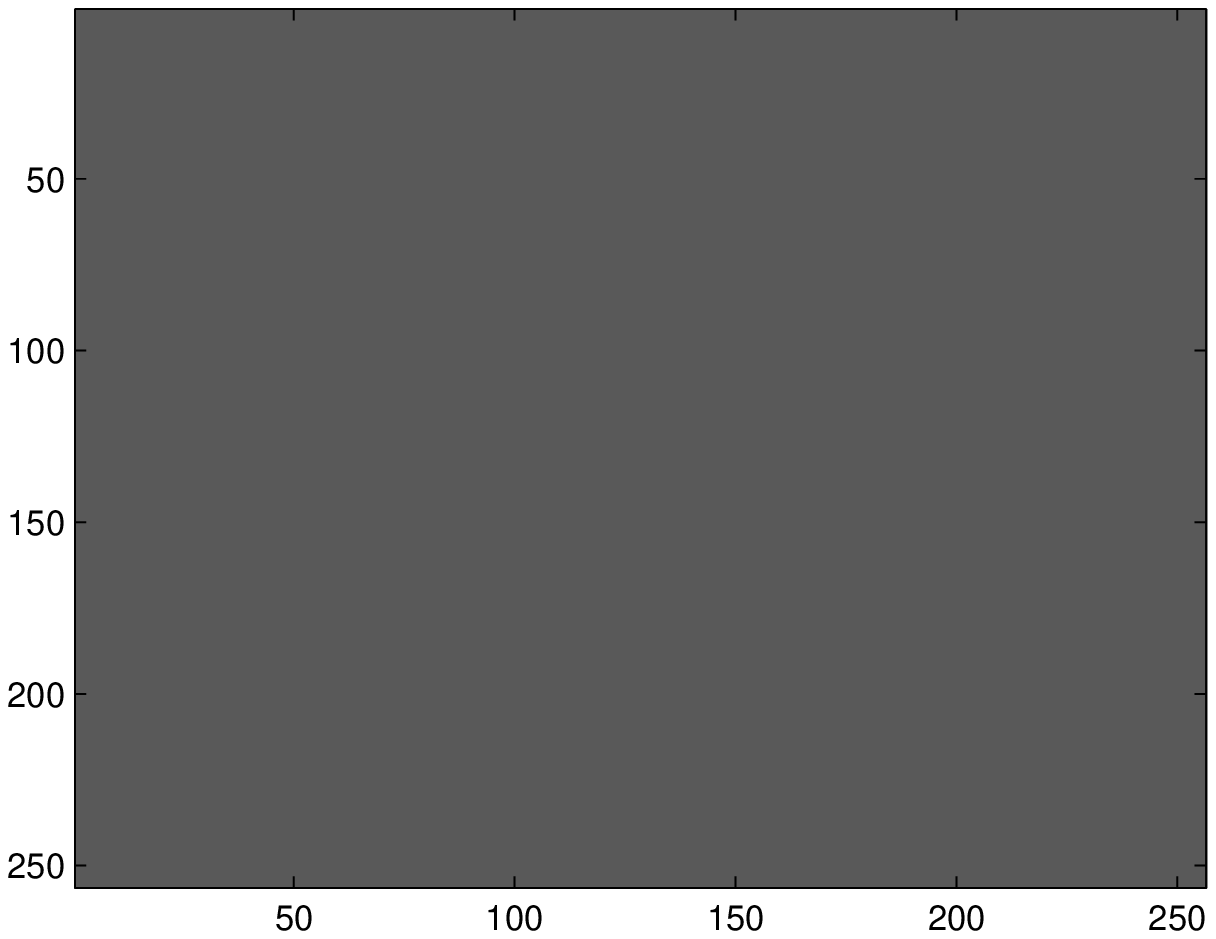,width=0.8\textwidth,height=0.24\textheight}}
  \centerline{(b) SPG}
\end{minipage}
\medskip
\begin{minipage}[b]{.48\linewidth}
  \centering   \centerline{\epsfig{file=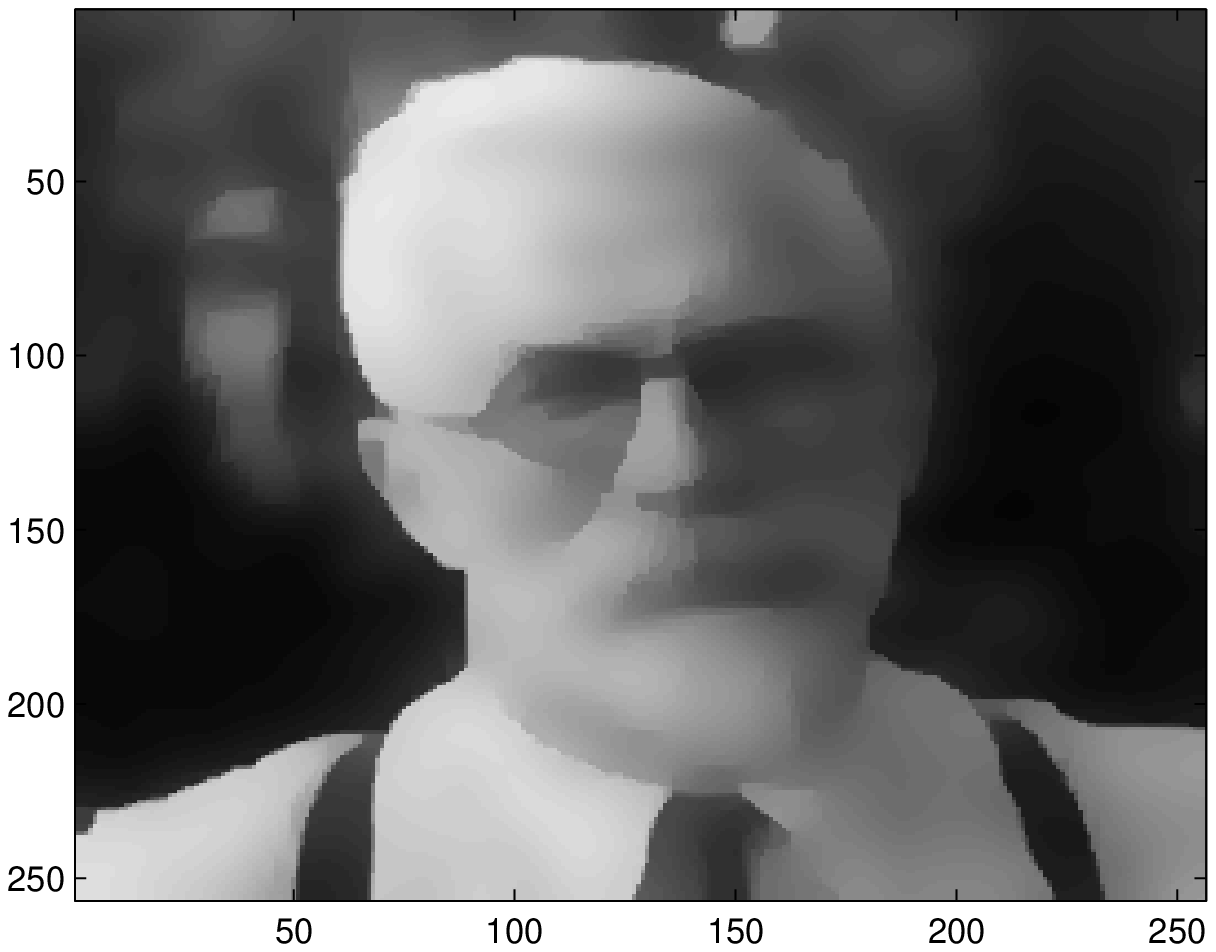,width=0.8\textwidth,height=0.24\textheight}}
  \centerline{(c) SB}
\end{minipage}
\medskip
    \begin{minipage}[b]{.48\linewidth}
  \centering   \centerline{\epsfig{file=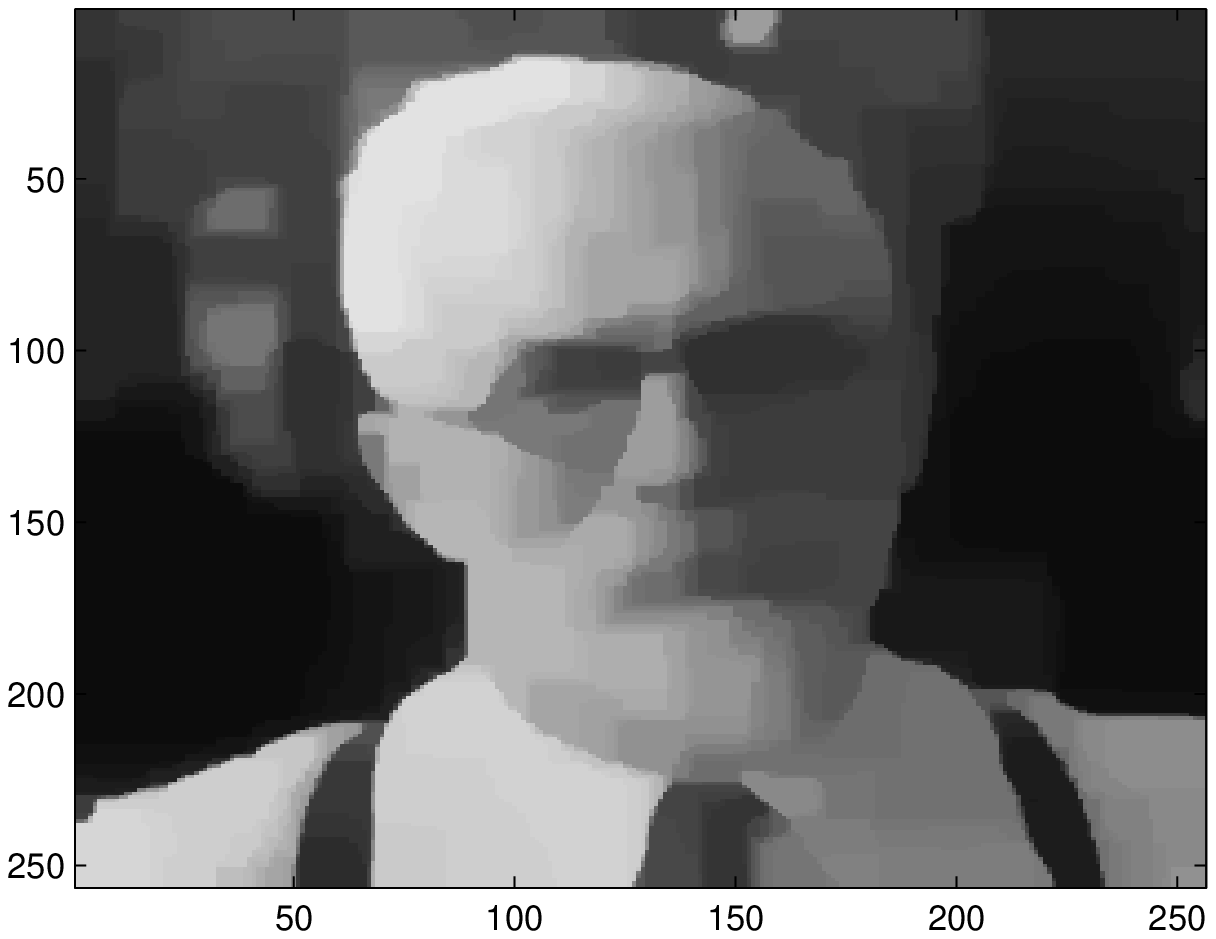,width=0.8\textwidth,height=0.24\textheight}}
  \centerline{(d) PathFLSA}
\end{minipage}

\caption{Image denoising with various algorithms for the FLR ($\lambda_2 = 1$). SPG fails to obtain the optimal solution.}
\label{image_res_2}
\end{center}
\end{figure}

\section{Conclusion}

In this paper, we have proposed an MM algorithm for the
fused lasso regression (FLR) problem with the generalized
fusion penalty. The proposed algorithm
is flexible in the sense that
it can be applied to the FLR with a wide class of
design matrices and penalty structures.
It is stable in the sense that the convergence
of the algorithm is not sensitive to the dimension of the 
problem, the choice of the regularization parameters, and
the sparsity of the true model.  
Even when a special structure on the design matrix
or the penalty is imposed, the MM algorithm shows  comparable
performance to the algorithms tailored to the
special structure.
These features make the proposed algorithm 
an attractive choice as an off-the-shelf FLR algorithm.
Moreover, the performance of
the MM algorithm  can be improved by parallelizing
it with GPU, when the standard fused lasso penalty is imposed
and the dimension increases.
Extension of the GPU algorithm to the two-dimensional FLR problem
is our future direction of research.

\section*{Acknowledgement}

We would like to thank Dr. Gui-Bo Ye for providing the MATLAB code for
the Split Bregman algorithm.

{}

\appendix

\section{Proofs}
\subsection{Proof of Proposition 1}\label{sec:proofs:prop1}

For part (i), write
\[
f_{\eps} (\beta) = \displaystyle \frac{1}{2} \| {\bf y}- {\bf X} \beta \|_2^2  
+ \lambda_1 \sum_{j=1}^p a_j (\beta)  
 + \lambda_2 \sum_{(j,k) \in E} b_{jk} (\beta),
\]
where
\begin{align*}
a_j(\beta) &= |\beta_j | - \eps \log \Big(1 + \frac{|\beta_j|}{\eps} \Big) \\
b_{jk}(\beta)&= |\beta_j - \beta_{k}| - \eps \log \Big(1 + \frac{|\beta_j - \beta_{k}|}{\eps} \Big).
\end{align*} 
for $j,k=1,2,\ldots,p$.
The functions $a_j(\beta)$ and $b_{jk}(\beta)$ are continuous and convex in $\beta \in \mathbb{R}^p$ and so is $f_{\eps} (\beta)$.  
In addition, for each $j=1,2,\ldots,p$, we see that $\lim_{|\beta_j| \rightarrow  \infty} f_{\epsilon} (\beta) = \infty$
 held fixed. This, along with the convexity of $f_{\epsilon}(\beta)$, shows that $\lim_{ \|\beta\|_2 \rightarrow \infty} f_{\epsilon} (\beta) = \infty$.

For part (ii), recall that
\begin{align*}
0 \le f(\beta) - f_\eps(\beta) & =  \lambda_1 \sum_{j=1}^p \eps\log \left(1+ \frac{|\beta_j|}{\eps} \right)  
	+  \lambda_2 \sum_{(j,k)\in E} \eps \log \left(1+ \frac{|\beta_j-\beta_{k}|}{\eps} \right) \\
 & \le   \lambda_1 \sum_{j=1}^p \eps\log \left(1+ \frac{\sup_{\beta \in \bf{C}} \big| \beta_j \big|}{\eps} \right)  + \lambda_2 \sum_{(j,k)\in E} \eps\log \left(1+ \frac{\sup_{\beta \in \bf{C}} \big| \beta_j-\beta_{k}\big|}{\eps} \right).
\end{align*}
The suprema in the rightmost side are achieved and finite because $\bf C$ is compact. 
Then the rightmost side monotonically decreases to $0$ as $\eps$ goes to $0$.

To see part (iii), for a given $\epsilon>0$ and $x_0 \in \mathbb{R}$, consider the functions
\begin{eqnarray}
p_\eps(x) &=& |x|-\eps \log \big(1+ \frac{|x|}{\eps} \big)
\nonumber \\
q_\eps(x|x^0) &=& |x^0|-\eps \log \bigg(1+ \frac{|x^0|}{\eps} \bigg) + 	
\frac{x^2 - (x^0)^2}{2(|x^0|+\eps)}, \nonumber
\end{eqnarray}
and their difference $ h(y |y^0 ) = q_\eps(x|x^0) - p_\eps(x)$, where $y=|x|$ and $y_0=|x_0|$. The function $h(y|y^0)$ is twice differentiable, and simple algebra shows that it is increasing for $y \ge y_0$ and
decreasing for $y < y_0$. We
have
\[
h\big(y \big| y^0 \big) = q_\eps(x|x^0) - p_\eps(x) \ge h \big(y^0 \big| y^0 \big)=0,
\]
where the equality holds if and only if $y=y^0$ (equivalently, $|x|=|x^0|$). This shows that
$q_\eps(x|x^0)$ majorizes $p_\eps(x)$ at $x =x^0$, thus
$g_\eps(\beta|\beta')$ majorizes $f_\eps(\beta)$ at $\beta =\beta'$.
Finally, the majorizing function $g_\eps(\beta|\beta')$ is strictly convex for fixed $\beta'$ since
\begin{equation} \nonumber
\frac{\partial^2 g_\eps(\beta|\beta')}{\partial \beta^2} = {\bf X}^{T} {\bf X} + \lambda_1 {\bf A}' + \lambda_2 {\bf B}'
\end{equation}
is positive definite,  where ${\bf A}'$ and ${\bf B}'$ are defined
in (\ref{eqn:lin-sys}). 
Hence it has a unique global minimum.

\subsection{Proof of Lemma \ref{lemma:perturbed}}\label{sec:proofs:lemma_perturbed}

Define the level set 
\[
\Omega \big( \epsilon, c \big) = \big\{ \beta \in \mathbb{R}^p \big|
f_\eps \big( \beta \big) \le c \big\},
\]
for every $\epsilon >0$ and $c>0$. Similarly define 
\[
\Omega \big( 0,  c \big) = \big\{ \beta \in \mathbb{R}^p \big|
f \big( \beta \big) \le c \big\}.
\]
Then, from part (i) of Proposition 1, $\Omega \big( \eps, c \big)$ is compact for every $\epsilon>0$ and $c>0$. 
Compactness of $\Omega \big( 0, c \big)$ follows from the uniform convergence of $f_{\eps}(\beta)$ to $f(\beta)$.

Note that
$\Omega\big(\eps, f( \widehat{\beta}) \big)$ is a non-empty 
compact set since by the assumption that the solution set $\{ \beta ~|~ f(\beta) = f(\widehat{\beta})\}$ is non-empty.
Then we have  $f_{\epsilon}(\beta) \le f(\widehat{\beta})$ for every $\beta \in \Omega\big(\eps, f(\widehat{\beta}) \big)$.
By construction, $f_{\eps} (\widehat{\beta}_{\eps}) \le f_{\eps}(\beta)$.
Thus,
\[
\min_{\gamma} f_{\eps} (\gamma)=f_{\eps} ( \widehat{\beta}_\eps ) \le f_{\epsilon} (\beta ) \le f \big(\widehat{\beta}\big) = \min_{\gamma} f\big(\gamma \big).
\]
In other words, $\widehat{\beta}_\eps \in \Omega\big(\eps, f( \widehat{\beta})\big)$. Similarly we see $\widehat{\beta} \in \Omega\big(0, f( \widehat{\beta})\big)$ that is non-empty.

Since $f(\widehat{\beta}_\eps) \ge f(\widehat{\beta})$ by the definition of $\widehat{\beta}$,
\begin{align*}
	  	0 \le f \big(\widehat{\beta}_\eps \big) - f \big(\widehat{\beta} \big) & \le 
	  	f \big(\widehat{\beta}_\eps \big) - f_\eps \big(\widehat{\beta}_\eps \big) + f_\eps \big(\widehat{\beta} \big) - f \big(\widehat{\beta} \big) \nonumber \\
	  	& \le   \big|f \big(\widehat{\beta}_\eps \big) - f_\eps\big(\widehat{\beta}_\eps \big) \big|
 + \big|f_\eps \big(\widehat{\beta} \big) - f\big(\widehat{\beta} \big) \big|. 
\end{align*}
The rightmost side of the above inequality goes to zero because $f_{\eps}(\beta)$ converges to $f(\beta)$ uniformly on both $\Omega\big(\eps,f(\widehat{\beta})\big)$ and $\Omega\big(0,f(\widehat{\beta})\big)$.
Then, for a limit point $\beta^*$ of the sequence $\{ \widehat{\beta}_{\eps_n} \}_{n\ge 1}$ with $\eps_n \downarrow 0$, we see
\[
\lim_{n \rightarrow \infty} f \big( \widehat{\beta}_{\eps_n}\big) = f\big( \beta^* \big) = f \big(\widehat{\beta} \big)=
\min_{\gamma} f (\gamma )
\]
by the continuity of $f(\beta)$, i.e., $\beta^*$ minimizes $f(\beta)$.

\subsection{Proof of Lemma \ref{lemma:lyapunov}}\label{sec:proofs:lemma_lyapunov}

We first claim that the recursive updating rule implicitly defined by the MM algorithm is continuous.

\begin{claim}\label{claim:continuity}
The function $M(\alpha) = \argmin_{\beta} g_{\eps}(\beta | \alpha)$
is continuous in $\alpha$,
where $g_\eps(\beta|\alpha) $ is the proposed majorizing function of $f_\eps(\beta)$ at $\alpha$.
\end{claim}

\begin{proof}
Consider a sequence $\{\alpha_k\}_{k\ge 0}$ converging to $\tilde{\alpha}$ with finite $f_\eps(\tilde{\alpha})$ and satisfying
$f_{\eps}(\alpha_k) \le f_{\eps}(\alpha_0) < \infty$ for $k\ge 1$.
We want to show that $M(\alpha_k)$ converges to $M(\tilde{\alpha})$. 

The descent property of the MM algorithm states that
\[
f_\eps\big( M(\alpha_k)\big) \le g_{\eps} \big( M(\alpha_k ) \big| \alpha_k \big) \le g_{\eps} \big( \alpha_k \big| \alpha_k \big)
= f_\eps\big(\alpha_k\big)
\]
for every $k$.
Thus the sequence $\big\{ M(\alpha_k ), k=0,1,2,\ldots \big\}$ is a subset of the level set $\Omega \big( \epsilon, f_{\eps} (\alpha_0) \big) = \{ \beta \in \mathbb{R}^p \big| f_{\eps} (\beta ) \le f_{\eps} (\alpha_0) \}$.

Since the level set  $\Omega \big( \epsilon,  f_{\eps} (\alpha_0) \big)$ is compact (see Section \ref{sec:proofs:lemma_perturbed}), 
the sequence $\{ M(\alpha_k)\}_{k=0}$  is bounded. Hence we can find 
a convergent subsequence $\{ M(\alpha_{k_n} )\}_{n\ge 1}$; denote its limit point by $\tilde{M}$. Then, by the continuity of $g_{\eps} (\beta \big| \alpha )$ in both $\beta$ and $\alpha$, for all $\beta$,
\[
g_{\eps}( \tilde{M} \big| \tilde{\alpha} ) = \overline{\lim}_{n \rightarrow \infty}
 g_{\eps}( M (\alpha_{k_n} ) | \alpha_{k_n} )
\le  \overline{\lim}_{n \rightarrow \infty} g_{\eps}( \beta | \alpha_{k_n} )
=g_{\eps}(\beta \big| \tilde{\alpha} ).
\]
Thus $\tilde{M}$ minimizes $g_{\eps} \big(\beta \big| \tilde{\alpha} \big)$. Since the minimizer of
$g_{\eps} \big(\beta \big| \alpha\big)$ is unique by Proposition \ref{prop:mm}, we have $\tilde{M} = M \big(\tilde{\alpha} \big)$, i.e.,
$M(\alpha_{k_n}\big)$ converges to $M(\tilde{\alpha})$.	
\end{proof}

We then state a result on the global convergence of the updating rule.
\begin{claim}\label{claim:global}
(Convergence Theorem A, \citet[p.91]{Zangwill1969})
Let the point-to-set map $M: X \rightarrow X$
determine an algorithm that given a point
$x_1 \in X$ generates the sequence $\{ x_k\}_{k=1}^\infty$.
Also let a solution set $\Gamma \subset X$ be given.
Suppose that: (i) all points $x_k$ are in a compact set $C \subset X$;
(ii) there is a continuous function $u : X \rightarrow \mathbb{R}$
such that (a) if $x \notin \Gamma$, $u(y) > u(x)$
for all $y \in M(x)$, and (b) if $x \in \Gamma$,
then either the algorithm terminates or 
$u(y) \ge u(x)$ for all $y \in M(x)$;
(iii) the map $M$ is closed at $x$ if $x \notin \Gamma$.
Then either the algorithm stops at a solution, or the limit of
any convergent subsequence is a solution.
\end{claim}
\noindent 
Closedness of point-to-set maps is an extension of continuity of point-to-point maps: a point-to-set map $A:V \rightarrow V$ on a set $V$ is said to be closed if (a) $z^k \rightarrow z^{\infty}$, (b) $y^k \in A(z^k)$, and (c) $y^k \rightarrow y^{\infty}$ imply $y^{\infty} \in A(z^{\infty})$. The map $A$ is closed on a set $X \subset V$ if it is closed at each $z \in X$ \citep[p.88]{Zangwill1969}.
The proof of this claim is similar to \cite{Vaida2005}.

\bigskip
Finally we are ready to show the main result of this section.
\begin{proof}[Proof of Lemma \ref{lemma:lyapunov}]
We apply Claim \ref{claim:global} with 
$M(\alpha) = \argmin_\beta g_{\eps}(\beta|\alpha)$,
$\Gamma = \{ \beta \,|\, 0 \in \partial_{\beta} f_\eps(\beta)  \}$,
and $u(\beta) = - f_\eps(\beta)$, where $\partial_{\beta} f_\eps(\beta)$ is the 
sub-differential of $f_{\eps}(\beta)$ at $\beta$.  
We check the three conditions for Claim \ref{claim:global} as follows.
Define the level set of $f_\eps(\beta)$ 
as $\Omega \big( \epsilon, f_{\eps} (\widehat{\beta}^{(0)}) \big) = \big\{ \beta \in \mathbb{R}^p \big| f_{\eps} (\beta ) \le f_{\eps} (\widehat{\beta}^{(0)})  \big\}$, which is compact (see Section \ref{sec:proofs:lemma_perturbed}).
Since $f_{\eps}(\widehat{\beta}^{(r+1)})  \le f_{\eps}(\widehat{\beta}^{(r)})$
for every $r = 0,1,2,\ldots$,
$\{ \widehat{\beta}^{(r)}\}_{r\ge 0}$ is a 
subset of $ \Omega \big( \epsilon, f_{\eps} (\widehat{\beta}^{(0)}) \big)$ (see Section \ref{sec:proofs:lemma_perturbed})
and condition (i) of Claim \ref{claim:global} is satisfied.
Lemma \ref{lemma:perturbed} states that $M(\alpha)$ is a continuous point-to-point map in $\alpha$. 
That is, $M(\alpha)$ is a closed map on any compact set of $\mathbb{R}^p$, and so on $\Gamma^c$.
Thus condition (iii) of Claim \ref{claim:global} is met.
Finally condition (ii) is guaranteed by the  strict convexity of $g_{\eps}(\beta|\alpha)$
(see Proposition \ref{prop:mm}) and the definition of $M(\alpha)$.
Therefore all the limit points of $\{ \widehat{\beta}^{(r)}\}_{r\ge 0}$ are stationary points of $f_{\eps}(\beta)$.
\end{proof}

\section{Choice of $\mu$ in the SB algorithm}\label{sec:app_SB}

In the SB algorithm, the current solution $\widehat{\beta}^{(r+1)}$
is obtained from the solution of the following linear system
\begin{equation} \nonumber
\big( {\bf X}^T {\bf X} + \mu {\bf I} + \mu {\bf D}^T {\bf D} \big) \beta
= {\bf c}^{(r)}(\mu).
\end{equation}
The additional parameter $\mu$ affects the convergence rate
but there is no optimal rule for its choice in this problem \citep{Ghadimi2012}.
Thus, \cite{Ye2011} suggest a pre-trial procedure to choose  $\mu$ as one of $\{0.2, 0.4, 0.6, 0.8, 1 \} \times \|{\bf y}\|_2 $ by testing
computation times for given $(\lambda_1, \lambda_2)$.
In addition to this recipe, we consider $\frac{1}{n}\|{\bf y}\|_2$ and  $\frac{1}{n^2}\|{\bf y}\|_2$ as candidate values of $\mu$.
In our design of examples, 
we observe that the values other than $\frac{1}{n}\|{\bf y}\|_2$ makes the SB algorithm prematurely
stop before reaching the optimal solution.
The value $\frac{1}{n}\|{\bf y}\|_2$ also leads the best convergence rate from our pretrials in Figure \ref{SB_mu}.
Thus, we set $\mu = \frac{1}{n}\|{\bf y}\|_2$ in the numerical studies.

\begin{figure}[htb!] 
 \setcounter{figure}{0}
\renewcommand\thefigure{A.\arabic{figure}}
\begin{center}
  \begin{minipage}[b]{.48\linewidth}
  \centering   \centerline{\epsfig{file=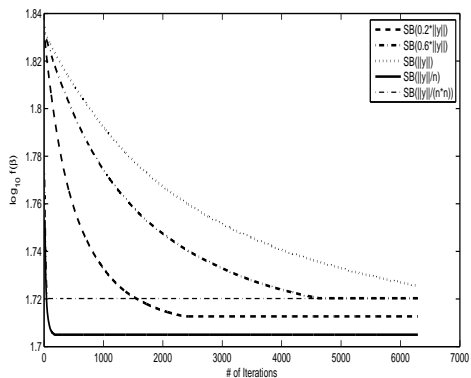,width=0.9\textwidth,height=0.25\textheight}}
  \centerline{(a) $n=200$, $p=1000$}
 \end{minipage}
\medskip
  \begin{minipage}[b]{.48\linewidth}
  \centering   \centerline{\epsfig{file=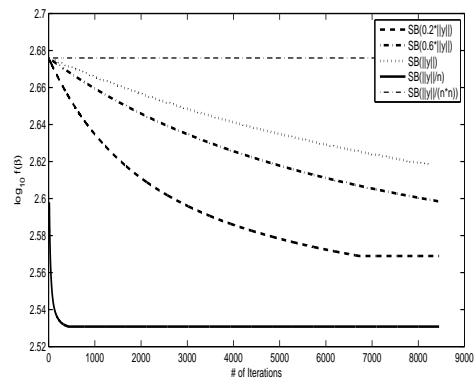,width=0.9\textwidth,height=0.25\textheight}}
  \centerline{(b) $n=1000$, $p=10000$}
\end{minipage}
\medskip

\caption{The number of iterations in SB algorithm as $\mu$ changes for $(\lambda_1,\lambda_2) = (0.1,0.1)$}
\label{SB_mu}
\end{center}
\end{figure}

\section{Tables of results for {\bf (C1)}--{\bf (C4)}}

\Blan
\newpage
\begin{table}[!htb]
\centering
\caption*{Table C.1: Summary of ${\bf (C1)}$ with
computation times and the numbers of iterations for $n=1000$.}
\label{tb_c1}
\medskip
{
\begin{tabular}{|c|c|c|c|c|c|c|c|c|c|c|c|c|c|} \hline
\multirow{2}{*}{$\lambda_1$} & \multirow{2}{*}{$\lambda_2$} & 
\multirow{2}{*}{$p$} & \multicolumn{5}{|c|}{Computation time (sec.)}
 & \multicolumn{5}{|c|}{Number of iterations}\\ \cline{4-13}
 & & & MMGPU & MM & EFLA & SPG & SB    &  MMGPU & MM & EFLA & SPG & SB\\ \hline
\multirow{4}{*}{0.1}  &  \multirow{4}{*}{0.1}  &  200  &  0.2858 &  0.0090 &  0.0093 &  0.0152 &  0.0089 &  3.0 &  3.0 &  27.8 &  42.5 &  3.0\\ 
  &   &  1000 &  2.6552&  1.5070&  0.3085&  0.5799&  2.1276&  22.6&  23.0&  258.9&  405.3&  15.1\\ 
  &   &  10000 &  7.2544&  66.9573&  32.8320&  102.5558&  99.1599&  84.1&  83.9&  1067.8&  1756.1&  158.2\\ 
  &   &  20000 &  13.3395&  133.2365&  77.4856&  281.4834&  209.6247&  90.3&  90.3&  1287.3&  2437.4&  206.6\\  \hline 
\multirow{4}{*}{0.1}  &  \multirow{4}{*}{1.0}  &  200  &  0.2207 &  0.0092 &  0.0067 &  0.0273 &  0.0656 &  4.0 &  4.0 &  24.3 &  73.8 &  28.7\\ 
  &   &  1000 &  2.2735&  1.1956&  0.1987&  0.4140&  10.0084&  26.3&  26.4&  166.9&  285.8&  75.3\\ 
  &   &  10000 &  5.3980&  51.1480&  19.3404&  58.0435&  117.2644&  86.3&  86.1&  618.7&  988.9&  186.9\\ 
  &   &  20000 &  9.8197&  101.1593&  49.4796&  163.0641&  238.4468&  94.2&  93.7&  808.7&  1411.3&  234.9\\  \hline 
\multirow{4}{*}{1.0}  &  \multirow{4}{*}{0.1}  &  200  &  0.1469 &  0.0074 &  0.0069 &  0.0168 &  0.0150 &  3.2 &  3.2 &  26.1 &  43.0 &  5.8\\ 
  &   &  1000 &  1.6121&  0.8684&  0.2128&  0.4359&  9.2411&  23.2&  23.0&  166.5&  266.7&  69.7\\ 
  &   &  10000 &  7.4584&  75.5762&  15.5698&  48.1091&  54.5518&  82.9&  82.9&  506.3&  822.1&  87.7\\ 
  &   &  20000 &  14.9428&  163.5872&  44.3698&  133.7532&  152.5818&  92.8&  92.8&  737.6&  1155.5&  152.4\\  \hline 
\multirow{4}{*}{1.0}  &  \multirow{4}{*}{1.0}  &  200  &  0.1477 &  0.0072 &  0.0068 &  0.0157 &  0.0873 &  4.0 &  4.0 &  25.0 &  42.2 &  37.8\\ 
  &   &  1000 &  1.2295&  0.6092&  0.1527&  0.2986&  10.1666&  24.4&  24.4&  128.1&  209.6&  79.2\\ 
  &   &  10000 &  4.5901&  44.8417&  14.6557&  32.6244&  67.7480&  82.5&  82.4&  465.8&  554.9&  110.9\\ 
  &   &  20000 &  8.4576&  90.2439&  35.3691&  96.1971&  199.6160&  88.6&  88.5&  575.4&  830.1&  201.5\\  \hline 

\end{tabular}
}
\end{table}

\newpage
\begin{table}[!htb]
\centering
\caption*{Table C.2: Summary of ${\bf (C2)}$ with
computation times and the numbers of iterations for $n=1000$.
The gray cells denote that
the EFLA fails to reach
the optimal solution for one or two samples.
We report the average computation times of the EFLA removed the failed cases. }
\label{tb_c2}
\medskip
{
\begin{tabular}{|c|c|c|c|c|c|c|c|c|c|c|c|c|c|} \hline
\multirow{2}{*}{$\lambda_1$} & \multirow{2}{*}{$\lambda_2$} & 
\multirow{2}{*}{$p$} & \multicolumn{5}{|c|}{Computation time (sec.)}
 & \multicolumn{5}{|c|}{Number of iterations}\\ \cline{4-13}
 & & & MMGPU & MM & EFLA & SPG & SB    &  MMGPU & MM & EFLA & SPG & SB\\ \hline
\multirow{4}{*}{0.1}  &  \multirow{4}{*}{0.1}  &  200  &  0.2247 &  0.0108 &  0.0077 &  0.0170 &  0.0088 &  3.0 &  3.0 &  25.0 &  46.8 &  3.0\\ 
  &   &  1000 &  3.0984&  1.4673&  0.3732&  0.7777&  1.6792&  22.2&  21.8&  311.8&  508.8&  19.9\\ 
  &   &  10000 &  4.1122&  45.8289&  \cellcolor{gray!20}91.7927&  263.7253&  266.5620&  36.3&  35.5&  \cellcolor{gray!20}2976.3&  4527.5&  466.7\\ 
  &   &  20000 &  4.4272&  52.7964&   \cellcolor{gray!20}242.0749&  821.5516&  704.5284&  27.2&  27.6&  \cellcolor{gray!20} 4043.3&  7108.9&  776.4\\  \hline 
\multirow{4}{*}{0.1}  &  \multirow{4}{*}{1.0}  &  200  &  0.2276 &  0.0115 &  0.0072 &  0.0168 &  0.0595 &  4.2 &  4.2 &  25.4 &  41.9 &  26.8\\ 
  &   &  1000 &  3.6282&  1.5755&  0.2211&  0.5595&  4.7977&  28.1&  28.2&  180.5&  359.1&  57.7\\ 
  &   &  10000 &  10.3537&  109.7741&  47.8311&  132.3191&  107.4803&  53.9&  53.9&  1526.8&  2272.4&  188.4\\ 
  &   &  20000 &  23.6943&  271.3948&  236.7397&  503.1258&  232.9924&  63.9&  63.7&  3915.2&  4347.2&  256.3\\  \hline 
\multirow{4}{*}{1.0}  &  \multirow{4}{*}{0.1}  &  200  &  0.1279 &  0.0086 &  0.0061 &  0.0171 &  0.0081 &  3.1 &  3.1 &  21.8 &  44.6 &  2.7\\ 
  &   &  1000 &  1.7817&  0.9148&  0.1901&  0.5022&  2.9606&  20.3&  20.3&  160.7&  325.7&  34.9\\ 
  &   &  10000 &  10.8871&  127.3424&  51.5586&  169.2509&  142.5682&  91.6&  91.5&  1665.9&  2910.1&  249.8\\ 
  &   &  20000 &  20.1332&  244.0547&  141.8118&  526.4848&  313.0687&  88.7&  89.7&  2383.4&  4556.9&  344.7\\  \hline 
\multirow{4}{*}{1.0}  &  \multirow{4}{*}{1.0}  &  200  &  0.1467 &  0.0084 &  0.0065 &  0.0176 &  0.0633 &  4.2 &  4.2 &  23.4 &  44.6 &  28.8\\ 
  &   &  1000 &  1.7161&  0.8342&  0.1476&  0.4186&  3.0164&  22.5&  22.5&  122.6&  269.4&  36.0\\ 
  &   &  10000 &  7.7781&  91.8040&  38.7895&  120.4112&  64.9480&  92.1&  92.2&  1226.0&  2077.8&  114.0\\ 
  &   &  20000 &  17.6899&  212.5753&  106.2446&  339.2643&  174.6900&  102.3&  102.3&  1756.6&  2929.0&  192.6\\  \hline 
\end{tabular}
}
\end{table}

\newpage
\begin{table}[!htb]
\centering
\caption*{Table C.3: Summary of ${\bf (C3)}$ with
computation times and the numbers of iterations for $n=1000$.
The gray cells denote that
the EFLA fails to reach
the optimal solution for one sample.
We report the average computation times of the EFLA removed failed case.}
\label{tb_c3}
\medskip
{
\begin{tabular}{|c|c|c|c|c|c|c|c|c|c|c|c|c|c|} \hline
\multirow{2}{*}{$\lambda_1$} & \multirow{2}{*}{$\lambda_2$} & 
\multirow{2}{*}{$p$} & \multicolumn{5}{|c|}{Computation time (sec.)}
 & \multicolumn{5}{|c|}{Number of iterations}\\ \cline{4-13}
 & & & MMGPU & MM & EFLA & SPG & SB    &  MMGPU & MM & EFLA & SPG & SB\\ \hline
\multirow{4}{*}{0.1}  &  \multirow{4}{*}{0.1}  &  200  &  0.2511 &  0.0171 &  0.0081 &  0.0170 &  0.0095 &  3.0 &  3.0 &  26.7 &  45.2 &  3.0\\ 
  &   &  1000 &  2.9360&  1.5933&  0.3793&  0.8082&  1.2007&  21.3&  21.9&  327.4&  511.3&  13.2\\ 
  &   &  10000 &  3.7909&  42.9165&  90.8098&  282.7641&  286.2912&  34.7&  34.0&  3025.2&  4789.1&  498.3\\ 
  &   &  20000 &  4.0806&  48.1163&  248.1919&  871.9119&  756.4445&  26.0&  25.8&  4108.3&  7553.7&  824.2\\  \hline 
\multirow{4}{*}{0.1}  &  \multirow{4}{*}{1.0}  &  200  &  0.3401 &  0.0133 &  0.0069 &  0.0174 &  0.0662 &  4.2 &  4.2 &  25.3 &  44.9 &  27.4\\ 
  &   &  1000 &  3.9890&  1.8233&  0.2031&  0.5462&  3.8157&  27.4&  27.3&  170.8&  352.3&  43.5\\ 
  &   &  10000 &  15.0492&  180.0133&  54.7479&  144.0007&  108.6189&  91.8&  91.9&  1795.4&  2446.6&  189.3\\ 
  &   &  20000 &  28.2846&  343.3138&  169.7143&  505.0450&  238.1494&  96.8&  97.0&  2762.1&  4376.6&  259.7\\  \hline 
\multirow{4}{*}{1.0}  &  \multirow{4}{*}{0.1}  &  200  &  0.0890 &  0.0058 &  0.0072 &  0.0178 &  0.0130 &  3.4 &  3.4 &  26.5 &  46.3 &  5.0\\ 
  &   &  1000 &  1.6862&  0.8142&  0.3392&  0.8342&  2.3927&  23.8&  23.8&  287.1&  555.2&  24.8\\ 
  &   &  10000 &  10.7328&  127.3884&  50.5020&  176.1017&  127.5324&  88.6&  88.6&  1681.8&  2985.3&  222.3\\ 
  &   &  20000 &  19.5732&  239.8993&  150.7861&  520.9941&  307.9786&  85.4&  85.8&  2497.4&  4511.6&  335.6\\  \hline 
\multirow{4}{*}{1.0}  &  \multirow{4}{*}{1.0}  &  200  &  0.1840 &  0.0085 &  0.0067 &  0.0169 &  0.0094 &  4.4 &  4.4 &  24.9 &  44.3 &  3.4\\ 
  &   &  1000 &  1.4819&  0.7155&  0.1548&  0.4869&  1.9491&  19.7&  19.6&  124.3&  308.7&  23.0\\ 
  &   &  10000 &  8.1713&  97.7493& \cellcolor{gray!20} 38.0236&  129.9449&  89.1512&  95.4&  95.5&   \cellcolor{gray!20} 1240.0&  2206.9&  155.5\\ 
  &   &  20000 &  17.7725&  219.5605&  110.7899&  365.3513&  173.9217&  101.9&  101.9&  1799.7&  3156.2&  189.4\\  \hline 

\end{tabular}
}
\end{table}

\newpage
\begin{table}[!htb]
\centering
\caption*{Table C.4: Summary of ${\bf (C4)}$ with
computation times and the numbers of iterations for $n=1000$.}
\label{tb_c4}
\medskip
{
\begin{tabular}{|c|c|c|c|c|c|c|c|c|c|} \hline
\multirow{2}{*}{$\lambda_1$} & \multirow{2}{*}{$\lambda_2$} & 
\multirow{2}{*}{$q \times q$} & \multicolumn{3}{|c|}{Computation time (sec.)}
 &\multicolumn{3}{|c|}{Number of iterations}\\ \cline{4-9}
 & & &  MM & SPG & SB  &  MM & SPG & SB  \\ \hline
\multirow{4}{*}{0.1}  &  \multirow{4}{*}{0.1}  &  $16 	\times 16$  &  0.0477 &  0.0252 &  0.0230 &  5.0 &  54.9 &  2.8\\ 
  &   &  $32 \times 32$ &  1.9948&  1.1403&  1.9887&  21.3&  603.6&  16.5\\ 
  &   &  $64 \times 64$ &  45.1773&  68.6246&  164.9881&  63.7&  2745.9&  456.0\\ 
  &   &  $128 \times 128$ &  50.2078&  615.4028&  979.8359&  28.8&  6414.6&  1066.6\\  \hline 
\multirow{4}{*}{0.1}  &  \multirow{4}{*}{1.0}  &  $16 	\times 16$  &  0.0549 &  0.0234 &  0.0971 &  8.3 &  53.3 &  26.7\\ 
  &   &  $32 \times 32$ &  1.1902&  0.5846&  4.1854&  23.0&  294.8&  37.9\\ 
  &   &  $64 \times 64$ &  15.1176&  23.0897&  33.2731&  39.6&  923.6&  90.8\\ 
  &   &  $128 \times 128$ &  72.8389&  208.2940&  119.7563&  49.8&  2169.3&  123.5\\  \hline 
\multirow{4}{*}{1.0}  &  \multirow{4}{*}{0.1}  &  $16 	\times 16$  &  0.0612 &  0.0232 &  0.0180 &  7.3 &  56.2 &  3.0\\ 
  &   &  $32 \times 32$ &  1.4630&  0.6096&  2.3580&  20.6&  326.3&  21.5\\ 
  &   &  $64 \times 64$ &  48.5678&  46.7797&  92.2762&  80.1&  1876.4&  254.2\\ 
  &   &  $128 \times 128$ &  241.6816&  343.2788&  657.3494&  85.3&  3579.7&  714.9\\  \hline 
\multirow{4}{*}{1.0}  &  \multirow{4}{*}{1.0}  &  $16 	\times 16$  &  0.0476 &  0.0221 &  0.1006 &  9.5 &  51.5 &  27.3\\ 
  &   &  $32 \times 32$ &  0.8718&  0.4960&  3.3169&  20.5&  273.1&  30.3\\ 
  &   &  $64 \times 64$ &  30.0718&  24.6371&  35.1351&  63.5&  978.6&  95.9\\ 
  &   &  $128 \times 128$ &  189.6904&  237.8033&  233.3495&  109.0&  2476.8&  248.1\\  \hline 

\end{tabular}
}
\end{table}

\Elan

\end{document}